\documentclass[11pt,]{article}
\usepackage{lmodern}
\usepackage{setspace}
\usepackage{amssymb,amsmath}
\usepackage{amsthm}
\theoremstyle{plain}

\newtheorem{Result}{Result}[section]
\newtheorem{Fact}{Fact}[section]
\newtheorem{Corollary}{Corollary}[section]

\theoremstyle{definition}
\newtheorem{Definition}{Definition}[section]

\theoremstyle{remark}

\usepackage{ifxetex,ifluatex}
\usepackage{fixltx2e} % provides \textsubscript
\ifnum 0\ifxetex 1\fi\ifluatex 1\fi=0 % if pdftex
  \usepackage[T1]{fontenc}
  \usepackage[utf8]{inputenc}
\else % if luatex or xelatex
  \ifxetex
    \usepackage{mathspec}
  \else
    \usepackage{fontspec}
  \fi
  \defaultfontfeatures{Ligatures=TeX,Scale=MatchLowercase}
\fi
\IfFileExists{upquote.sty}{\usepackage{upquote}}{}
\IfFileExists{microtype.sty}{%
\usepackage{microtype}
\UseMicrotypeSet[protrusion]{basicmath} % disable protrusion for tt fonts
}{}
\usepackage[margin=1in]{geometry}
\usepackage{hyperref}
\hypersetup{unicode=true,
            pdftitle={Estimating adult death rates from sibling histories: a network approach},
            pdfborder={0 0 0},
            breaklinks=true}
\usepackage{graphicx,grffile}
\makeatletter
\def\maxwidth{\ifdim\Gin@nat@width>\linewidth\linewidth\else\Gin@nat@width\fi}
\def\maxheight{\ifdim\Gin@nat@height>\textheight\textheight\else\Gin@nat@height\fi}
\makeatother
\setkeys{Gin}{width=\maxwidth,height=\maxheight,keepaspectratio}
\IfFileExists{parskip.sty}{%
\usepackage{parskip}
}{% else
\setlength{\parindent}{0pt}
\setlength{\parskip}{6pt plus 2pt minus 1pt}
}
\providecommand{\tightlist}{%
  \setlength{\itemsep}{0pt}\setlength{\parskip}{0pt}}
\ifx\paragraph\undefined\else
\let\oldparagraph\paragraph
\renewcommand{\paragraph}[1]{\oldparagraph{#1}\mbox{}}
\fi
\ifx\subparagraph\undefined\else
\let\oldsubparagraph\subparagraph
\renewcommand{\subparagraph}[1]{\oldsubparagraph{#1}\mbox{}}
\fi

\usepackage{rotating}
\usepackage{sidecap}
\usepackage{bbm}
\usepackage{cases}
\usepackage{amsthm}
\usepackage{longtable}
\usepackage{makecell}
\usepackage{booktabs}

\makeatletter
\@ifpackageloaded{subfig}{}{\usepackage{subfig}}
\@ifpackageloaded{caption}{}{\usepackage{caption}}
\AtBeginDocument{%

}
\AtBeginDocument{%

}
\@ifpackageloaded{float}{}{\usepackage{float}}
\floatstyle{ruled}
\@ifundefined{c@chapter}{\newfloat{codelisting}{h}{lop}}{\newfloat{codelisting}{h}{lop}[chapter]}
\floatname{codelisting}{Listing}

\makeatother

\title{Estimating adult death rates from sibling histories:\\
a network approach\footnote{For helpful feedback on earlier versions of
  the manuscript, the authors would like to thank the participants in
  the 2018 Formal Demography Workshop at UC Berkeley; participants in
  the 2018 PAA Session ``Social capital and older adults in developing
  countries''; and Stephane Helleringer.}}
\author{Dennis M. Feehan\footnote{UC Berkeley, feehan@berkeley.edu}~ and Gabriel
M. Borges\footnote{IBGE, gmendesb@hotmail.com}}
\date{June 27, 2019}

\begin{document}
\maketitle
\begin{abstract}
Hundreds of millions of people live in countries that do not have
complete death registration systems, meaning that most deaths are not
recorded and critical quantities like life expectancy cannot be directly
measured. The sibling survival method is a leading approach to
estimating adult mortality in the absence of death registration. The
idea is to ask a survey respondent to enumerate her siblings and to
report about their survival status. In many countries and time periods,
sibling survival data are the only nationally-representative source of
information about adult mortality. Although a huge amount of sibling
survival data has been collected, important methodological questions
about the method remain unresolved. To help make progress on this issue,
we propose re-framing the sibling survival method as a network sampling
problem. This approach enables us to formally derive statistical
estimators for sibling survival data. Our derivation clarifies the
precise conditions that sibling history estimates rely upon; it leads to
internal consistency checks that can help assess data and reporting
quality; and it reveals important quantities that could potentially be
measured to relax assumptions in the future. We introduce the \texttt{R}
package \texttt{siblingsurvival}, which implements the methods we
describe.
\end{abstract}

~

\newpage

\hypertarget{toc}{}

\newpage

\hypertarget{sec:intro}{%
\section{Introduction}\label{sec:intro}}

Death rates at adult ages are a core component of population health and
a central topic of study for demography. Unfortunately, most of the
world's poorest countries are victims of the \emph{scandal of
invisibility}: they do not have complete death registration systems,
meaning that most people die without ever having their existence
officially recorded (Setel et al. 2007; AbouZahr et al. 2015). This lack
of complete death registration means that critical quantities like life
expectancy cannot be directly measured. Improving death registration
systems is the long-term solution to the scandal of invisibility, but
progress has been very slow (Mikkelsen et al. 2015). Until complete
death registration systems are available everywhere, sample-based
approaches to adult mortality estimation will continue to play a
critical role in understanding population health and wellbeing.

The leading approach to collecting information about adult mortality in
the absence of death registration is the sibling survival method
(Rutenberg and Sullivan 1991; Brass 1975). The idea is to ask survey
respondents to report the number of siblings they have, and to then ask
for each sibling's gender, date of birth and date of death (where
appropriate). This data collection strategy produces \emph{sibling
histories} which contain information about the survival status of all of
the members of the respondent's sibship.

Since high-quality household surveys are routinely conducted in most
countries---including countries that lack death registration
systems---the sibling survival method offers the opportunity to try to
estimate adult death rates in many places that have no other
nationally-representative adult mortality data. Over the past two
decades, a huge amount of sibling history data has been collected; for
example, as a part of the DHS program alone, sibling histories have been
collected in more than 150 surveys from dozens of countries around the
world (Corsi et al. 2012; Fabic, Choi, and Bird 2012).

However, understanding how to analyze sibling histories has proven to be
very challenging. Researchers have long been aware that the method
suffers from many possible sources of bias (Gakidou and King 2006;
Graham, Brass, and Snow 1989; Masquelier 2013; Reniers, Masquelier, and
Gerland 2011; Trussell and Rodriguez 1990). Previous studies have
concluded that sibling history estimates can be problematic if (i) there
are sibships with no surviving members who could potentially be sampled
and interviewed in the survey; (ii) more generally, there is a
relationship between sibship size and mortality (\emph{e.g.} larger
sibships face higher death rates); and (iii) respondents' reports about
their siblings are inaccurate (\emph{e.g.} respondents omit siblings or
misreport a sibling's survival status). There has also been confusion
about whether the survival status of the respondent herself should be
included in the calculations, since respondents are always alive
(Masquelier 2013; Reniers, Masquelier, and Gerland 2011).

Researchers have worked on addressing these concerns about the sibling
survival method in three main ways: they have collected empirical
information about possible sources of bias in sibling reports (e.g.,
Helleringer, Pison, Kanté, et al. 2014); they have used microsimulation
to illustrate how large certain sources of bias can be under different
scenarios (Masquelier 2013); and they have used regression models to
pool information from different countries and time periods (Gakidou and
King 2006; Timaeus and Jasseh 2004; Obermeyer et al. 2010). Together,
these studies have produced many important insights about the sibling
survival method. However, these insights have not yet brought about a
consensus on how sibling histories should be analyzed. Currently, there
is partial evidence about many individual sources of possible bias, but
there is no way to integrate all of this evidence together. Thus, even
if we knew the exact size and direction of all the different sources of
possible error, we still would not understand how the errors would
combine to affect estimated death rates. More generally, little has been
proven about the precise conditions under which sibling survival
estimates can be expected to have attractive statistical properties such
as consistency or unbiasedness.

In this study, our goal is to help resolve some of the methodological
uncertainty about sibling survival. Our analysis is based on the insight
that the sibling relation induces a particular type of social network
among the members of a population. In this network, two people are
connected to one another if they are siblings; thus, estimating death
rates from sibling histories can be understood as a problem in network
sampling. Starting from the principles of network reporting, we describe
how to mathematically derive a sibling survival estimator. Deriving an
estimator from first principles in this way enables us to (i) clarify
the precise assumptions that the estimator requires in order to be
consistent, unbiased, and efficient; (ii) describe how violations of any
and all assumptions can combine to affect estimated death rates; (iii)
identify quantities that could potentially be measured in the future to
relax assumptions; and (iv) develop internal consistency checks that can
be used to assess data and reporting quality in a given sample.

\hypertarget{sec:deriving}{%
\section{Setup}\label{sec:deriving}}

Figure~\ref{fig:sib-nr-illustration} illustrates how we understand
sibling histories as a network reporting problem. The left-hand panel
shows a small population whose members are connected if they are
siblings (\emph{i.e.}, two nodes are connected if they have the same
mother\footnote{Respondents are typically asked to consider `siblings'
  to be all children born to their mother.}). Clear nodes are alive and
grey nodes are dead at the time of the survey. Since the sibling
relation is transitive, the network is entirely composed of fully
connected components, or cliques; each of these cliques is one sibship.
The middle panel shows one specific sibship, and the right-hand panel
illustrates the bipartite reporting network that is generated when all
of the surviving members of that sibship are asked to report about their
siblings\footnote{The dead person (in grey) cannot be interviewed, and
  so is not shown on the left-hand side of the bipartite reporting
  network.}. In the bipartite reporting network, each directed edge
represents one sibling reporting about another---so, for example, the
edge \(13 \rightarrow 12\) indicates that node 13 reports about node 12.
Feehan and Salganik (2016a) describes bipartite reporting networks in
greater detail.

\begin{figure}
\centering

\subfloat[]{\includegraphics[width=\textwidth,height=0.25\textheight]{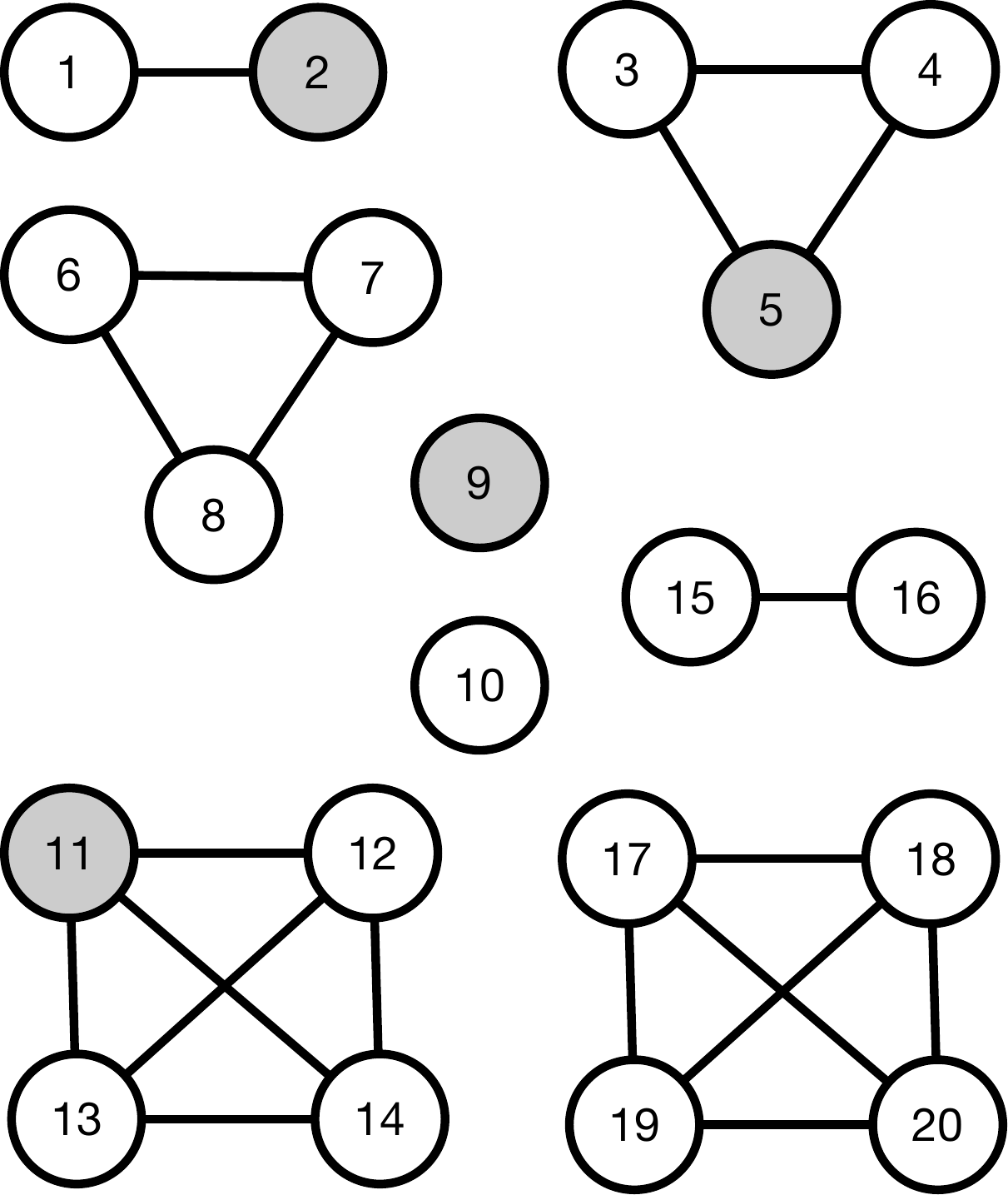}\label{fig:entire-net}}\hfill
\subfloat[]{\includegraphics[width=\textwidth,height=0.25\textheight]{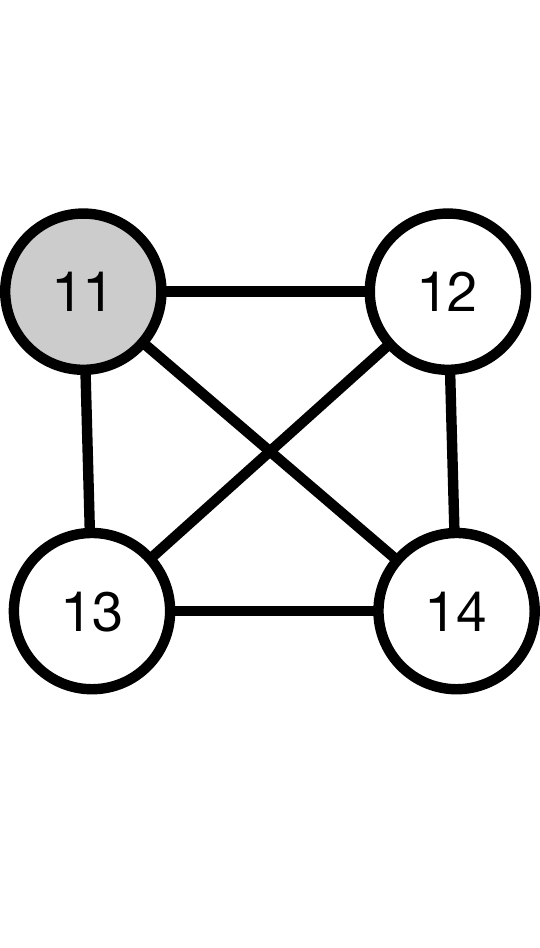}\label{fig:sibship-net}}\hfill
\subfloat[]{\includegraphics[width=\textwidth,height=0.25\textheight]{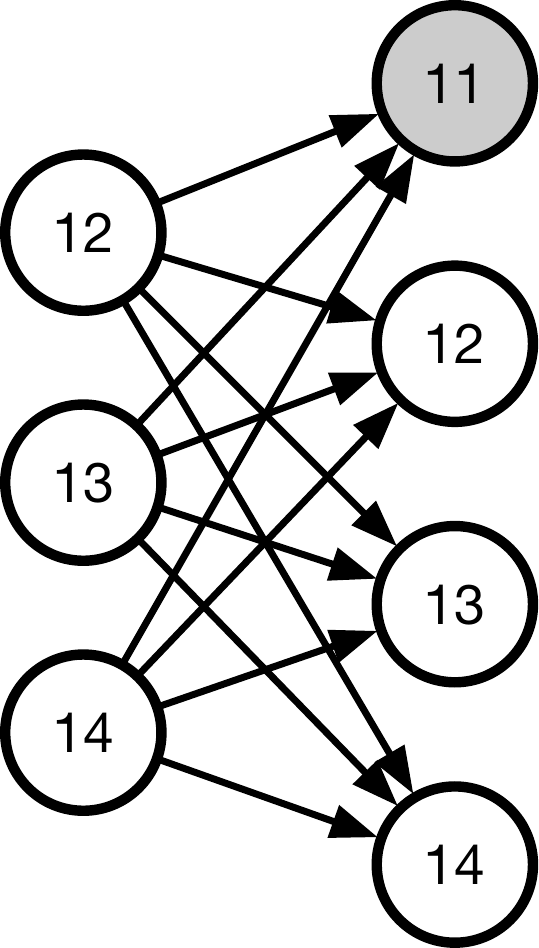}\label{fig:sibship-reporting}}

\caption{Framing sibling survival as a network reporting estimator. (a)
A population connected through a sibship network. Because the sibling
relation is transitive, the network is composed entirely of cliques. (b)
The network for a single sibship. (c) Bipartite reporting graph for a
single sibship.}

\label{fig:sib-nr-illustration}

\end{figure}

Our quantity of interest is \(M_\alpha\), the death rate for a specific
group \(\alpha\) (for example, \(\alpha\) might be all women aged 30-34
in 2018). \(M_\alpha\) is defined as \begin{equation}
M_\alpha = \frac{D_\alpha}{N_\alpha},
\label{eq:asdr-defn}\end{equation}

where \(D_\alpha\) is the number of deaths in group \(\alpha\) and
\(N_\alpha\) is the person-years of exposure among members of group
\(\alpha\). We can develop an estimator for \(M_\alpha\) by separately
estimating the numerator and the denominator of
Equation~\ref{eq:asdr-defn}; thus, the challenge is to derive sibling
history-based estimators for \(D_\alpha\) and \(N_\alpha\).

Figure~\ref{fig:sibship-reporting} illustrates the fact that each
sibling can potentially be reported as many times as she has living
sibship members who are eligible to respond to the survey (Sirken 1970;
Gakidou and King 2006; Masquelier 2013). Inferences from sibling reports
must somehow account for this fact. Our approach is to distinguish
between two groups of people: the first group is people who have no
siblings who are eligible to respond to the survey. These people will
never appear in the sibling history data -- they are \emph{invisible} to
the sibling histories. The second group is \emph{visible} people who do
have siblings eligible to respond to the survey.

Formally, let \(U\) be the population being studied, and let
\(F \subset U\) be the \emph{frame population}, which is the set of all
people who are eligible to respond to the survey. We define person
\(j \in U\)'s \emph{visibility}, \(v(j,F)\), to be the number of living
siblings who would report person \(j\) in a census of \(F\)\footnote{The
  idea behind the notation \(v(\cdot, \cdot)\) is that the first
  argument is whoever is being reported about, and the second argument
  is the set of people who make reports; so, \(v(j, F)\) is the number
  of times the person \(j\) is reported about by members of the frame
  population \(F\). When we add a bar, we mean the average taken with
  respect to the first argument - so \(\bar{v}(A, F)\) is
  \(v(A,F)/|A|\), the average number of times a member of \(A\) is
  reported about by \(F\).}. Everyone in the population is either
visible (\(v(j,F) > 0\)) or invisible (\(v(j,F)=0\)). Thus, we can write
the number of deaths in group \(\alpha\) as: \begin{equation}
\begin{aligned}
D_\alpha &=
\underbrace{
  \sum_{\substack{ j \in D_\alpha \\ v(j,F) > 0}} 1}_{\text{
  visible deaths
}} +
\underbrace{
  \sum_{\substack{j \in D_\alpha \\ v(j,F)=0}} 1}_{\text{
  invisible deaths
}}
= D_\alpha^V + D_\alpha^I,
\end{aligned}
\label{eq:deaths-vis-invis}\end{equation}

where \(D_\alpha^V\) is the number of \emph{visible deaths} that could
be learned about using sibling reports and \(D_\alpha^I\) is the number
of \emph{invisible deaths} that cannot be learned about using sibling
reports. We can define analogous quantities for the denominator
\(N_\alpha = N_\alpha^V + N_\alpha^I\), where the \(N_\alpha^V\) is the
\emph{visible exposure} and \(N_\alpha^I\) is the \emph{invisible
exposure}. Finally, we define
\(M_\alpha^I = \frac{D_\alpha^I}{N_\alpha^I}\) to be the \emph{invisible
death rate}, \(M_\alpha^V = \frac{D_\alpha^V}{N_\alpha^V}\) to be the
\emph{visible death rate}, and
\(M_\alpha = \frac{D_\alpha^I + D_\alpha^V}{N_\alpha^I + N_\alpha^V}\)
to be the \emph{total death rate}.

In Section~\ref{sec:adjusting-visibility}, we show how sibling history
data can be used to develop death rate estimators for the visible
population. We address possible differences between the visible and the
invisible populations as part of a more general sensitivity framework,
introduced in Section~\ref{sec:sensitivity}. Our sensitivity framework
consists of mathematical expressions that describe how sensitive death
rate estimates are to all of the conditions that the estimators rely
upon, including differences between the visible and invisible
populations; reporting errors; and structural variation in sibship
networks. Section~\ref{sec:example} contains an empirical illustration
of our technical results using the 2000 Malawi Demographic and Health
Survey. As part of our empirical demonstration, we discuss variance
estimation, and we introduce empirical checks that researchers can
perform to assess some of the conditions that sibling estimators rely
upon. Section~\ref{sec:recommendations} compares the estimators we
introduce, and discusses the implication of our results for practice.
Finally, Section~\ref{sec:conclusion} concludes and outlines directions
for future work.

\hypertarget{sec:adjusting-visibility}{%
\section{Adjusting for visibility in sibling
reports}\label{sec:adjusting-visibility}}

The visible death rate can be estimated from sibling histories using an
expression of the form \begin{equation}
\widehat{M}^V_\alpha = \frac{\widehat{D}^V_\alpha}{\widehat{N}^V_\alpha},
\label{eq:visdr-form}\end{equation}

where \(\widehat{D}^V_\alpha\) is an estimator for the number of visible
deaths in group \(\alpha\) and \(\widehat{N}^V_\alpha\) is an estimator
for the amount of visible exposure in group \(\alpha\). In order to
estimate these two quantities, we face the challenge that even people
who are visible to the sibling histories may still differ in the extent
to which they are visible; for example, visible people from larger
sibships may tend to have different death rates than visible people from
smaller sibships. We address this challenge by introducing statistical
estimators that adjust for how visible reported siblings are.

We consider two different approaches to adjusting for differential
visibility: \emph{aggregate visibility} estimation and \emph{individual
visibility} estimation. These two approaches lead to two different
estimators for the visible death rate. In both cases, we start by
describing how to derive population-level relationships, and then we use
these population relationships as the basis for sample-based estimators.

\hypertarget{the-aggregate-visibility-approach}{%
\subsection*{The aggregate visibility
approach}\label{the-aggregate-visibility-approach}}
\addcontentsline{toc}{subsection}{The aggregate visibility approach}

The \emph{aggregate visibility} approach is based on the idea that
reports about siblings can first be aggregated, and then the aggregated
reports can be adjusted to account for visibility (Bernard et al. 1989;
Rutstein and Guillermo Rojas 2006; Feehan and Salganik 2016a). To
illustrate this approach, we first focus on reports about visible deaths
among siblings, \(D^V_\alpha\). Throughout the main paper, we assume
there are no \emph{false positive} reports -- \emph{i.e.}, we assume
that respondents' reports may omit siblings, but that they never
mistakenly include someone who is not truly a sibling. Of course, this
could in fact happen -- but this assumption makes the exposition much
cleaner, and the results derived in the Appendixes consider reporting
with false positives.

Let \(y(F,D^V_\alpha)\) be the total number of deaths that would be
reported among respondents' siblings in a census of the frame
population\footnote{The idea behind the notation \(y(\cdot, \cdot)\) is
  that the first argument is the set of people reporting, and the second
  argument is the set of people who are being reported about; so,
  \(y(F, D^V_\alpha)\) is the total number of deaths in \(D^V_\alpha\)
  reported by people in the frame population \(F\).}. Appendix
\ref{sec:agg-vis} shows that if there are no false positive reports,
then the total number of reports about sibling deaths, divided by the
average visibility of visible deaths, will be equal to the number of
visible deaths\footnote{To avoid over-complicating notation, we use
  \(D^V_\alpha\) to mean both the number of visible deaths in group
  \(\alpha\), and the set of visible deaths in groups \(\alpha\).}:

\begin{equation}
D^V_{\alpha} = \frac{y(F, D^V_\alpha)}{\bar{v}(D^V_\alpha, F)}.
\label{eq:agg-step1}\end{equation}

The idea is that in a census, the average death will be reported
\(\bar{v}(D^V_\alpha, F)\) times; thus, the total number of reports
about deaths, \(y(F, D^V_\alpha)\), can be divided by
\(\bar{v}(D^V_\alpha, F)\) to recover the total number of visible
deaths. Appendix \ref{sec:agg-vis} shows that a similar analysis can be
applied to the denominator of the visible death rate, yielding \[
N^V_{\alpha} = \frac{y(F, N^V_\alpha)}{\bar{v}(N^V_\alpha, F)}.
\]

Thus, in a census,

\begin{equation}
M^V_\alpha = \frac{D^V_{\alpha}}{N^V_{\alpha}} = \frac{y(F, D^V_\alpha)}{y(F, N^V_\alpha)} 
\cdot \frac{\bar{v}(N_\alpha^V, F)}{\bar{v}(D_\alpha^V,F)}.
\label{eq:agg-step3}\end{equation}

Under the assumption that the visibility of deaths is the same as the
visibility of exposure,
\(\bar{v}(D^V_\alpha,F) = \bar{v}(N^V_\alpha,F)\), and
Equation~\ref{eq:agg-step3} simplifies to the \emph{aggregate visibility
estimand}

\begin{equation}
M^V_{\alpha,\text{agg}} = \frac{y(F, D^V_\alpha)}{y(F, N^V_\alpha)}.
\label{eq:agg-step4}\end{equation}

Finally, the population-level relationship in
Equation~\ref{eq:agg-step4} motivates the sample-based estimator
(Appendix \ref{sec:agg-vis-deathrate}):

\begin{equation}
\begin{aligned}
\widehat{M}^{V}_{\alpha,\text{agg}}
&= 
\frac{\sum_{i \in s} w_i~y(i, D^V_\alpha)}{\sum_{i \in s} w_i~y(i, N^V_\alpha)}.
\end{aligned}
\label{eq:mhat-agg}\end{equation}

Since this approach is based on adding together reports about all of the
siblings, and then adjusting for the visibility of these aggregate
reports, we call \(\widehat{M}^V_{\alpha,\text{agg}}\) an
\emph{aggregate visibility} estimator (Feehan 2015; Feehan, Mahy, and
Salganik 2017; Bernard et al. 1989). Appendix
\ref{sec:agg-vis-deathrate} (Result \ref{res:mvis-agg-estimator})
formally derives the estimator; the derivation reveals that this
approach can be expected to produce essentially unbiased estimates as
long as reports about siblings are accurate, and as long as there is no
relationship between sibship visibility and mortality (\emph{i.e.}, as
long as \(\bar{v}(D^V_\alpha,F) = \bar{v}(N^V_\alpha,F)\)).

\hypertarget{the-individual-visibility-approach}{%
\subsection*{The individual visibility
approach}\label{the-individual-visibility-approach}}
\addcontentsline{toc}{subsection}{The individual visibility approach}

The \emph{individual visibility} approach is based on the idea that
reports about siblings can first be adjusted for visibility and then the
adjusted reports can be aggregated (Sirken 1970; Gakidou and King 2006;
Lavallee 2007; Feehan 2015). To illustrate this approach, consider
reports about a specific deceased sibling \(j \in D_\alpha\) that are
made in a census of the frame population. Let \(y(F,j)\) be the number
of times that people in the frame population \(F\) report the deceased
sibling \(j\). This quantity \(y(F,j)\) will be equal to the visibility
of \(j\) to \(F\), \(v(j,F)\), as long as there are no false positive
reports. Thus, for every visible sibling \(j\), we have \begin{equation}
\frac{y(F,j)}{v(j,F)} 
= 1.
\label{eq:ind-step1}\end{equation}

Summing over all visible deaths, we obtain

\begin{equation}
D_\alpha^V 
= \sum_{j \in D_\alpha^V} 1
= \sum_{j \in D_\alpha^V} \frac{y(F,j)}{v(j,F)}
= \sum_{i \in F} \sum_{j \in D_\alpha^V} \frac{y(i, j)}{v(j,F)},
\label{eq:ind-step2}\end{equation}

where the last step follows because \(y(F,j) = \sum_{i \in F} y(i,j)\).
Equation~\ref{eq:ind-step2} expresses the number of visible deaths in
terms of each survey respondent's reported connections to each death
(\(y(i,j)\)) and the visibility of each death (\(v(j,F)\)). In Appendix
\ref{sec:ind-vis}, we show that when reporting is accurate, the
visibility of any dead sibling \(j \in D_\alpha^V\), can be written as
\(v(j,F) = y(i,F) + 1\) for any survey respondent \(i\) who is in the
same sibship as \(j\). Using this relationship,
Equation~\ref{eq:ind-step2} can be re-written as

\begin{equation}
D_\alpha^V
= \sum_{i \in F} \sum_{j \in D_\alpha^V} \frac{y(i, j)}{v(j,F)}
= \sum_{i \in F} \sum_{j \in D_\alpha^V} \frac{y(i,j)}{y(i,F) + 1}
= \sum_{i \in F} \frac{y(i, D_\alpha)}{y(i,F) + 1},
\label{eq:ind-D-estimand}\end{equation}

where \(y(i,F)\) is \(i\)'s reported number of siblings on the sampling
frame. Equation~\ref{eq:ind-D-estimand} relates the population-level
number of visible deaths \(D^V_\alpha\) to survey respondents' reports
about deaths in their sibships, \(y(i,D_\alpha)\), and survey
respondents' reports about the number of frame population members in
their sibships, \(y(i,F)\).

A parallel argument, found in Appendix \ref{sec:ind-vis}, reveals that
when reporting is accurate, \(N^V_{\alpha}\) can be written as

\begin{equation}
\begin{aligned}
N^V_{\alpha} &= 
\sum_{i \in F} \left[
\underbrace{
\frac{y(i, N_\alpha \cap F)}{y(i,F)} 
}_{\text{siblings in $F$}}
+ 
\underbrace{
\frac{y(i, N_\alpha - F)}{y(i,F) + 1}
}_{\text{siblings not in $F$}}
\right],
\end{aligned}
\label{eq:ind-N-estimand}\end{equation}

where \(y(i, N_\alpha \cap F)\) is \(i\)'s reported number of siblings
on the sampling frame who contributed exposure; and
\(y(i, N_\alpha - F)\) is \(i\)'s reported number of siblings not on the
sampling frame who contributed exposure. Combining
Equation~\ref{eq:ind-D-estimand} and Equation~\ref{eq:ind-N-estimand},
we have the population-level \emph{individual visibility estimand}:

\begin{equation}
\begin{aligned}
M^{V}_{\alpha,\text{ind}}
&= 
\frac{
\sum_{i \in F} \left[\frac{y(i, D_\alpha)}{y(i,F) + 1}\right]
}{
\sum_{i \in F} \left[
\frac{y(i, N_\alpha \cap F)}{y(i,F)} 
+ 
\frac{y(i, N_\alpha -F)}{y(i,F) + 1}
\right]
}.
\end{aligned}
\label{eq:ind-M-estimand}\end{equation}

Finally, the population-level relationship in
Equation~\ref{eq:ind-M-estimand} motivates a sample-based estimator for
\(M^V_\alpha\) (Appendix \ref{sec:ind-vis-deathrate}):

\begin{equation}
\begin{aligned}
\widehat{M}^{V}_{\alpha,\text{ind}}
&= \frac{
\sum_{i \in s} w_i~\left[\frac{y(i, D_\alpha)}{y(i,F) + 1}\right]
}{
\sum_{i \in s} w_i~\left[
\frac{y(i, N_\alpha \cap F)}{y(i,F)} 
+ 
\frac{y(i, N_\alpha -F)}{y(i,F) + 1}
\right]},
\end{aligned}
\label{eq:mhat-ind}\end{equation}

where \(i\) indexes survey respondents in the probability sample \(s\)
and \(w_i\) is \(i\)'s sampling weight. Since this approach is based on
adjusting for the visibility of each individual reported sibling, we
call it \emph{individual visibility} estimation. Appendix
\ref{sec:ind-vis-deathrate} formally derives the estimator in
Equation~\ref{eq:mhat-ind} (Result \ref{res:mvis-ind-estimator}),
including the precise conditions required for it to provide consistent
and essentially unbiased estimates of the visible death rate.

\hypertarget{relationship-to-previous-work}{%
\subsubsection*{Relationship to previous
work}\label{relationship-to-previous-work}}
\addcontentsline{toc}{subsubsection}{Relationship to previous work}

To the best of our knowledge, our study is the first to derive the
aggregate visibility estimator from first principles. However, the
estimator itself is not new: the aggregate visibility estimator is
probably the most common approach to estimating death rates from sibling
history data. For example, Equation~\ref{eq:mhat-agg} is the estimator
used to produce age-specific adult death rate estimates in all
Demographic and Health Survey reports (Rutstein and Guillermo Rojas
2006). The estimator appears to have been first proposed in Rutenberg
and Sullivan (1991), and it has since been the subject of several
methodological analyses, including Masquelier (2013), Gakidou and King
(2006), Hill et al. (2006), Timaeus and Jasseh (2004), Stanton,
Abderrahim, and Hill (2000), and Garenne et al. (1997). By focusing on
the networked structure of sibling relations, our derivation reveals
that the aggregate visibility estimator is related to other network
estimation approaches, including the network scale-up method (Bernard et
al. 1989); and the network survival estimator (Feehan, Mahy, and
Salganik 2017).

The individual visibility estimator has its origins in multiplicity
sampling (Sirken 1970; Lavallee 2007; Feehan 2015). In the context of
sibling survival, an estimator similar to the one derived here was
introduced by Gakidou and King (2006) and then further discussed by
Obermeyer et al. (2010) and Masquelier (2013). The actual individual
estimator in Equation~\ref{eq:mhat-ind} is somewhat different from the
one proposed by Gakidou and King (2006), but both are motivated by the
idea that observed information can be used to adjust for visibility at
the level of individual reports.

\hypertarget{sec:sensitivity}{%
\section{Framework for sensitivity analysis}\label{sec:sensitivity}}

Both the individual and aggregate visibility estimators rely on several
conditions to guarantee that they will produce consistent and
essentially unbiased estimates of the death rate. These conditions make
precise longstanding concerns researchers have had about sibling
survival estimates. For example, researchers have often worried that
inaccurate reports about siblings may lead to biased death rate
estimates; our results reveal exactly how reports about siblings must be
accurate in order to produce consistent and essentially unbiased
estimates of death rates. They also reveal precise quantities which
could potentially be measured to adjust sibling reports to account for
reporting errors.

Appendices \ref{sec:agg-vis-deathrate} and \ref{sec:ind-vis-deathrate}
contain detailed derivations of the sensitivity frameworks for both the
individual and aggregate visibility estimators; here, we present and
discuss the results of that analysis. The simulation study in Appendix
\ref{sec:simulation} empirically illustrates the sensitivity frameworks
and confirms their correctness.

\hypertarget{sensitivity-of-the-aggregate-visibility-estimator}{%
\subsection{Sensitivity of the aggregate visibility
estimator}\label{sensitivity-of-the-aggregate-visibility-estimator}}

The relationship between the population death rate in group \(\alpha\)
and the aggregate visibility estimand can be written as:

\begin{equation}
M_\alpha 
= 
\underbrace{
%\frac{y(F, D_\alpha)}{y_(F, N_\alpha)} 
M^V_{\alpha,\text{agg}}
}_{\substack{\text{aggregate} \\ \text{visibility} \\ \text{estimand}}}
\times
\underbrace{
\frac{\bar{d}^V(N_\alpha,F)}{\bar{d}^V(D_\alpha,F)} 
}_{\substack{\text{visibility} \\ \text{ratio}}}
\times
\underbrace{
\frac{\gamma(F, N_\alpha)}{\gamma(F,D_\alpha)} 
}_{\substack{\text{reporting} \\ \text{accuracy}}}
\times
\underbrace{
\left[1 + p^I_{N_\alpha} (K-1)\right].
}_{\substack{\text{difference between} \\ \text{invisible and} \\ \text{visible populations}}}
\label{eq:main-agg-mult-ubersens}\end{equation}

Equation~\ref{eq:main-agg-mult-ubersens} shows that the visible death
rate can be decomposed into the product of the aggregate visibility
estimand and several adjustment factors. When all of these adjustment
factors are equal to 1, the aggregate visibility estimand is equal to
the total death rate\footnote{More generally, if these adjustment
  factors multiply out to be 1, the aggregate visibility estimand will
  be the total death rate. This means that the conditions that the
  estimator relies upon are sufficient, but not necessary.}.

The first group of adjustment factors--called the \emph{visibility
ratio}-- describes how a relationship between visibility and mortality
would affect estimated death rates. It is the ratio of the average
visibility of all siblings who contribute to exposure
(\(\bar{d}^V(N_\alpha, F)\)) and the average visibility of siblings who
die (\(\bar{d}^V(D_\alpha, F)\))\footnote{In the visibility ratio,
  \(\bar{d}(A,B)\) refers to the true average number of sibship
  connections between the average member of group \(A\) and group \(B\);
  \(\bar{d}(\cdot,\cdot)\) can differ from \(\bar{v}(\cdot,\cdot)\)
  because \(\bar{v}(\cdot,\cdot)\) could be affected by reporting
  errors. These reporting errors are accounted for in the reporting
  accuracy factor of the visibility framework. Appendix
  \ref{sec:agg-vis-deathrate} has the details.}. When there is no
relationship between these two quantities, the visibility ratio will be
1. When, say, siblings who have died tend to be in less visible sibships
than siblings overall, then this factor will tend to be greater than 1
(meaning that the death rate will be under-estimated).

The next group of adjustment factors captures the extent to which
reports about siblings are accurate. It is the ratio of a quantity that
captures the net accuracy of reports about exposure
(\(\gamma(F, N_\alpha)\)) and a quantity that captures the net accuracy
of reports about deaths (\(\gamma(F, D_\alpha)\)). Two particularly
salient findings emerge from the derivation of this group of adjustment
factors (Appendix \ref{sec:agg-vis-deathrate}): first, the estimator
requires that reporting be accurate \emph{in aggregate}, but not
necessarily at the individual level; as long as reporting errors across
individuals cancel one another out, estimates will not be affected.
Second, this group of adjustment factors shows that aggregate visibility
estimates will not be affected if reporting errors about deaths and
reporting errors about exposure balance out. In other words, if
respondents tend to, say, omit older siblings at a constant rate,
independent of the survival status of older siblings
(\(\gamma(F, D_\alpha) = \gamma(F, N_\alpha) < 1\)), then
Equation~\ref{eq:main-agg-mult-ubersens} reveals that death rates can
still be accurate because reporting errors about deaths and about
exposure will cancel out. Thus, Equation~\ref{eq:main-agg-mult-ubersens}
reveals that the death rate estimator can be robust to situations in
which respondents' reports are imperfect, but imperfect in similar ways
for siblings who die and siblings who survive.

Finally, the last group captures the conditions needed to be able to use
only information about visible deaths to estimate the total death rate.
This group depends on two quantities: \(p^I_{N_\alpha}\), the amount of
exposure that is invisible; and \(K = \frac{M^I_\alpha}{M^V_\alpha}\),
an index for how different the invisible and visible death rates are.
Below, in Section~\ref{sec:example}, we will use an empirical example to
illustrate this group of adjustment factors in more depth.

\hypertarget{sensitivity-of-the-individual-visibility-estimator}{%
\subsection{Sensitivity of the individual visibility
estimator}\label{sensitivity-of-the-individual-visibility-estimator}}

The derivation in Appendix \ref{sec:ind-vis-deathrate} reveals that the
relationship between the population death rate in group \(\alpha\) and
the individual visibility estimand can be written as:

\begin{equation}
\begin{aligned}
M_{\alpha}
&=
\underbrace{
%\frac{D^V_{\alpha, \text{ind}}}{N^V_{\alpha, \text{ind}}} 
M^V_{\alpha, \text{ind}}
}_{\substack{\text{individual} \\ \text{visibility} \\ \text{estimand}}}
\times
\underbrace{
    \frac{\bar{\gamma}^{\star}_N ~ (1 + K_N)}{\bar{\gamma}_D ~ (1 + K_D)} 
}_{\substack{\text{reporting} \\ \text{accuracy}}}
    \times
\underbrace{
\left[1 + p^I_{N_\alpha} (K-1)\right].
}_{\substack{\text{difference} \\ \text{between} \\ \text{visibile and} \\ \text{invisibile populations}}}
\end{aligned}
\label{eq:main-ind-mult-ubersens}\end{equation}

Equation~\ref{eq:main-ind-mult-ubersens} again decomposes the population
death rate into the product of the individual visibility estimand and
several groups of adjustment factors. The main insights from
Equation~\ref{eq:main-agg-mult-ubersens} also apply to the individual
visibility expression in Equation~\ref{eq:main-ind-mult-ubersens}.
However, it is worth noting a few differences between the two
frameworks. First, note that Equation~\ref{eq:main-ind-mult-ubersens}
does not have any adjustment factors related to a visibility ratio. This
is an advantage of the individual visibility estimator: it does not need
to make any assumptions about the absence of a relationship between
visibility and mortality (within the visible population). Second, the
individual visibility expression in
Equation~\ref{eq:main-ind-mult-ubersens} has a more complex set of
adjustment factors that capture reporting accuracy. This more complex
expression describes the extent to which reporting errors are correlated
with reports about deaths and reports about exposure; for example, if
reporting tends to omit deaths in sibships that have more deaths, then
\(K_D > 1\). More generally, even if deaths and exposure are
under-reported at the same average rate, so that the average reporting
adjustment factor is equal to 1, there can still be problems if the
reporting errors are different in how they are correlated with the
number of sibship deaths and exposure. In general, reporting errors
appear to be more complex under the individual visibility estimator; we
will discuss the implications of this in
Section~\ref{sec:recommendations}.

\hypertarget{sec:example}{%
\section{Empirical illustration}\label{sec:example}}

We motivate our technical results with an empirical example: the sibling
history data from the 2000 Malawi Demographic and Health Survey (Malawi
National Statistical Office and ORC Macro 2001). We chose this example
for two reasons: first, we wanted the empirical example to be a
Demographic and Health Survey, since DHS surveys are the largest
available source of sibling history data; as we write, more than 150 DHS
surveys in 60 countries have collected sibling histories over a period
of about 30 years. Second, among Demographic and Health Surveys, we
chose the 2000 Malawi DHS because it had low missingness in sibling
reports, and because its sample size of 13,220 women was close to
average\footnote{The average DHS survey that included the sibling
  history module interviewed 14,224 women (in DHS surveys, the sibling
  history module is typically only asked of women).}. Our analysis of
the 2000 Malawi DHS sibling histories uses the \texttt{siblingsurvival}
R package, which we created as a companion to this paper\footnote{The
  \texttt{siblingsurvival} package is open source and freely available
  for other researchers to use:
  \texttt{https://www.github.com/dfeehan/siblingsurvival} }.

Figure~\ref{fig:sib-ests} shows estimated death rates and confidence
intervals for male and female death rates in Malawi over the seven year
period before interviews were conducted. For males, the individual and
aggregate visibility estimates are qualitatively quite similar. For
females, however, aggregate visibility estimates are systematically
higher than individual visibility estimates, and these differences are
larger than the estimated sampling variation. We discuss a possible
explanation for this difference in more detail in Appendix
\ref{sec:agg-vis-special}.

\hypertarget{variance-estimation}{%
\subsection{Variance estimation}\label{variance-estimation}}

Researchers need to be able to estimate the sampling uncertainty of
death rate estimates calculated using Equation~\ref{eq:mhat-agg} and
Equation~\ref{eq:mhat-ind}. We recommend that researchers use a
resampling approach called the rescaled bootstrap to do so (Rao and Wu
1988; Rao, Wu, and Yue 1992). The rescaled bootstrap is appealing
because (i) it accounts for the complex sampling design that is typical
of surveys like the DHS; (ii) it enables researchers to use a single
approach to estimating the sampling uncertainty for death rates and for
other quantities (such as the internal consistency checks we introduce
below); and (iii) it has been successfully applied to other network
reporting studies (e.g.~Feehan, Mahy, and Salganik 2017). The
\texttt{siblingsurvival} package uses the rescaled bootstrap to provide
estimated sampling uncertainty for death rate estimates.

An alternate approach to estimating sampling uncertainty is to derive a
mathematical expression that relates sampling variation to a function of
study design parameters and population characteristics that are known or
that can be approximated; this approach is discussed in Appendix
\ref{sec:ap-variance}, where we illustrate how linearization can be used
to derive an approximate variance estimator for estimated death rates.
In practice, we expect the bootstrap approach discussed here to be most
useful in empirical analyses.

Figure~\ref{fig:sib-ests} shows confidence intervals for estimated male
and female death rates in Malawi over the seven year period before
interviews were conducted. To compare the amount of estimated sampling
uncertainty for the aggregate and individual visibility estimates in
Figure~\ref{fig:sib-ests}, we define the relative standard error of an
estimate of \(\widehat{M}^V_\alpha\) to be
\(\frac{\widehat{\text{se}}(\widehat{M}^V_\alpha)}{\widehat{M}^V_\alpha}\),
where \(\widehat{\text{se}}(\widehat{M}^V_\alpha)\) is the rescaled
bootstrap estimated standard error. We then calculate the average of
these relative standard errors across all age-sex groups for each
estimator. The results suggest that the individual visibility estimator
has slightly larger sampling variance than the aggregate visibility
estimator (Table \ref{tab:ci-compare-bysex}). This empirical finding is
consistent with simulation results discussed in Appendix
\ref{sec:simulation}.

\begin{figure}
\hypertarget{fig:sib-ests}{%
\centering
\includegraphics{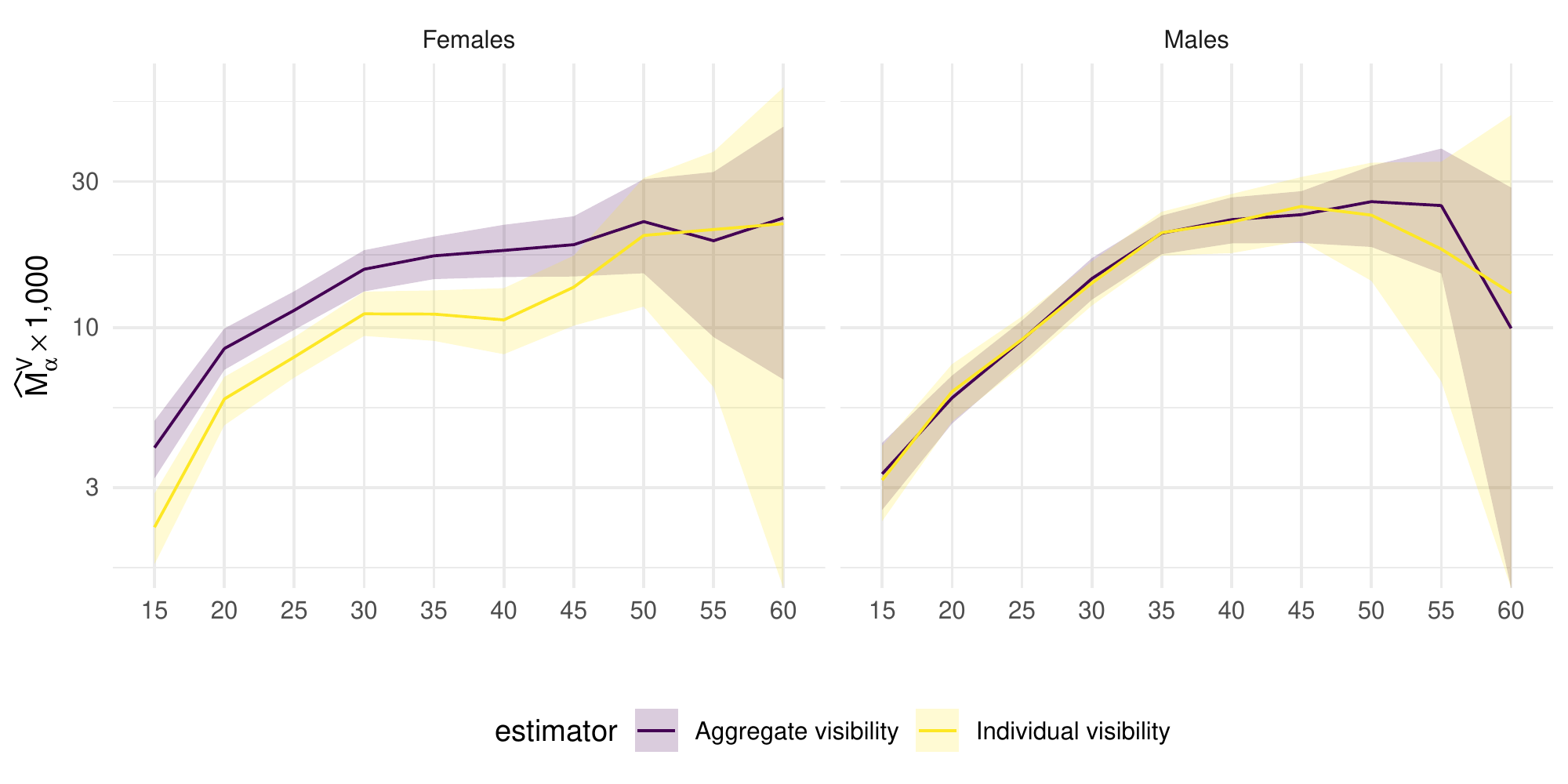}
\caption{Aggregate and individual visibility estimates for log-adult
death rates from the Malawi 2000 DHS sibling histories. Estimates are
for the 7 years prior to the interview. Confidence bands show sampling
uncertainty, accounting for the complex sample design; they were
produced using the rescaled bootstrap method (Rao and Wu 1988; Feehan
and Salganik 2016b).}\label{fig:sib-ests}
}
\end{figure}

\begin{table}[t]

\caption{\label{tab:ci-compare-bysex}Comparison in average relative standard error for aggregate and individual visibility estimates of death rates across all ages for males and for females.}
\centering
\begin{tabular}{lrr}
\toprule
\multicolumn{1}{c}{} & \multicolumn{2}{c}{\makecell[c]{Estimated Average\\Relative Standard Error}} \\
\cmidrule(l{3pt}r{3pt}){2-3}
Estimator & Females & Males\\
\midrule
Aggregate visibility & 0.16 & 0.18\\
Individual visibility & 0.21 & 0.25\\
\bottomrule
\end{tabular}
\end{table}

\hypertarget{applying-the-sensitivity-framework}{%
\subsection{Applying the sensitivity
framework}\label{applying-the-sensitivity-framework}}

\hypertarget{sensitivity-to-invisible-deaths}{%
\subsubsection*{Sensitivity to invisible
deaths}\label{sensitivity-to-invisible-deaths}}
\addcontentsline{toc}{subsubsection}{Sensitivity to invisible deaths}

Equation~\ref{eq:main-agg-mult-ubersens} expresses the difference
between the visible and total death rate in terms of two parameters:
\(K\), an index for how different the visible and invisible death rates
are; and \(p^I_{N_\alpha}\), the proportion of exposure that is
invisible. It is impossible to know the fraction of exposure that is
invisible, \(p^I_{N_\alpha}\), from sibling history data. However, we
can try to approximate this quantity by taking advantage of the fact
that we have a random sample of the frame population that is currently
alive. We can use the sample to estimate what fraction of respondents
would be invisible to sibling histories at the time of the survey, i.e.,
what fraction of survey respondents have no siblings on the sampling
frame.

Figure~\ref{fig:invis-sens-ego} shows the estimated fraction of
respondents to the 2000 Malawi DHS who would not be visible to sibling
histories. As our derivations reveal, it is this visibility---and not
sibship size \emph{per se}---that matters for estimating death rates.
Figure~\ref{fig:invis-sens-ego} gives an approximate sense for the type
of values that we might expect to see for \(p^I_{N_\alpha}\) by age.
Figure~\ref{fig:invis-sens-ego} reveals that the share of respondents
that is invisible has a U-shaped relationship with age, reaching its
highest levels among the youngest and oldest survey respondents. This
relationship is likely due to the definition of the frame population,
which included women aged 15-49; age groups close to the boundaries of
the frame population will tend to have siblings who are too old or too
young to be included as respondents, reducing the visibility of these
ages. At worst, about 40\% of women at ages 45-49 would be invisible to
sibling histories and at best, only about 15\% of women ages 30-34 would
be invisible to sibling histories.

What difference would this range of invisibility make to death rate
estimates? The sensitivity relationship in
Equation~\ref{eq:main-agg-mult-ubersens} reveals that the answer relies
upon understanding how different death rates are in the invisible and
visible populations. Figure~\ref{fig:invis-sens-surface-alt} illustrates
by showing the relative error that would result from a range of
differences between invisible and visible death rates from 20\% higher
death rates in the invisible population to 20\% lower death rates in the
invisible population (\(K\) parameter, shown on the y-axis) and a range
of different proportions of exposure that is invisible, from 15\% of
exposure invisible to 30\% of exposure invisible (\(p^I_{N_\alpha}\),
shown on the x axis). Even relatively large values for the two
parameters appear to result in modest relative errors.

\begin{figure}
\centering

\subfloat[]{\includegraphics[width=0.5\textwidth,height=\textheight]{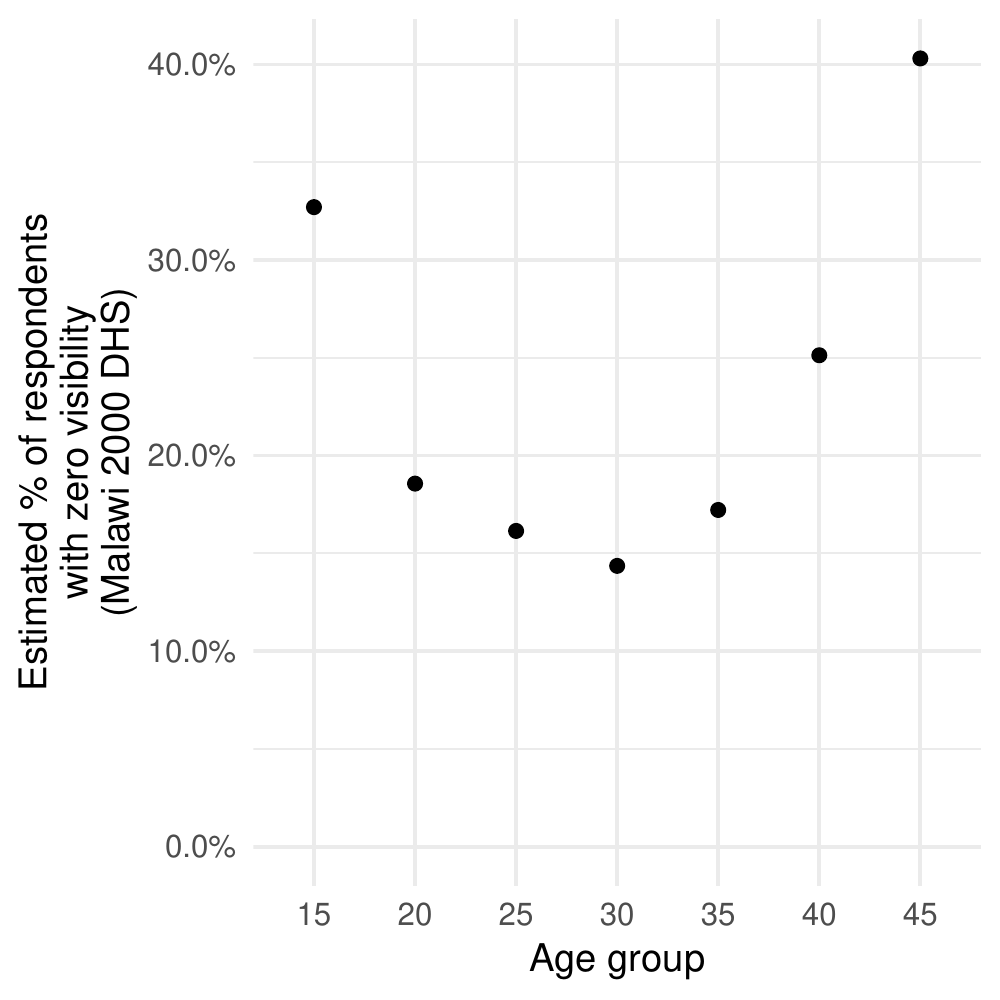}\label{fig:invis-sens-ego}}\hfill
\subfloat[]{\includegraphics[width=0.5\textwidth,height=\textheight]{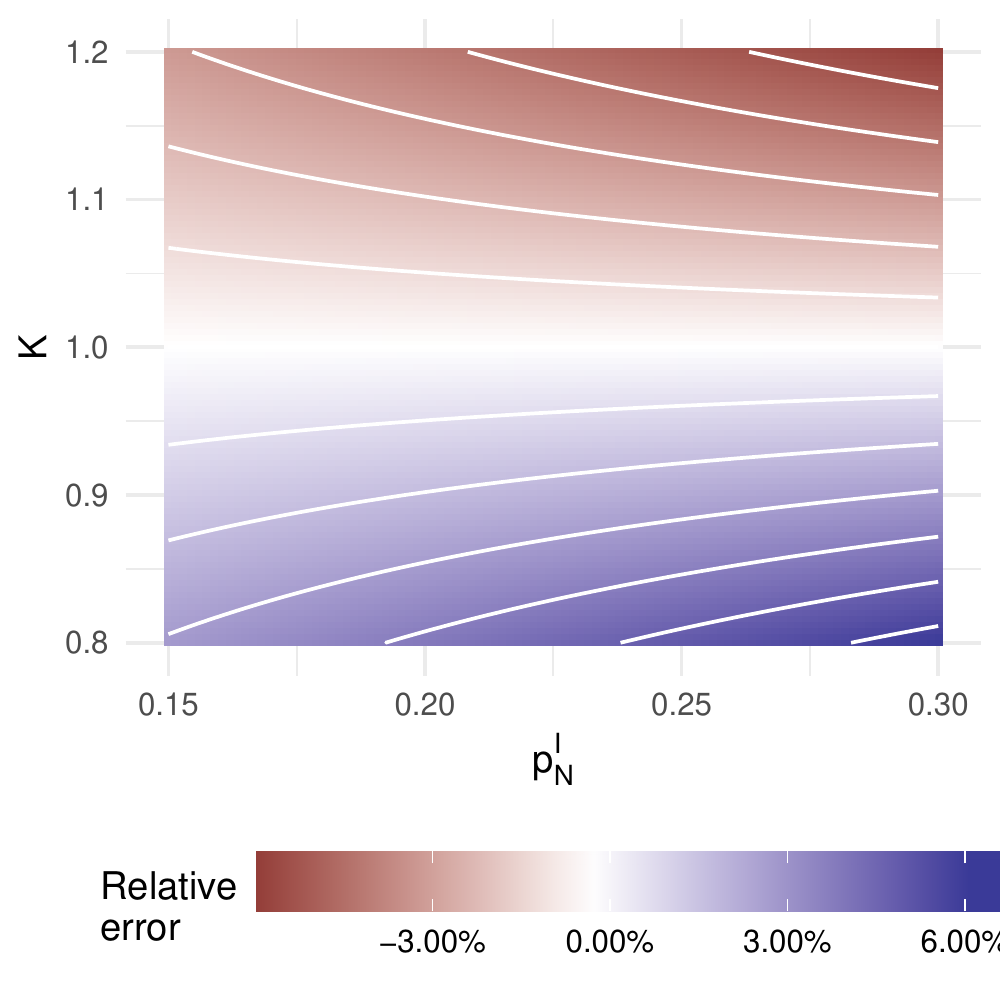}\label{fig:invis-sens-surface-alt}}

\caption{Sensitivity to invisible deaths.\\
(a) Estimated fraction of respondents to the Malawi 2000 DHS who would
not be visible in sibling history data.\\
(b) Illustration of the relative error in using the visible death rate
\(M^V\) as an estimate for the total death rate \(M\). The proportion of
exposure that is invisible, \(p^I_{N_{\alpha}}\), varies along the x
axis; the relationship between the visible and invisible death rates,
captured by the parameter \(K\) (Equation~\ref{eq:kfactor}), varies
along the y axis. The colors show the percentage relative error; so if
20\% of the population's exposure is invisible (\(p^I_{N_\alpha}=0.2\))
and the invisible death rate is 10\% higher than the visible death rate
(\(K=1.1\)), the relative error is about 2 percent. Relative error
increases as \(K\) gets farther away from 1 and as \(p^I_{N_{\alpha}}\)
increases.}

\label{fig:invis-sens}

\end{figure}

\hypertarget{reporting-errors}{%
\subsection*{Reporting errors}\label{reporting-errors}}
\addcontentsline{toc}{subsection}{Reporting errors}

Researchers have long been aware that reporting errors may affect the
quality of sibling survival estimates. We will illustrate two ways to
assess the sensitivity of aggregate visibility estimates to reporting
errors. First, we will use Equation~\ref{eq:main-agg-mult-ubersens} to
assess how much death rate estimates are affected by different levels of
reporting error. Second, we will illustrate how our network reporting
framework leads to data quality checks that can be performed on sibling
history data.

\hypertarget{analyzing-the-impact-of-reporting-errors-using-the-sensitivity-framework}{%
\subsubsection{Analyzing the impact of reporting errors using the
sensitivity
framework}\label{analyzing-the-impact-of-reporting-errors-using-the-sensitivity-framework}}

We will illustrate the sensitivity framework by focusing on aggregate
visibility estimates for brevity. The sensitivity framework in
Equation~\ref{eq:main-agg-mult-ubersens} shows that reporting errors
will affect the accuracy of aggregate visibility estimates through the
ratio of two parameters: \(\gamma(F, N_\alpha)\) and
\(\gamma(F, D_\alpha)\). Importantly, in principle it is possible to
design a study that could measure \(\gamma(F, N_\alpha)\) and
\(\gamma(F, D_\alpha)\). In fact, some promising research on the sibling
method to date has compared sibling reports to ground truth information
at a demographic surveillance site in southeastern Senegal (e.g.,
Helleringer, Pison, Kanté, et al. 2014). Studies like Helleringer,
Pison, Kanté, et al. (2014) were designed to estimate somewhat different
reporting parameters from \(\gamma(F, D_\alpha)\) and
\(\gamma(F, N_\alpha)\); thus, there is not currently any direct
evidence about these parameters available. However, to illustrate our
sensitivity framework, we can base some back of the envelope
calculations on the data reported by Helleringer and colleagues. Suppose
that respondents never mistakenly count non-siblings as siblings but
that, on average, about four percent of living siblings in group
\(\alpha\) get omitted by respondents' reports, and about nine percent
of dead siblings in group \(\alpha\) get omitted by respondents'
reports. Then \(\gamma(F, N_\alpha) \approx 0.96\) and
\(\gamma(F, D_\alpha) \approx 0.91\). If these approximations were
correct, Equation~\ref{eq:main-agg-mult-ubersens} shows that, in order
to adjust for reporting errors, the estimated death rates should be
multiplied by about \(\frac{0.95}{0.91} = 1.05\). In other words, under
this scenario, the unadjusted aggregate visibility estimator produces
estimates that have a relative error of about 5\% because of imperfect
reporting\footnote{We stress that this is a back of the envelope
  calculation and do not recommend that this value be used in practice;
  we intend it to illustrate how our framework might be used once
  measurements of \(\gamma(F, D_\alpha)\) and \(\gamma(F, N_\alpha)\)
  are available.}.

\hypertarget{sec:ic-checks}{%
\subsubsection{Internal consistency checks}\label{sec:ic-checks}}

A second strategy for assessing sensitivity to reporting errors is to
perform data quality checks on sibling history data (e.g., Helleringer,
Pison, Masquelier, et al. 2014; Masquelier and Dutreuilh 2014; Stanton,
Abderrahim, and Hill 2000; Garenne and Friedberg 1997; Rutenberg and
Sullivan 1991). We now show how our reporting framework enables us
introduce several new internal consistency checks that can be used to
assess how accurate sibling reports are. The idea is to use the network
reporting framework to identify several quantities that can be estimated
in two different ways using independent subsets of the sibling history
data (e.g., Feehan and Cobb 2019). If reporting is highly accurate, then
we expect these independent estimates to agree; when these independent
estimates are very different, that suggests that there may be
considerable amounts of reporting error.

We base our internal consistency checks on reports about siblings' ages.
For a particular age, say 30, the symmetry of the sibling relation
guarantees that in a census of the frame population, \begin{equation}
\substack{\text{\# connections people in $F$ aged 30}\\\text{report to sibs in $F$ not aged 30}} =
\substack{\text{\# connections people in $F$ not aged 30}\\\text{report to sibs in $F$ aged 30}}.
\label{eq:ic-ex}\end{equation}

The quantity on the left-hand side of Equation~\ref{eq:ic-ex} can be
estimated from the survey respondents aged 30, and the quantity on the
right-hand side of Equation~\ref{eq:ic-ex} can be estimated from all of
the survey respondents who are not aged 30. These two estimates can then
be compared; if reporting is accurate, then we expect the estimates to
agree.

Formally, for age \(\alpha\), we write \(y(F_\alpha, F_{-\alpha})\) for
the total reported connections from respondents aged \(\alpha\) to
siblings who are in \(F\) but not aged \(\alpha\); similarly, we write
\(y(F_{-\alpha}, F_{\alpha})\) for the total reported connections from
respondents not aged \(\alpha\) to siblings who are in \(F\) and who are
age \(\alpha\). In theory, these are same quantity. From sibling history
data, we can independently estimate \(y(F_\alpha, F_{-\alpha})\) and
\(y(F_{-\alpha}, F_{\alpha})\) and then compare how similar these two
independent estimates of the same quantity are by calculating
\(\Delta_\alpha\): \[
\Delta_{\alpha} = \widehat{y}(F_\alpha,F_{-\alpha}) - \widehat{y}(F_{-\alpha}, F_\alpha).
\] When the two estimates agree, \(\Delta_\alpha\) is close to zero. If
there is considerable reporting error, then \(\Delta_\alpha\) can be
very different from zero.

Figure~\ref{fig:ic-check-example} illustrates this idea by showing
internal consistency checks for each age from the 2000 Malawi DHS
sibling histories (Malawi National Statistical Office and ORC Macro
2001). The figure shows an internal consistency check based on each age
from 15 to 49. Each point shows the difference between two independent
estimates for the same quantity. If these independent estimates agreed
perfectly, they would all lie on the horizontal \(y=0\) line. The
confidence intervals capture estimated sampling variation. Most of the
confidence intervals include 0, suggesting that reports are internally
consistent; however, the figure also suggests that there is some
misreporting, particularly between ages 20 and 25.

In general, we expect that plots or quantitative summaries of internal
consistency checks like Figure~\ref{fig:ic-check-example} will be a
useful way for researchers to assess the face validity of sibling
history data. As we describe below, these internal consistency
relationships could also form the basis for developing model-based
approaches to analyzing sibling histories.

\begin{figure}
\hypertarget{fig:ic-check-example}{%
\centering
\includegraphics{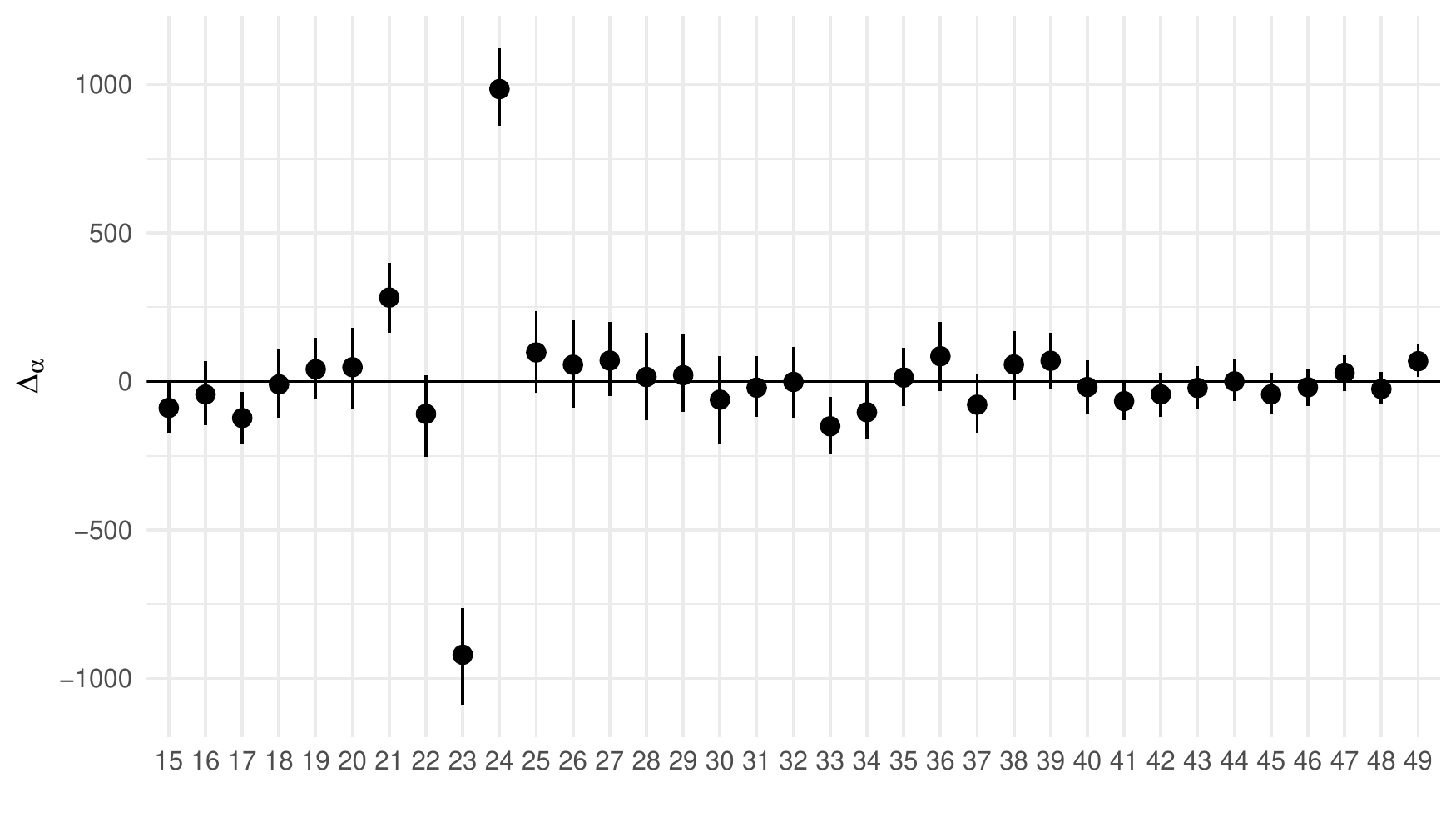}
\caption{Internal consistency checks using sibling reports from the 2000
Malawi DHS sibling histories. Each point shows the difference between
two independent estimates for the same quantity. If these independent
estimates agreed perfectly, they would all lie on the horizontal \(y=0\)
line. Most of the checks' confidence intervals include 0, suggesting
that reports are internally consistent. However, the figure also
suggests that there is some misreporting, particularly between ages 20
and 25. Confidence intervals show sampling uncertainty, accounting for
the complex sample design; they were produced using the rescaled
bootstrap method (Rao and Wu 1988; Feehan and Salganik
2016b).}\label{fig:ic-check-example}
}
\end{figure}

\hypertarget{sec:recommendations}{%
\section{Recommendations for practice}\label{sec:recommendations}}

\hypertarget{should-the-respondent-be-included-in-the-denominator}{%
\subsection*{Should the respondent be included in the
denominator?}\label{should-the-respondent-be-included-in-the-denominator}}
\addcontentsline{toc}{subsection}{Should the respondent be included in
the denominator?}

The sibling survival literature has debated whether or not respondents
should be included in the sibling reports. The concern is that
respondents are, by definition, alive: thus, it seems possible that
including respondents will bias estimated death rates downwards
(Trussell and Rodriguez 1990; Masquelier 2013). The estimates published
in DHS reports, which use the aggregate visibility estimator, do not
include the respondent. More recently, Gakidou and King (2006) argued
that respondents should be included in sibling reports, but this was
disputed by Masquelier (2013).

In our framework, deciding to include or not include respondents in the
denominator of the estimator amounts to deciding upon the
\emph{definition} of the visible and invisible populations. Since
sibling history methods estimate a visible death rate, the question is:
which visible population's death rate is more likely to be a good
estimate for the total population death rate \(M_\alpha\)?

To address this question, we analyzed the difference between the two
visible populations in Appendix \ref{sec:includerespondent}. Recall that
\(M^V_\alpha = \frac{D^V_\alpha}{N^V_\alpha}\) is the visible death rate
when respondents are excluded from sibling histories, and let
\(M^{\prime V}_\alpha = \frac{D^{\prime V}_\alpha}{N^{\prime V}_\alpha}\)
be the visible population when respondents are included in the reports.
In Appendix \ref{sec:includerespondent}, we show that the relationship
between \(M^V_\alpha\) and \(M^{\prime V}_\alpha\) can be written as

\begin{equation}
\begin{aligned}
M^{V}_\alpha &= \frac{D^{V}_{\alpha}}{N^{V}_{\alpha}}= \frac{D^{\prime V}_\alpha}{N^{\prime V}_{\alpha} - C},
\end{aligned}
\label{eq:includer-mdiff}\end{equation}

where \(C\) is the number of people who contribute exposure and who have
visibility of exactly 1 when respondents are included in the
denominator; that is, \(C\) is the number of people in the population
who are in group \(\alpha\), are eligible to respond to the survey, and
who have no siblings who would be eligible to respond to the survey.
\(C\) will tend to be bigger, inducing a bigger difference between
\(M^{\prime V}_\alpha\) and \(M^V_\alpha\), when (i) there is more
overlap between the group \(\alpha\) and the frame population; and, (ii)
visibilities tend to be smaller, meaning that more people have a
visibility exactly equal to 1 when respondents are included in the
denominator.

Equation~\ref{eq:includer-mdiff} reveals that the decision to include
respondents in reports will move some people into the denominator, but
it can never move anyone into the numerator.

A model developed in Appendix \ref{sec:includerespondent-model} shows
that, in a simple situation in which everyone has the same probability
of dying, it is most natural to exclude respondents from sibling
reports. Under the model, when respondents are excluded, the invisible
and visible populations have the same death rate, and that death rate
can be estimated from sibling reports. Including respondents, on the
other hand, induces a difference between the death rate in the visible
and invisible population (even though every individual faces the same
probability of death). Thus, our model suggests that excluding
respondents from reports is preferable, at least in the simple world it
describes. This conclusion agrees with the earlier modeling work of
Trussell and Rodriguez (1990), which also argued that respondents should
be excluded from reports.

However, these are suggestive results: without additional information,
there is no way to be certain that \(M^V_\alpha\) or
\(M^{\prime V}_\alpha\) will produce a better approximation to the
population death rate \(M_\alpha\) in a given population. The estimates
presented here exclude respondents from reports, but we note that the
full derivations of all of our estimators in Appendixes
\ref{sec:ind-vis-deathrate} and \ref{sec:agg-vis-deathrate} cover both
including and not including the respondent in the denominator. Thus,
researchers who wish to include respondents in the denominator of the
death rate can find the appropriate estimators there.

\hypertarget{aggregate-vs-individual-visibility-estimator}{%
\subsection*{Aggregate vs individual visibility
estimator}\label{aggregate-vs-individual-visibility-estimator}}
\addcontentsline{toc}{subsection}{Aggregate vs individual visibility
estimator}

Our analysis suggests that there are strengths and limitations to both
the individual and aggregate visibility estimators. Aggregate visibility
estimates are typically based on the assumption that the visibility of
deaths and the visibility of exposure are equal. The individual
visibility estimator avoids this assumption altogether; thus, in
situations where no information about adjustment factors is available,
we recommend using the individual visibility estimator. In Appendix
\ref{sec:agg-vis-special}, we analyze the difference in aggregate and
individual visibility estimates for Malawi, and we show that it is
likely that this visibility assumption explains most of the difference
between death rate estimates for females in Figure~\ref{fig:sib-ests}.

The individual visibility estimator has some disadvantages. Table
\ref{tab:ci-compare-bysex} and the simulation study in Appendix
\ref{sec:simulation} suggest that the individual visibility estimator
has slightly higher sampling variance than the aggregate visibility
estimator. We hope that future research will continue to systematically
compare the aggregate and individual visibility estimators; in the
meantime, our view is that the individual visibility estimator's
relatively small loss in precision is a price that is worth paying to
avoid having to make assumptions about the visibility of deaths and
exposure.

Another disadvantage of the individual visibility estimator comes from
the comparing the aggregate sensitivity framework
(Equation~\ref{eq:main-agg-mult-ubersens}) with the individual
sensitivity framework (Equation~\ref{eq:main-ind-mult-ubersens}). This
comparison reveals that the quantities that would be needed adjust the
individual visibility estimator for reporting error are much more
complex than the analogous quantities needed to adjust the aggregate
visibility estimator. Thus, if data that can be used to estimate
adjustment factors become more widely available, then we expect the
relative appeal of the aggregate visibility estimator to increase.

To recap, we recommend excluding respondents from reports. We also
recommend using the individual visibility estimator in the absence of
any empirical estimates for adjustment factors. However, we expect
empirical information about adjustment factors to be much easier to
collect for the aggregate visibility estimator. Thus, as empirical
information about adjustment factors becomes available, we expect the
aggregate visibility estimator to be more attractive. In all cases, we
recommend that researchers who produce sibling history-based estimates
use the sensitivity frameworks to assess how estimated death rates are
affected by the assumptions used to produce them.

\hypertarget{sec:conclusion}{%
\section{Discussion and Conclusion}\label{sec:conclusion}}

We showed how sibling history data can be understood as a type of
network reporting. We explained how to derive network-based estimators
for adult death rates, how to devise internal consistency checks, and
how to understand how sensitive death rate estimates can be to the
different conditions that the estimators rely upon. We illustrated with
an empirical example, based on our freely available \texttt{R} package
\texttt{siblingsurvival}, and we outlined several recommendations for
practitioners who wish to estimate death rates from sibling histories.

We see several important avenues for future research. Methodologically,
a deeper comparison between the individual and aggregate visibility
estimators would be useful. In particular, analytic results could help
better understand our empirical finding that the individual visibility
estimator has somewhat higher sampling variance than the aggregate
visibility estimator. This analysis could also produce insights that
might be useful for designing future data collection.

Our results also suggest next steps for developing models for death rate
estimates based on sibling histories. This paper has focused on
design-based estimators for death rates using sibling histories. Future
research can use these design-based estimators as the starting point for
developing model-based estimators. For example, the internal consistency
checks that we discuss could form the basis for model-based adjustments
of sibling reports (see McCormick, Salganik, and Zheng (2010) for a
similar approach that has been developed in the context of aggregate
relational data). Our framework also offers a natural way to think about
how to pool information across countries and time periods.

We believe that collecting more information on sibling reports from
settings where gold-standard adult death rates are available is crucial;
Helleringer, Pison, Kanté, et al. (2014) is a useful template for the
type of study design that could help produce more information. This type
of study investigates the properties of sibling history reports in small
areas where a gold-standard underlying truth about adult death rates is
available. Data collected in this way can produce the information needed
to estimate the adjustment factors in the individual and aggregate
visibility sensitivity relationships. Combined with the framework we
introduce here, estimates for these adjustment factors could provide a
principled way to adjust national-level sibling survival estimates from
surveys like the DHS, relaxing the conditions required for the estimates
to be accurate.

Our analysis focused on death rate estimates for the time period
immediately preceding the survey. In principle, the estimators discussed
here could also be used for more distant time periods; however, the
assumptions--though mathematically the same--presumably get stronger
farther into the past. Future work could investigate this topic in more
depth.

Our framework can also be applied to other demographic estimation
techniques related to sibling survival, such as methods in which people
report about their parents, spouses, or children (Hill et al. 1983;
Moultrie et al. 2013). More generally, ideas from the sibling survival
literature can be used to develop new methods for collecting data that
have the potential to overcome some of the limitations of sibling
histories. Feehan, Mahy, and Salganik (2017) explored how reports about
two social network relationships could form the basis for death rate
estimates at adult ages (see also Feehan et al. (2016)). Future research
could continue to explore how to collect reports about more general
types of relations, such as broader kin relationships or other types of
social networks, with the goal of producing information that is timely
and accurate enough to estimate adult death rates.

\hypertarget{references}{%
\section*{References}\label{references}}
\addcontentsline{toc}{section}{References}

\hypertarget{refs}{}
\leavevmode\hypertarget{ref-abouzahr_universal_2015}{}%
AbouZahr, Carla, Don de Savigny, Lene Mikkelsen, Philip W. Setel, Rafael
Lozano, and Alan D. Lopez. 2015. ``Towards Universal Civil Registration
and Vital Statistics Systems: The Time Is Now.'' \emph{The Lancet}.
\url{http://www.sciencedirect.com/science/article/pii/S0140673615601702}.

\leavevmode\hypertarget{ref-bernard_estimating_1989}{}%
Bernard, H. Russell, Eugene C. Johnsen, Peter D. Killworth, and Scott
Robinson. 1989. ``Estimating the Size of an Average Personal Network and
of an Event Subpopulation.'' In \emph{The Small World}, edited by
Manferd Kochen, 159--75. Norwood, NJ: Ablex Publishing.

\leavevmode\hypertarget{ref-brass_methods_1975}{}%
Brass, William. 1975. ``Methods for Estimating Fertility and Mortality
from Limited and Defective Data.'' \emph{Methods for Estimating
Fertility and Mortality from Limited and Defective Data.}
\url{http://www.cabdirect.org/abstracts/19762901082.html}.

\leavevmode\hypertarget{ref-corsi_demographic_2012}{}%
Corsi, Daniel J., Melissa Neuman, Jocelyn E. Finlay, and S. V.
Subramanian. 2012. ``Demographic and Health Surveys: A Profile.''
\emph{International Journal of Epidemiology} 41 (6): 1602--13.
\url{http://ije.oxfordjournals.org/content/41/6/1602.short}.

\leavevmode\hypertarget{ref-fabic_systematic_2012}{}%
Fabic, Madeleine Short, YoonJoung Choi, and Sandra Bird. 2012. ``A
Systematic Review of Demographic and Health Surveys: Data Availability
and Utilization for Research.'' \emph{Bulletin of the World Health
Organization} 90 (8): 604--12.
\url{http://www.scielosp.org/scielo.php?pid=S0042-96862012000800012/\&script=sci_arttext}.

\leavevmode\hypertarget{ref-feehan_network_2015}{}%
Feehan, Dennis M. 2015. ``Network Reporting Methods.'' PhD thesis,
Princeton University. \url{http://gradworks.umi.com/37/29/3729745.html}.

\leavevmode\hypertarget{ref-feehan_using_2019}{}%
Feehan, Dennis M., and Curtiss Cobb. 2019. ``Using an Online Sample to
Learn About an Offline Population.'' \emph{arXiv:1902.08289 {[}Stat{]}},
February. \url{http://arxiv.org/abs/1902.08289}.

\leavevmode\hypertarget{ref-feehan_network_2017}{}%
Feehan, Dennis M., Mary Mahy, and Matthew J. Salganik. 2017. ``The
Network Survival Method for Estimating Adult Mortality: Evidence from a
Survey Experiment in Rwanda.'' \emph{Demography} 54 (4): 1503--28.
\url{https://link.springer.com/article/10.1007/s13524-017-0594-y}.

\leavevmode\hypertarget{ref-feehan_generalizing_2016}{}%
Feehan, Dennis M., and Matthew J. Salganik. 2016a. ``Generalizing the
Network Scale-up Method: A New Estimator for the Size of Hidden
Populations.'' \emph{Sociological Methodology} 46 (1): 153--86.
\url{http://128.84.21.199/pdf/1404.4009.pdf}.

\leavevmode\hypertarget{ref-feehan_surveybootstrap_2016}{}%
---------. 2016b. \emph{Surveybootstrap: Tools for the Bootstrap with
Survey Data}. \url{https://CRAN.R-project.org/package=surveybootstrap}.

\leavevmode\hypertarget{ref-feehan_quantity_2016}{}%
Feehan, Dennis M., Aline Umubyeyi, Mary Mahy, Wolfgang Hladik, and
Matthew J. Salganik. 2016. ``Quantity Versus Quality: A Survey
Experiment to Improve the Network Scale-up Method.'' \emph{American
Journal of Epidemiology}, March, kwv287.

\leavevmode\hypertarget{ref-gakidou_death_2006}{}%
Gakidou, E., and G. King. 2006. ``Death by Survey: Estimating Adult
Mortality Without Selection Bias from Sibling Survival Data.''
\emph{Demography} 43 (3): 569--85.
\url{http://www.springerlink.com/index/W2Q1X41501666JL0.pdf}.

\leavevmode\hypertarget{ref-garenne_accuracy_1997}{}%
Garenne, M., and F. Friedberg. 1997. ``Accuracy of Indirect Estimates of
Maternal Mortality: A Simulation Model.'' \emph{Studies in Family
Planning}, 132--42. \url{http://www.jstor.org/stable/10.2307/2138115}.

\leavevmode\hypertarget{ref-garenne_direct_1997}{}%
Garenne, M., R. Sauerborn, A. Nougtara, M. Borchert, J. Benzler, and J.
Diesfeld. 1997. ``Direct and Indirect Estimates of Maternal Mortality in
Rural Burkina Faso.'' \emph{Studies in Family Planning}, 54--61.
\url{http://www.jstor.org/stable/10.2307/2137971}.

\leavevmode\hypertarget{ref-graham_estimating_1989}{}%
Graham, Wendy, William Brass, and Robert W. Snow. 1989. ``Estimating
Maternal Mortality: The Sisterhood Method.'' \emph{Studies in Family
Planning}, 125--35. \url{http://www.jstor.org/stable/1966567}.

\leavevmode\hypertarget{ref-hanley_confidence_1996}{}%
Hanley, James A., Catherine A. Hagen, and Tesfaye Shiferaw. 1996.
``Confidence Intervals and Sample-Size Calculations for the Sisterhood
Method of Estimating Maternal Mortality.'' \emph{Studies in Family
Planning}, 220--27.

\leavevmode\hypertarget{ref-helleringer_reporting_2014}{}%
Helleringer, Stéphane, Gilles Pison, Almamy M. Kanté, Géraldine Duthé,
and Armelle Andro. 2014. ``Reporting Errors in Siblings' Survival
Histories and Their Impact on Adult Mortality Estimates: Results from a
Record Linkage Study in Senegal.'' \emph{Demography} 51 (2): 387--411.
\url{http://link.springer.com/article/10.1007/s13524-013-0268-3}.

\leavevmode\hypertarget{ref-helleringer_improving_2014}{}%
Helleringer, Stéphane, Gilles Pison, Bruno Masquelier, Almamy Malick
Kanté, Laetitia Douillot, Géraldine Duthé, Cheikh Sokhna, and Valérie
Delaunay. 2014. ``Improving the Quality of Adult Mortality Data
Collected in Demographic Surveys: Validation Study of a New Siblings'
Survival Questionnaire in Niakhar, Senegal.'' \emph{PLoS Med} 11 (5):
e1001652.

\leavevmode\hypertarget{ref-hill_how_2006}{}%
Hill, K., S. El Arifeen, M. Koenig, A. Al-Sabir, K. Jamil, and H.
Raggers. 2006. ``How Should We Measure Maternal Mortality in the
Developing World? A Comparison of Household Deaths and Sibling History
Approaches.'' \emph{Bulletin of the World Health Organization} 84 (3):
173--80.
\url{http://www.scielosp.org/scielo.php?pid=S0042-96862006000300011/\&script=sci_arttext}.

\leavevmode\hypertarget{ref-hill_manual_1983}{}%
Hill, K., H. Zlotnik, J. Trussell, United Nations Dept of International
Economic, Social Affairs Population Division, National Research Council
(US) Committee on Population, Demography, and National Academy of
Sciences (US). 1983. \emph{Manual X: Indirect Techniques for Demographic
Estimation}. United Nations.

\leavevmode\hypertarget{ref-lavallee_indirect_2007}{}%
Lavallee, P. 2007. \emph{Indirect Sampling}. New York: Springer-Verlag.
\url{http://books.google.com/books?hl=en/\&lr=/\&id=o93cnlP9tMMC/\&oi=fnd/\&pg=PA1/\&dq=indirect+sampling+lavallee/\&ots=nhf1KvhIEk/\&sig=_7W13JSq39Iqe1WNclnr3HPk9ts}.

\leavevmode\hypertarget{ref-malawinationalstatisticaloffice_malawi_2001}{}%
Malawi National Statistical Office, and ORC Macro. 2001. \emph{Malawi
Demographic and Health Survey 2000}. Zomba, Malawi: National Statistical
Office.

\leavevmode\hypertarget{ref-masquelier_adult_2013}{}%
Masquelier, Bruno. 2013. ``Adult Mortality from Sibling Survival Data: A
Reappraisal of Selection Biases.'' \emph{Demography} 50 (1): 207--28.
\url{http://link.springer.com/article/10.1007/s13524-012-0149-1}.

\leavevmode\hypertarget{ref-masquelier_sibship_2014}{}%
Masquelier, Bruno, and Catriona Dutreuilh. 2014. ``Sibship Sizes and
Family Sizes in Survey Data Used to Estimate Mortality.''
\emph{Population, English Edition} 69 (2): 221--38.
\url{http://muse.jhu.edu/journals/population/v069/69.2.masquelier.html}.

\leavevmode\hypertarget{ref-mccormick_how_2010}{}%
McCormick, Tyler H., Matthew J. Salganik, and Tian Zheng. 2010. ``How
Many People Do You Know?: Efficiently Estimating Personal Network
Size.'' \emph{Journal of the American Statistical Association} 105
(489): 59--70.

\leavevmode\hypertarget{ref-mikkelsen_global_2015}{}%
Mikkelsen, Lene, David E. Phillips, Carla AbouZahr, Philip W. Setel, Don
de Savigny, Rafael Lozano, and Alan D. Lopez. 2015. ``A Global
Assessment of Civil Registration and Vital Statistics Systems:
Monitoring Data Quality and Progress.'' \emph{The Lancet}.
\url{http://www.sciencedirect.com/science/article/pii/S0140673615601714}.

\leavevmode\hypertarget{ref-moultrie_tools_2013a}{}%
Moultrie, TA, RE Dorrington, AG Hill, K Hill, IM Timaeus, and B Zaba.
2013. \emph{Tools for Demographic Estimation.} Paris: International
Union for the Scientific Study of Population.
\url{http://demographicestimation.iussp.org/}.

\leavevmode\hypertarget{ref-obermeyer_measuring_2010}{}%
Obermeyer, Z., J. K. Rajaratnam, C. H. Park, E. Gakidou, M. C. Hogan, A.
D. Lopez, and C. J. L. Murray. 2010. ``Measuring Adult Mortality Using
Sibling Survival: A New Analytical Method and New Results for 44
Countries, 19742006.'' \emph{PLoS Medicine} 7 (4): e1000260.
\url{http://dx.plos.org/10.1371/journal.pmed.1000260}.

\leavevmode\hypertarget{ref-rao_double_1968}{}%
Rao, J. N. K., and Norma P. Pereira. 1968. ``On Double Ratio
Estimators.'' \emph{Sankhyā: The Indian Journal of Statistics, Series A
(1961-2002)} 30 (1): 83--90.

\leavevmode\hypertarget{ref-rao_resampling_1988}{}%
Rao, J. N. K., and C. F. J. Wu. 1988. ``Resampling Inference with
Complex Survey Data.'' \emph{Journal of the American Statistical
Association} 83 (401): 231--41.

\leavevmode\hypertarget{ref-rao_recent_1992}{}%
Rao, JNK, CFJ Wu, and K Yue. 1992. ``Some Recent Work on Resampling
Methods for Complex Surveys.'' \emph{Survey Methodology} 18 (2):
209--17.

\leavevmode\hypertarget{ref-reniers_adult_2011}{}%
Reniers, G., B. Masquelier, and P. Gerland. 2011. ``Adult Mortality in
Africa.'' In \emph{International Handbook of Adult Mortality}, 151--70.
\url{http://www.springerlink.com/index/G82222M300072147.pdf}.

\leavevmode\hypertarget{ref-rutenberg_direct_1991}{}%
Rutenberg, N., and J. M. Sullivan. 1991. ``Direct and Indirect Estimates
of Maternal Mortality from the Sisterhood Method.'' In \emph{Proceedings
of the Demographic and Health Surveys World Conference}, 3:1669--96.

\leavevmode\hypertarget{ref-rutstein_guide_2006}{}%
Rutstein, S. O., and M. C. S. Guillermo Rojas. 2006. \emph{Guide to DHS
Statistics}. ORC Macro, Calverton, MD.

\leavevmode\hypertarget{ref-sarndal_model_2003}{}%
Sarndal, C. E., B. Swensson, and J. Wretman. 2003. \emph{Model Assisted
Survey Sampling}. Springer Verlag.
\url{http://books.google.com/books?hl=en/\&lr=/\&id=ufdONK3E1TcC/\&oi=fnd/\&pg=PR5/\&dq=sarndal+swensson+wretman+model+assisted/\&ots=7eZV4u7FOC/\&sig=tdK954DVTis0gvMz7r4SapBVnYg}.

\leavevmode\hypertarget{ref-setel_scandal_2007}{}%
Setel, P. W., S. B. Macfarlane, S. Szreter, L. Mikkelsen, P. Jha, S.
Stout, and C. AbouZahr. 2007. ``A Scandal of Invisibility: Making
Everyone Count by Counting Everyone.'' \emph{The Lancet} 370 (9598):
1569--77.
\url{http://www.sciencedirect.com/science/article/pii/S0140673607613075}.

\leavevmode\hypertarget{ref-sirken_household_1970}{}%
Sirken, Monroe G. 1970. ``Household Surveys with Multiplicity.''
\emph{Journal of the American Statistical Association} 65 (329):
257--66.

\leavevmode\hypertarget{ref-stanton_assessment_2000}{}%
Stanton, C., N. Abderrahim, and K. Hill. 2000. ``An Assessment of DHS
Maternal Mortality Indicators.'' \emph{Studies in Family Planning} 31
(2): 111--23.
\url{http://onlinelibrary.wiley.com/doi/10.1111/j.1728-4465.2000.00111.x/abstract}.

\leavevmode\hypertarget{ref-thompson_sampling_2002}{}%
Thompson, Steven K. 2002. \emph{Sampling}. 2nd ed. Wiley.

\leavevmode\hypertarget{ref-timaeus_adult_2004}{}%
Timaeus, I. M., and M. Jasseh. 2004. ``Adult Mortality in Sub-Saharan
Africa: Evidence from Demographic and Health Surveys.''
\emph{Demography} 41 (4): 757--72.
\url{http://www.springerlink.com/index/A2023R3756536V92.pdf}.

\leavevmode\hypertarget{ref-trussell_note_1990}{}%
Trussell, J., and G. Rodriguez. 1990. ``A Note on the Sisterhood
Estimator of Maternal Mortality.'' \emph{Studies in Family Planning} 21
(6): 344--46.

\leavevmode\hypertarget{ref-wolter_introduction_2007}{}%
Wolter, Kirk. 2007. \emph{Introduction to Variance Estimation}. 2nd ed.
New York: Springer.

\newpage
\appendix

\hypertarget{online-appendix}{%
\section*{Online Appendix}\label{online-appendix}}
\addcontentsline{toc}{section}{Online Appendix}

~

\hypertarget{sec:notation}{%
\section{Notation}\label{sec:notation}}

We follow the notation used in Feehan and Salganik (2016a) and Feehan,
Mahy, and Salganik (2017). For convenience, here is a table summarizing
key features of the notation:

%%%%%%%%%%%%%%%%%%%%%%%%%%%%%%%%%%%%%%%%%%%%%%%%%%%%%%%%%%%%
% notation table
%%%%%%%%%%%%%%%%%%%%%%%%%%%%%%%%%%%%%%%%%%%%%%%%%%%%%%%%%%%%

%\begin{table}
%{\small
\begin{longtable}{ r c p{10cm} }
\caption{Notation for network reporting quantities used in this paper.}
\label{tab:ap-notation} \\
     Quantity & %
     &%
     Explanation \\
     \hline \endhead
     \hline \endfoot
     %%%%%%%%%%%%%%%%%%%%%%%%%
    $U$ &%
    &%
    the entire population\\
    %%%%%%%%%%%%%%%%%%%%%%%%%
    $F$ &%
    &%
    the frame population (typically adults over a certain age)\\
    %%%%%%%%%%%%%%%%%%%%%%%%%
     $|U| = N$ & %
     &%
     size of the entire population, $|U|$ (i.e., everyone who could ever be interviewed or reported about) \\
     %%%%%%%%%%%%%%%%%%%%%%%%%
     $|F| = N_F$ & %
     &%
     size of the frame population, $|F|$ (i.e., everyone who could ever be interviewed)\\
     %%%%%%%%%%%%%%%%%%%%%%%%%
    $y_{i,A} = y(i,A)$ & %
    &%
    out-reports from $i$ about connections to $A$ 
    (e.g. $i$'s reported number of siblings in group $A$)\\
    %%%%%%%%%%%%%%%%%%%%%%%%%
    $y_{F,A} = y(F,A)$ & %
     &%
     total out-reports from the frame population $F$ about connections to $A$ \\
     %%%%%%%%%%%%%%%%%%%%%%%%%
    $y^{+}_{F,A} = y^{+}(F,A)$ & %
     &%
     true positive out-reports from the frame population $F$ about connections to $A$
     (i.e., the sum of reported connections that actually lead to people in $A$)\\
    %%%%%%%%%%%%%%%%%%%%%%%%%
     $d_{i,F} = d(i,F)$ & %
     &%
     number of network connections from $i$ to $F$, i.e., number of $i$'s siblings in $F$
     (which is not necessarily the same as the number of reported connections
      from $i$ to $F$)\\
    %%%%%%%%%%%%%%%%%%%%%%%%%
      $d_{A,B} = d(A,B)$ & %
     &%
     number of network connections from group $A \subset U$ to group $B \subset U$ 
     (note that $A$ and $B$ could be the same group, they could be entirely distinct
      groups, or they could overlap partially) \\
     %%%%%%%%%%%%%%%%%%%%%%%%%
      $\bar{d}_{A,B} = \bar{d}(A,B) = \frac{d_{A,B}}{N_A}$ & %
     &%
     average number of network connections from group $A$ to group $B$, per member of group $A$ 
     (note that we always take averages with respect to the first subscript)\\
     %%%%%%%%%%%%%%%%%%%%%%%%%
     $\bar{v}_{A,F} = \bar{v}(A,F) = \frac{v_{A,F}}{|A|}$ & %
     &%
     average visibility (number of in-reports) per member of A \\
     %%%%%%%%%%%%%%%%%%%%%%%%%
     $s$ & %
     &%
     a probability sample of people from the frame population $F$ \\
     %%%%%%%%%%%%%%%%%%%%%%%%%
     $\pi_i$ & %
     &%
     the probability that $i \in F$ is included in the sample, which
     comes from the sampling design\\
     %%%%%%%%%%%%%%%%%%%%%%%%%
     $\tau_F$ & %
     &%
     the true positive rate for out-reports from $F$\\
     %%%%%%%%%%%%%%%%%%%%%%%%%
     $\delta_F$ & %
     &%
     the degree ratio of hidden population members relative to frame population members\\
     %%%%%%%%%%%%%%%%%%%%%%%%%
     $\eta_F$ & %
     &%
     the false positive rate for out-reports from $F$ \\
     %%%%%%%%%%%%%%%%%%%%%%%%%
     $\alpha$ & %
     &%
     a demographic group (e.g., women aged 45-54 in 2010) \\
     %%%%%%%%%%%%%%%%%%%%%%%%%
     $F_\alpha$ & %
     &%
     frame population members who are in demographic group $\alpha$ \\
     %%%%%%%%%%%%%%%%%%%%%%%%%
     $|F_\alpha|$ & %
     &%
     the number of frame population members who are in demographic group $\alpha$ \\
     %%%%%%%%%%%%%%%%%%%%%%%%%
     $|N_\alpha|$ & %
     &%
     the number of people in the entire population $U$ who are also in demographic group $\alpha$ \\
     %%%%%%%%%%%%%%%%%%%%%%%%%
     $D_\alpha$ & %
     &%
     the number of deaths in demographic group $\alpha$ 
     (e.g. number of deaths among women aged 45-54 in 2010)\\
     %%%%%%%%%%%%%%%%%%%%%%%%%
     $M_\alpha = \frac{D_\alpha}{N_\alpha}$ & %
     &%
     the death rate in demographic group $\alpha$ 
     (e.g. the death rate for women aged 45-54 in 2010; exposure is approximated by the population size)\\
     %%%%%%%%%%%%%%%%%%%%%%%%%
     $\Sigma$ & %
     &%
     the set of all sibships in the population\\
     %%%%%%%%%%%%%%%%%%%%%%%%%
     $\sigma$ & %
     &%
     a specific sibship $\sigma \in \Sigma$\\  
     %%%%%%%%%%%%%%%%%%%%%%%%%
     $\sigma[i]$ & %
     &%
     the specific sibship containing person $i$\\  
     %%%%%%%%%%%%%%%%%%%%%%%%%
\end{longtable}
%}
%\caption{Notation used for network scale-up and network reporting analysis in this document.}
%\label{tab:ap-notation}
%\end{table}

\textbf{Sibship structure}

We define \(\Sigma\) to be the set of sibships in the population, and we
use \(\sigma\) to index the sibships in \(\Sigma\). \(\Sigma\) is a
partition of the population, meaning that each population member \(i\)
is in one and only one sibship \(\sigma \in \Sigma\). We will sometimes
denote the sibship containing \(i\) by \(\sigma[i]\).

\hypertarget{sec:estimands}{%
\section{Estimands}\label{sec:estimands}}

This appendix focuses on developing several important \emph{estimands}:
population-level relationships that form the basis for the different
death rate estimators that are developed in subsequent appendixes.

We shall see that researchers have two important questions to answer
when forming a death rate estimator from sibling histories: (1) should
reports about siblings be adjusted for visibility at the individual or
at the aggregate level?; and (2) should information about the survey
respondent be included or excluded from the reports? Taken together,
these two questions lead to four different ways that visible death rates
can be estimated from sibling histories.

This appendix begins by developing general expressions for visibility in
each of the four cases: individual visibility when respondents are
included in reports; individual visibility when respondents are excluded
from reports; aggregate visibility when respondents are included in
reports; and aggregate visibility when respondents are excluded from
reports.

Next, we use these expressions for visibility to develop
population-level identities for (1) the number of deaths; (2) the amount
of exposure; and (3) the visible death rate in each case. These
identities hold in a census when reporting is perfectly accurate; later
appendixes will show how estimates will be affected by sampling and by
different types of reporting error.

\hypertarget{sec:visibility}{%
\subsection{Visibility and whether or not the respondent is
included}\label{sec:visibility}}

As described in the main paper (Section~\ref{sec:adjusting-visibility}),
a critical step in making estimates using sibling histories is to adjust
for the \emph{visibility} of reported siblings--that is, to account for
the fact that each sibling could be reported multiple times in a census
of the frame population. In this section, we derive some important
relationships that will be helpful in developing estimators that adjust
for visibility at both the individual and at the aggregate level. We
will see that the decision to include or exclude information about
respondents from sibling reports affects how visibility is calculated.
Therefore, we develop expressions for visibility both when the
respondent is included in reports, and when the respondent is not
included in reports.

\begin{figure}
\centering

\subfloat[]{\includegraphics[width=0.33\textwidth,height=\textheight]{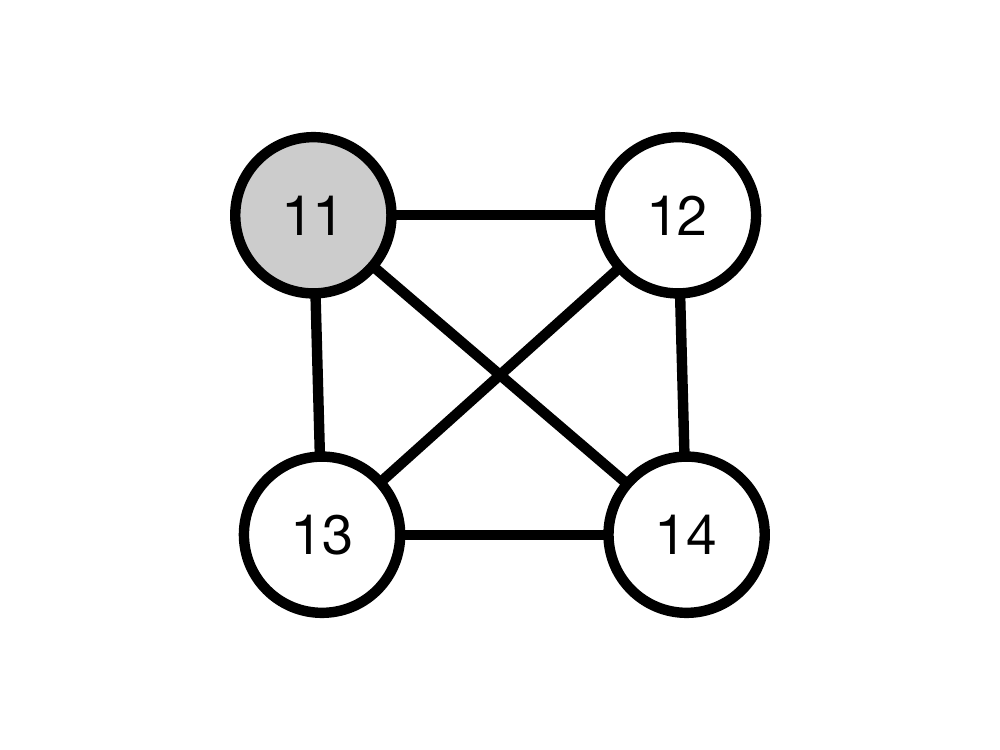}\label{fig:ap-sibship-net}}
\subfloat[]{\includegraphics[width=0.33\textwidth,height=\textheight]{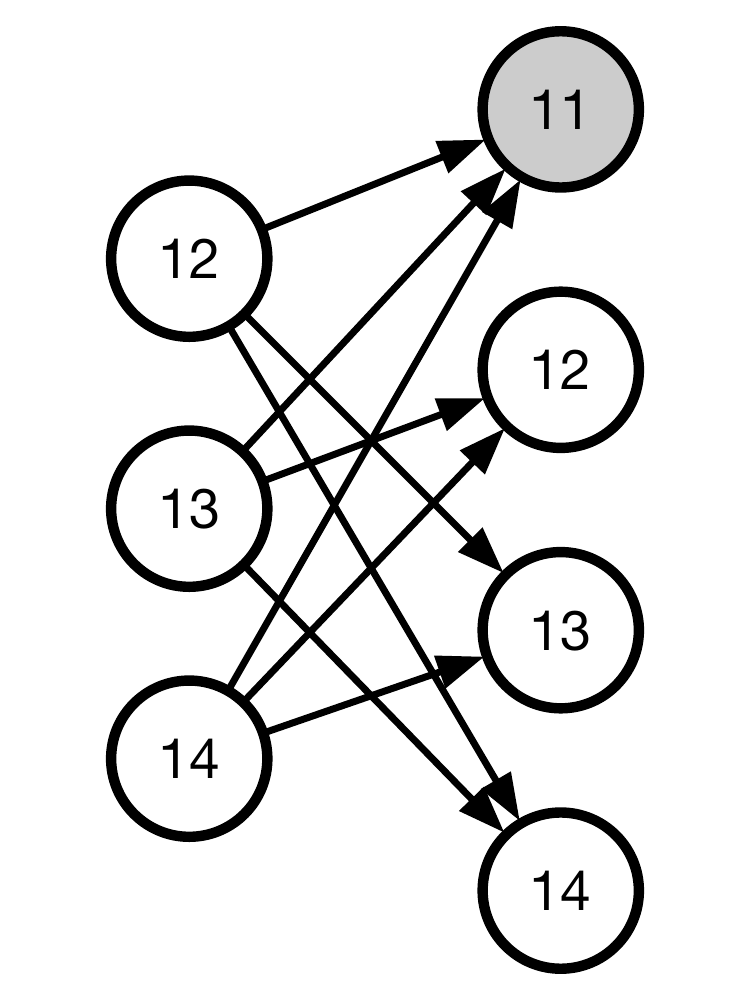}\label{fig:ap-sibship-reporting}}

\subfloat[]{\includegraphics[width=0.33\textwidth,height=\textheight]{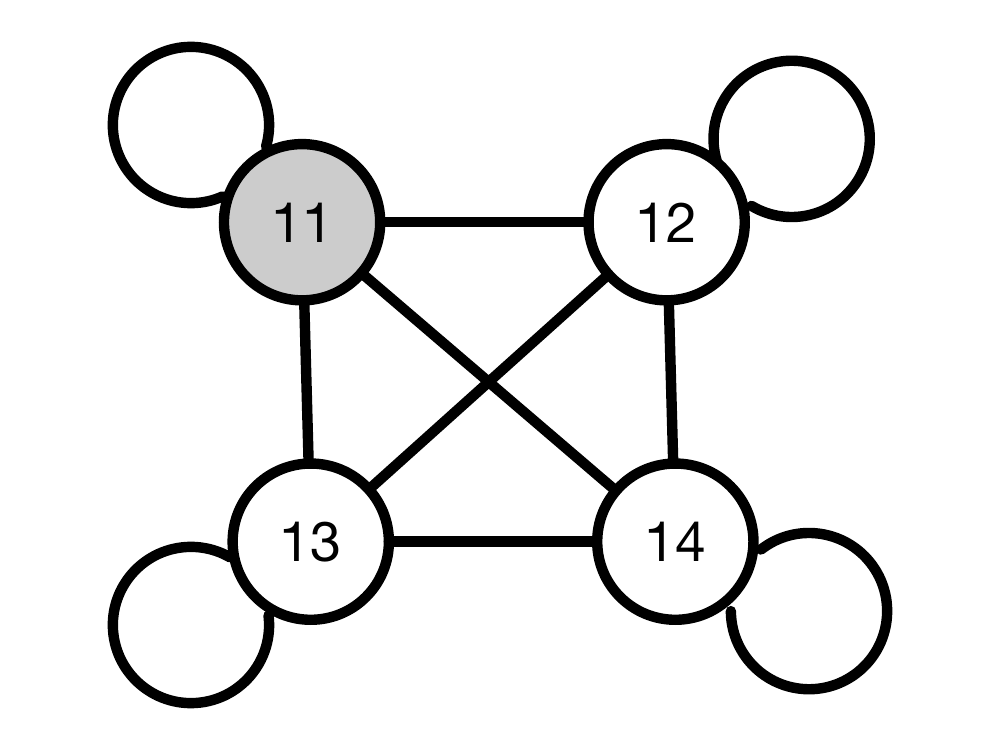}\label{fig:ap-sibship-net-withresp}}
\subfloat[]{\includegraphics[width=0.33\textwidth,height=\textheight]{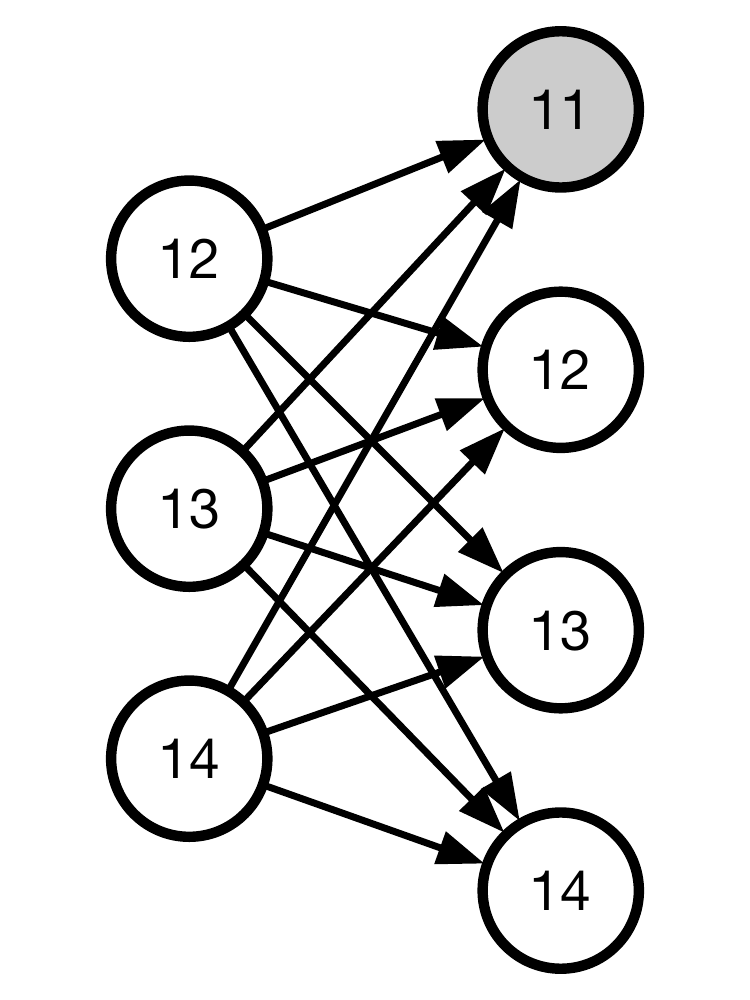}\label{fig:ap-sibship-reporting-withresp}}

\caption{Network reporting when respondents are not included in their
reports. (a) A single sibship. (b) The reporting network for the sibship
in (a), when respondents are not included in reports. (c) The same
sibship as (a) but with each sibling considered connected to herself.
(d) The reporting network for the sibship in (c), when respondents are
included in their reports.}

\label{fig:respondent-inclusion}

\end{figure}

Figure~\ref{fig:respondent-inclusion} illustrates the intuition behind
the results we derive. The Figure shows sibships and reporting networks
under the usual network reporting situation, where respondents are not
included in their reports (panels a and b), and in an alternate
situation, where respondents are included in their reports (panels c and
d). Figure~\ref{fig:respondent-inclusion} illustrates the fact that
whether or not the respondent includes herself in reports will only
affect the visibility of siblings who are in the frame population: the
visibility of node 11, which is dead (and thus not in the frame
population) is unchanged under the two different scenarios. The other
three nodes, on the other hand, have visibility 2 in panel b, when
respondents are not included in their reports; and they have visibility
3 in panel d, when respondents are included in their reports.

\hypertarget{sec:ind-vis}{%
\subsection{Expressions for individual visibility}\label{sec:ind-vis}}

Individual visibility estimation is based on separately adjusting for
the visibility of each person being reported about. Below, in Appendix
\ref{sec:ind-vis-deathrate}, we will use the expressions developed here
to derive estimators for deaths and exposure using an individual
visibility approach.

\textbf{When the respondent is included in the reports}

When the respondent is included in the reports, a reporting identity
says that for person \(i\) in sibship \(\sigma[i]\),

\begin{equation}
\begin{aligned}
v^{\prime}(i, \sigma[i] \cap F) &= y^{\prime+}(\sigma[i] \cap F, i),
\end{aligned}
\label{eq:nrid-includeresp}\end{equation}

where we use the notation \(v^\prime\) for visibilities when respondents
include themselves in reports and \(y^\prime\) for reports when
respondents include themselves.

When reports are perfect, each sibling will be reported once for each
member of the sibship on the frame population, so
Equation~\ref{eq:nrid-includeresp} can be written as

\begin{equation}
\begin{aligned}
v^{\prime}(i, \sigma[i] \cap F) &= y^{\prime+}(\sigma[i] \cap F, i)\\
&= d^{\prime}(\sigma[i] \cap F, i) && \text{(if reporting is perfect)} \\
&= |\sigma \cap F|\\
&= \bar{y}(\sigma[i] \cap F, \sigma[i] \cap F) + 1\\
&= y(j, F) + 1 && \text{for any $j \in \sigma[i] \cap F$}.
\end{aligned}
\label{eq:vis-includeresp-perfect}\end{equation}

Equation~\ref{eq:vis-includeresp-perfect} shows that, when analyzing
reports about individual \(i\) made by respondent \(j\), \(i\)'s
visibility under perfect reporting is \(y(j, F) + 1\).

\textbf{When the respondent is not included in the reports}

When the respondent does not include herself in her reports, a reporting
identity says that, for population member \(i\) in sibship
\(\sigma[i]\),

\begin{equation}
\begin{aligned}
v(i, \sigma[i] \cap F) &= y^{+}(\sigma[i] \cap F, i).
\end{aligned}
\label{eq:nrid}\end{equation}

When reporting is perfect, \(y^{+}(\sigma \cap F, i)\) is the number of
\(i\)'s siblings who are on the frame population. Within a given
sibship, this quantity will depend upon whether or not \(i\) is herself
on the frame population: in a sibship \(\sigma\) with
\(|\sigma \cap F|\) members in \(F\), there will be \(|\sigma \cap F|\)
siblings who can report a sibling \(i\) who is not in \(F\); on the
other hand, when \(i\) is in \(F\), then \(i\) counts towards the size
of \(|\sigma \cap F|\) and so there are \(|\sigma \cap F| - 1\)
\emph{other} siblings who can report \(i\).

Mathematically, when reports are perfect, Equation~\ref{eq:nrid} can be
written as

\begin{equation}
\begin{aligned}
v(i, \sigma \cap F) &= y^{+}(\sigma \cap F,i)\\
&= 
\begin{cases}
| \sigma \cap F |  & \text{if } i \notin F\\
|\sigma \cap F| - 1  & \text{if } \in F
\end{cases}
&& \text{(if reporting is perfect)}\\
&= 
\begin{cases}
\bar{y}(\sigma \cap F, \sigma \cap F) + 1 & \text{if } i \notin F\\
\bar{y}(\sigma \cap F, \sigma \cap F) & \text{if } \in F
\end{cases}\\
&= 
\begin{cases}
y(j, F) + 1 & \text{if } i \notin F\\
y(j, F) & \text{if } \in F
\end{cases}
&&
\text{for any $j \in \sigma[i] \cap F$}.
\end{aligned}
\label{eq:vis-noresp-perfect}\end{equation}

Equation~\ref{eq:vis-noresp-perfect} shows that, when analyzing reports
about individual \(i\) made by respondent \(j\), \(i\)'s visibility
under perfect reporting is either \(y(j, F)\) or \(y(j, F) + 1\),
depending on whether \(i\) is or is not in the frame population.

\hypertarget{individual-visibility-estimands-for-mortality}{%
\subsection{Individual visibility estimands for
mortality}\label{individual-visibility-estimands-for-mortality}}

\hypertarget{individual-visibility-estimation-including-the-respondent}{%
\subsubsection{Individual visibility estimation including the
respondent}\label{individual-visibility-estimation-including-the-respondent}}

When the respondent is included in the sibling reports, and when reports
are perfect, Equation~\ref{eq:vis-includeresp-perfect} shows that the
visibility of each sibling is
\(\bar{y}(\sigma \cap F, \sigma \cap F) + 1\). Note that when reports
are perfect,
\(\bar{y}(\sigma \cap F, \sigma \cap F) = y(i, \sigma \cap F)\) and
\(y^{\prime}(i, \sigma \cap F) = y(i, \sigma \cap F) + 1\) for all
\(i \in \sigma \cap F\).

\textbf{Estimand for deaths using reports that include the respondent}:

\begin{equation}
\begin{aligned}
D^{\prime V}_{\alpha,ind} &= 
\sum_{i \in F} \frac{y^\prime(i, D_\alpha)}{y(i,F) + 1}\\
&= \sum_{i \in F} \frac{y^\prime(i, D_\alpha)}{y^\prime(i,F)} 
\end{aligned}
\label{eq:ind-deaths-withresp}\end{equation}

\textbf{Estimand for exposure using reports that include the
respondent}:

\begin{equation}
\begin{aligned}
N^{\prime V}_{\alpha,ind} &= 
\sum_{i \in F} \frac{y^\prime(i, N_\alpha)}{y(i,F) + 1}\\
&= \sum_{i \in F} \frac{y^\prime(i, N_\alpha)}{y^\prime(i,F)} 
\end{aligned}
\label{eq:ind-exposure-withresp}\end{equation}

\textbf{Estimand for the death rate using reports that include the
respondent}:

\begin{equation}
\begin{aligned}
M^{\prime V}_{\alpha,ind}
&= \frac{D^{\prime V}_{\alpha,ind}}{N^{\prime V}_{\alpha,ind}}\\
&= \frac{
\sum_{i \in F} \frac{y^\prime(i, D_\alpha)}{y^\prime(i,F)}
}{
\sum_{i \in F} \frac{y^\prime(i, N_\alpha)}{y^\prime(i,F)}
}
\end{aligned}
\label{eq:ind-rate-withresp}\end{equation}

\hypertarget{individual-visibility-estimands-not-including-the-respondent}{%
\subsubsection{Individual visibility estimands not including the
respondent}\label{individual-visibility-estimands-not-including-the-respondent}}

When the respondent does not include herself in reports about deaths and
exposure, and when reports are perfect,
Equation~\ref{eq:vis-noresp-perfect} shows that the visibility of a
member \(i\) of a sibship \(\sigma\) reported by respondent
\(j \in \sigma[i] \cap F\) is given by:

\[
\begin{aligned}
v(i, \sigma \cap F) 
&= 
\begin{cases}
| \sigma \cap F |  = \bar{y}(\sigma \cap F, \sigma \cap F) + 1 & \text{if } i \notin F\\
|\sigma \cap F| - 1 = \bar{y}(\sigma \cap F, \sigma \cap F) & \text{if } \in F
\end{cases}
&& \text{(if reporting is perfect)}.
\end{aligned}
\]

Since \(\bar{y}(\sigma \cap F, \sigma \cap F) = y(j, \sigma \cap F)\)
for each \(j \in \sigma \cap F\) when reports are perfect, we use the
reported quantities \(y(j, \sigma \cap F)\) and
\(1 + y(j, \sigma \cap F)\) to estimate visibilities for siblings in the
frame and not in the frame, respectively.

\textbf{Estimand for deaths using reports that don't include the
respondent}:

\begin{equation}
\begin{aligned}
D^V_{\alpha,ind} 
&= 
\sum_{i \in F} \left[
\underbrace{
\frac{y(i, D_\alpha \cap F)}{y(i,F)} 
}_{\text{$D_\alpha \cap F = \emptyset$}}
+ 
\underbrace{
\frac{y(i, D_\alpha -F)}{y(i,F) + 1}
}_{\text{siblings not in $F$}}
\right] && \text{(when reports are perfect)}\\
&=
\sum_{i \in F} \frac{y(i, D_\alpha)}{y(i,F) + 1}.
\end{aligned}
\label{eq:ind-deaths-noresp}\end{equation}

In going from the first line to the second, we make use of the fact that
deaths cannot be on the frame population, so that
\(y(i, D_\alpha \cap F) = 0\) for all \(i\), and so
\(y(i, D_\alpha) = y(i, D_\alpha -F)\) for all \(i\).

\textbf{Estimand for exposure using reports that don't include the
respondent}:

\begin{equation}
\begin{aligned}
N^V_{\alpha,ind} &= 
\sum_{i \in F} \left[
\underbrace{
\frac{y(i, N_\alpha \cap F)}{y(i,F)} 
}_{\text{siblings in $F$}}
+ 
\underbrace{
\frac{y(i, N_\alpha -F)}{y(i,F) + 1}
}_{\text{siblings not in $F$}}
\right] &&
\text{(when reports are perfect)}.
\end{aligned}
\label{eq:ind-exposure-noresp}\end{equation}

\textbf{Estimand for the death rate using reports that don't include the
respondent}:

The estimand for the death rate is the ratio of the estimand for the
number of deaths and the estimand for the exposure:

\begin{equation}
\begin{aligned}
M^V_{\alpha,ind}  
&= \frac{
\sum_{i \in F} \frac{y(i, D_\alpha)}{y(i,F) + 1}.
}{
\sum_{i \in F} \left[
\frac{y(i, N_\alpha \cap F)}{y(i,F)} 
+ 
\frac{y(i, N_\alpha -F)}{y(i,F) + 1}
\right].
}
\end{aligned}
\label{eq:ind-rate-noresp}\end{equation}

\hypertarget{sec:agg-vis}{%
\subsection{Expressions for aggregate visibility}\label{sec:agg-vis}}

Aggregate visibility estimation is based on adjusting for the average
visibility of the people being reported about. In this section, we
develop expressions for the average visibility of siblings who are in
some group \(A \subset U\). By deriving results for a general group
\(A\), we will obtain expressions that can be readily used for both
deaths and exposure (see Appendix \ref{sec:agg-vis-deathrate}).

\textbf{When the respondent is included in the reports}

When reports include the respondent, the average visibility of siblings
in a group \(A\) is:

\begin{equation}
\begin{aligned}
\bar{v}^\prime(A, F)
&= |A|^{-1} \sum_{i \in A} v^{\prime}(i, F)\\
&= |A|^{-1} \sum_{i \in A} y^{\prime+}(F,i).
\end{aligned}
\label{eq:vis-includeresp-agg}\end{equation}

When reporting is perfect, this becomes

\begin{equation}
\begin{aligned}
\bar{v}^\prime(A, F)
&= |A|^{-1} \sum_{i \in A} y^{\prime+}(F,i)\\
&= |A|^{-1} \sum_{i \in A} 
\left[ 
\bar{y}(\sigma[i] \cap F, \sigma[i] \cap F) + 1
\right] &&
\text{(if reporting is perfect)}\\
&= 
|A|^{-1} 
\sum_{\sigma \in \Sigma} \sum_{i \in \sigma \cap A} 
|\sigma \cap F|\\
&= 
|A|^{-1} \sum_{\sigma \in \Sigma} 
\left[ 
|\sigma \cap A|~|\sigma \cap F| 
\right]\\
&= 
\sum_{\sigma \in \Sigma} 
\underbrace{
  \frac{|\sigma \cap A|}{|A|}
  }_{
  \substack{\text{fraction of} \\ \text{$A$ in} \\ \text{in sibship $\sigma$}}
  } 
\underbrace{
  |\sigma \cap F|
}_{
  \substack{\text{\# sibship} \\ \text{members} \\ \text{on frame}}
}.
%~= \sum_{i \in F} \frac{|\sigma[i] \cap A|}{|A|}.
\end{aligned}
\label{eq:vis-includeresp-agg-perfect}\end{equation}

Equation~\ref{eq:vis-includeresp-agg-perfect} shows that the visibility
can be expressed as a weighted average of the number of siblings on the
frame across all of the sibships in the population, where the weights
are given by the proportion of \(A\) in each sibship.

Unlike the results for individual visibility estimation, we see no way
to convert Equation~\ref{eq:vis-includeresp-agg-perfect} into a
sample-based estimator for visibility. As we will see below, when
estimating a death rate using an aggregate visibility estimator, we will
take the ratio of two aggregate visibility estimators. Under the
assumption that the visibility of the numerator and denominator are the
same the visibilities will cancel, so that aggregate visibility does not
need to be directly estimated.

\textbf{When the respondent is not included in the reports}

When reports do not include the respondent, the average visibility of
siblings in a group \(A\) is: \[
\begin{aligned}
\bar{v}(A, F) &= |A|^{-1} \sum_{i \in A} v(i,F)\\
&= |A|^{-1} \sum_{i \in A} y^{+}(F,i)\\
&= |A|^{-1} \sum_{\sigma \in \Sigma} \sum_{i \in \sigma \cap A} y^{+}(F,i).
\end{aligned}
\]

When reporting is perfect, this becomes

\begin{equation}
\begin{aligned}
\bar{v}(A, F)  
&= |A|^{-1} \sum_{\sigma \in \Sigma} \sum_{i \in \sigma \cap A} y^{+}(F,i)\\
&= |A|^{-1} 
\sum_{\sigma \in \Sigma}
\left[ 
\underbrace{
\sum_{i \in \sigma \cap A \cap F} \bar{y}(\sigma[i] \cap F, \sigma[i] \cap F)
}_{\text{members of $A$ on frame}}
+
\underbrace{
\sum_{i \in \sigma \cap A -F} (\bar{y}(\sigma[i] \cap F, \sigma[i] \cap F) + 1)
}_{\text{members of $A$ not on frame}}
\right] &&
\text{(if reporting is perfect)}\\
&= |A|^{-1} 
\sum_{\sigma \in \Sigma}
\left[ 
\sum_{i \in \sigma \cap A} \bar{y}(\sigma[i] \cap F, \sigma[i] \cap F)
+
\sum_{i \in \sigma \cap A -F} 1
\right]\\
&= 
  \sum_{\sigma \in \Sigma}
  \frac{|\sigma \cap A|}{|A|} 
  \bar{y}(\sigma \cap F, \sigma \cap F) +
  \sum_{\sigma \in \Sigma} \frac{|\sigma \cap A -F|}{|A|}
\\
&=
  \sum_{\sigma \in \Sigma}
  \frac{|\sigma \cap A|}{|A|}
  (|\sigma \cap F| - 1) +
  \sum_{\sigma \in \Sigma}
  \frac{|\sigma \cap A -F|}{|A|} \\
  &=
  \sum_{\sigma \in \Sigma}
  \left[
  \frac{|\sigma \cap A|}{|A|}
  |\sigma \cap F| - 
  \frac{|\sigma \cap A|}{|A|} +
  \frac{|\sigma \cap A -F|}{|A|}
  \right]
  \\
  &=
  \sum_{\sigma \in \Sigma}
  \frac{|\sigma \cap A|}{|A|}
  \left[
  |\sigma \cap F| - \frac{|\sigma \cap F \cap A|}{|\sigma \cap A|}
  \right].
\end{aligned}
\label{eq:vis-noresp-agg-perfect}\end{equation}

Comparing Equation~\ref{eq:vis-noresp-agg-perfect} to the analogous
expression we derived for the situation where respondents do include
themselves in reports (Equation~\ref{eq:vis-includeresp-agg-perfect}),
we find that the two expressions differ according to the second factor
in the sum: when respondents include themselves in reports, this second
factor is always the number of siblings in the frame population,
\(|\sigma \cap F|\); when respondents do not include themselves, this
second factor is the number of siblings in the frame population minus
the proportion of siblings in \(A\) that is on the frame population,
i.e., \(|\sigma \cap F| - |\sigma \cap F \cap A|/|\sigma \cap A|\).

Note also that, by using the average visibility we derived for reports
where respondents include themselves
(Equation~\ref{eq:vis-includeresp-agg-perfect}), we have an alternate
way to write the result in Equation~\ref{eq:vis-noresp-agg-perfect}:

\begin{equation}
\begin{aligned}
\bar{v}(A, F)  
  &=
\sum_{\sigma \in \Sigma}
\frac{|\sigma \cap A|}{|A|}
\left[
|\sigma \cap F| - \frac{|\sigma \cap F \cap A|}{|\sigma \cap A|}
\right]\\
  &=
\bar{v}^{\prime}(A,F) - \frac{|F \cap A|}{|A|}.
\end{aligned}
\label{eq:vis-noresp-agg-perfect-v2}\end{equation}

Equation~\ref{eq:vis-noresp-agg-perfect-v2} shows that the aggregate
visibility when respondents are not included in reports is equal to the
aggregate visibility when respondents are included in reports, minus the
fraction of the set \(A\) that is in the frame population.

\hypertarget{aggregate-visibility-estimands-for-mortality}{%
\subsection{Aggregate visibility estimands for
mortality}\label{aggregate-visibility-estimands-for-mortality}}

In Section~\ref{sec:agg-vis}, we saw that the expressions for the
average visibility did not readily lend themselves to forming
estimators. In the situation where we wish to estimate death rates,
however, the condition that visibility is the same for deaths and for
exposure leads to an estimator where the visibilities cancel, meaning
that they do not have to be directly estimated. We can then investigate
the sensitivity of estimated death rates to different visibilities as
part of the broader sensitivity framework
(Section~\ref{sec:agg-vis-deathrate}).

For the estimands below, we write aggregate visibilities for a set \(A\)
as \(\bar{v}(A,F)\) and \(\bar{v}^{\prime}(A,F)\), bearing in mind that
this cancellation will take place in the estimand for the death rate.

\hypertarget{aggregate-visibility-estimands-not-including-the-respondent}{%
\subsubsection{Aggregate visibility estimands not including the
respondent}\label{aggregate-visibility-estimands-not-including-the-respondent}}

\textbf{Estimand for deaths using reports that don't include the
respondent}

\begin{equation}
\begin{aligned}
D^V_{\alpha,agg} &= 
\frac{y_{F, D_\alpha}}{\bar{v}_{D_\alpha,F}}.
\end{aligned}
\label{eq:agg-deaths-noresp}\end{equation}

\textbf{Estimand for exposure using reports that don't include the
respondent}

\begin{equation}
\begin{aligned}
N^V_{\alpha,agg} &= 
\frac{y_{F, N_\alpha}}{\bar{v}_{N_\alpha,F}}.
\end{aligned}
\label{eq:agg-exposure-noresp}\end{equation}

\textbf{Estimand for the death rate using reports that don't include the
respondent}

\begin{equation}
\begin{aligned}
M^V_{\alpha,agg} &= 
\frac{D^V_{\alpha,agg}}{N^V_{\alpha,agg}}\\
&= \frac{y_{F,D_\alpha}}{\bar{v}_{D_\alpha, F}}~\frac{\bar{v}_{N_\alpha,F}}{y_{F,N_\alpha}}\\
&= \frac{y_{F, D_\alpha}}{y_{F, N_\alpha}}~\frac{\bar{v}_{N_\alpha,F}}{\bar{v}_{D_\alpha,F}}.
\end{aligned}
\label{eq:agg-rate-noresp}\end{equation}

\hypertarget{aggregate-visibility-estimands-including-the-respondent}{%
\subsubsection{Aggregate visibility estimands including the
respondent}\label{aggregate-visibility-estimands-including-the-respondent}}

\textbf{Estimand for deaths using reports that include the respondent}

\begin{equation}
\begin{aligned}
D^{\prime V}_{\alpha,agg} &= 
\frac{y^\prime_{F, D_\alpha}}{\bar{v}^\prime_{D_\alpha,F}}.
\end{aligned}
\label{eq:agg-deaths-withresp}\end{equation}

\textbf{Estimand for exposure using reports that include the
respondents}

\begin{equation}
\begin{aligned}
N^{\prime V}_{\alpha,agg} &= 
\frac{y^\prime_{F, N_\alpha}}{\bar{v}^\prime_{N_\alpha,F}}.
\end{aligned}
\label{eq:agg-exposure-withresp}\end{equation}

\textbf{Estimand for the death rate using reports that include the
respondent}

\begin{equation}
\begin{aligned}
M^{\prime V}_{\alpha,agg} 
&= \frac{D^{\prime V}_{\alpha,agg}}{N^{\prime V}_{\alpha,agg}}\\
&= \frac{y^\prime_{F,D_\alpha}}{\bar{v}^\prime_{D_\alpha, F}}~\frac{\bar{v}^\prime_{N_\alpha,F}}{y^\prime_{F,N_\alpha}}\\
&= \frac{y^\prime_{F, D_\alpha}}{y^\prime_{F, N_\alpha}}~\frac{\bar{v}^\prime_{N_\alpha,F}}{\bar{v}^\prime_{D_\alpha,F}}.
\end{aligned}
\label{eq:agg-rate-withresp}\end{equation}

\hypertarget{sensitivity-to-invisible-deaths-1}{%
\section{Sensitivity to invisible
deaths}\label{sensitivity-to-invisible-deaths-1}}

Reports about siblings can only tell us about the \emph{visible}
population -- i.e., the group of people who have siblings on the frame
population who can provide information about their survival. In this
Appendix, we develop expressions that relate the death rate in the
visible population to the death rate in the entire population. This
expression will help researchers understand how different death rates in
the visible population can be expected to be from death rates in the
entire population.

In order to analyze the sensitivity of sibling survival estimates to
invisible deaths, we need to develop notation that can be used to
distinguish between visible and invisible deaths. For a demographic
group \(\alpha\) (for example, women aged 15-25 in 2018), let
\begin{equation}
\begin{aligned}
p^V_{D_\alpha} &= \frac{D^V_{\alpha}}{D^V_{\alpha} + D^I_{\alpha}}, && \text{be the fraction of deaths that is visible;}\\
p^V_{N_\alpha} &= \frac{N^V_{\alpha}}{N^V_{\alpha} + N^I_{\alpha}}, && \text{be the fraction of exposure that is visible.} 
\end{aligned}
\label{eq:pvisdeath}\end{equation} We define analogous quantities for
the fraction of deaths and exposure that is invisible,
\(p^I_{D_\alpha}\) and \(p^I_{N_\alpha}\).

Note that \begin{equation}
\begin{aligned}
\frac{p^V_{D_\alpha}}{p^V_{N_\alpha}} &=
\frac{D^V_{\alpha}}{D^V_{\alpha} + D^I_{\alpha}} \times
\frac{N^V_{\alpha} + N^I_{\alpha}}{N^V_{\alpha}}\\
&= \frac{M^V_{\alpha}}{M_{\alpha}}.
\end{aligned}
\label{eq:p-d-over-p-a}\end{equation}

Thus, the ratio of the fraction of deaths that is visible to the
fraction of exposure that is visible is equal to the ratio of the
visible death rate to the total death rate.

Result \ref{res:agghm} shows that the total death rate \(M_\alpha\) can
be understood as a weighted harmonic mean of the invisible death rate
\(M^I_\alpha\) and the visible death rate \(M^V_{\alpha}\), where the
weights are given by the number of visible and invisible deaths. We now
use this insight to develop Result \ref{res:sens-invis}, which helps us
understand the formal relationship between the invisible death rate, the
visible death rate, and the total death rate.

~

\begin{Result}

\label{res:sens-invis}

Suppose that, for a demographic group \(\alpha\), the invisible death
rate (\(M^I_\alpha\)) and the visible death rate (\(M^V_\alpha\)) differ
by a factor of \(K\), so that \begin{equation}
M^I_\alpha = K M^V_\alpha
\label{eq:kfactor}\end{equation} for \(K > 0\). Then \begin{equation}
\begin{aligned}
M_\alpha &= M_\alpha^V \left[ \frac{K}{p^I_{D_{\alpha}} + K (1-p^I_{D_{\alpha}})} \right],
\end{aligned}
\label{eq:invistotaldiff}\end{equation} where \(p^I_{D_{\alpha}}\) is
the proportion of deaths that is invisible.

\end{Result}

\begin{proof}

Using the fact that \(M_\alpha\) is the weighted harmonic mean of
\(M^I_\alpha\) and \(M^V_\alpha\), we find \[
\begin{aligned}
M_\alpha &= \left[ \frac{p^I_{D_{\alpha}}}{M^I_\alpha} + \frac{p^V_{D_\alpha}}{M^V_\alpha} \right]^{-1}\\
&= \left[ \frac{p^I_{D_{\alpha}}}{K M^V_\alpha} + \frac{p^V_{D_\alpha}}{M^V_\alpha} \right]^{-1}\\
&= \left[ \frac{p^I_{D_{\alpha}} + K p^V_{D_\alpha}}{K M^V_\alpha} \right]^{-1}\\
&= M_\alpha^V \left[ \frac{K}{p^I_{D_{\alpha}} + K p^V_{D_\alpha}} \right]\\
&= M_\alpha^V \left[ \frac{K}{p^I_{D_{\alpha}} + K (1-p^I_{D_{\alpha}})} \right].
\end{aligned}
\]

\end{proof}

Result \ref{res:sens-invis} reveals that there is a relationship between
between \(K\), the difference between the visible and invisible death
rates, and \(p^I_{D_{\alpha}}\), which is related to the number of
invisible deaths relative to the number of visible deaths. Equation
\ref{eq:invistotaldiff} shows that

\begin{itemize}
\tightlist
\item
  when \(K = 1\), \(M_\alpha^V = M_\alpha\)
\item
  when \(p^I_{D_{\alpha}} = 0\), \(p^V_{D_\alpha}=1\) and so
  \(M_\alpha^V = M_\alpha\)
\end{itemize}

It can also be helpful to use Result \ref{res:sens-invis} to obtain an
expression for the relative error that would follow from using the
visible death rate \(M^V_\alpha\) as an estimate of the total death rate
\(M_\alpha\):

\begin{equation}
\begin{aligned}
\frac{M^V_\alpha - M_\alpha}{M_\alpha} &= \frac{M^V_\alpha}{M_\alpha} - 1\\
&= \frac{p^V_{D_{\alpha}} + K(1 - p^V_{D_{\alpha}})}{K} - 1\\
&= p^I_{D_\alpha} \left(\frac{1-K}{K}\right).
\end{aligned}
\label{eq:sens-invis-relerr}\end{equation}

In order to further develop intuition about how large we might expect
biases due to invisible deaths to be, we can investigate different
scenarios. For example, suppose that 10\% of deaths are invisible, and
the death rate is 20\% higher among the invisible population than among
the visible population. Then \(K = 1.2\), \(p^I_{D_{\alpha}} = 0.1\),
and the relative error calculated from Equation
\ref{eq:sens-invis-relerr} is about -.017; in other words, in this
scenario, death rate estimates based on the visible population alone
will be too low by about 1.7 percent.

Figure~\ref{fig:invis-death-sens} illustrates this relative error for a
range of values of \(K\) and \(p^I_{D_\alpha}\).

\begin{figure}
\hypertarget{fig:invis-death-sens}{%
\centering
\includegraphics[width=\textwidth,height=0.5\textheight]{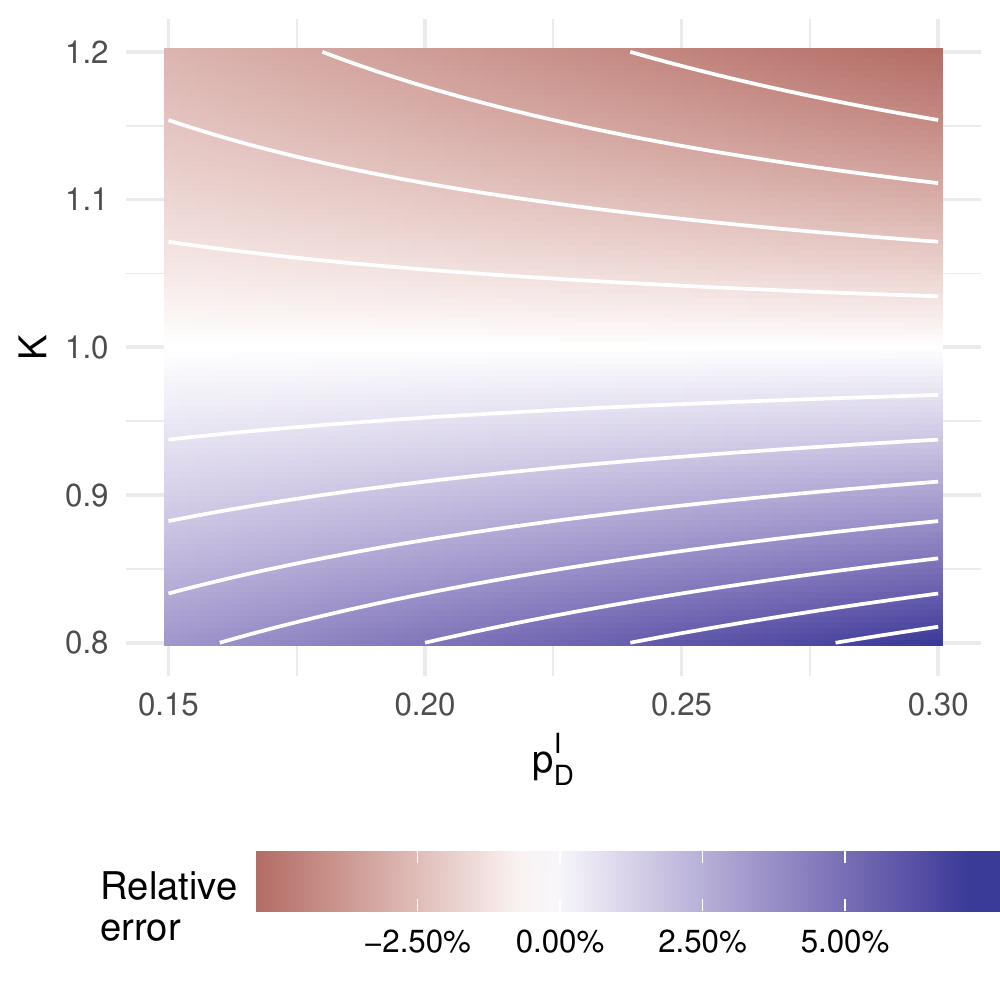}
\caption{Illustration of the relative error in using the visible death
rate \(M^V\) as an estimate for the total death rate \(M\). The
proportion of deaths that is invisible, \(p^I_{D_{\alpha}}\), varies
along the x axis; the relationship between the visible and invisible
death rates, captured by the parameter \(K\)
(Equation~\ref{eq:kfactor}), varies along the y axis. The colors show
the percentage relative error; so if 20\% of deaths are invisible
(\(p^I_{D_\alpha} = 0.2\)) and the invisible death rate is 10\% higher
than the visible death rate (\(K=1.1\)), the relative error is about 2
percent. Relative error increases as \(K\) gets farther away from 1 and
as \(p^V_{D_{\alpha}}\) increases.}\label{fig:invis-death-sens}
}
\end{figure}

Next, Result \ref{res:sens-invis-alt} provides a second expression that
analyzes the formal relationship between the invisible death rate, the
visible death rate, and the total death rate; this second result is
parameterized in terms of \(p^I_{N_\alpha}\), the proportion of exposure
that is invisible. This is the relationship used in the main text.

~

\begin{Result}

\label{res:sens-invis-alt}

Suppose that, for a demographic group \(\alpha\), the invisible death
rate (\(M^I_\alpha\)) and the visible death rate (\(M^V_\alpha\)) differ
by a factor of \(K\), so that \begin{equation}
M^I_\alpha = K M^V_\alpha
\label{eq:kfactor-alt}\end{equation} for \(K > 0\). Then
\begin{equation}
\begin{aligned}
M_\alpha &= M_\alpha^V \left[ 1 + p^I_{N_\alpha}(K-1) \right],
\end{aligned}
\label{eq:invistotaldiff-alt}\end{equation} where \(p^I_{N_{\alpha}}\)
is the proportion of exposure that is invisible.

\end{Result}

\begin{proof}

Using the fact that \(M_\alpha\) is the weighted arithmetic mean of
\(M^I_\alpha\) and \(M^V_\alpha\), we find \[
\begin{aligned}
M_\alpha &= p^V_{N_\alpha} M^V_{\alpha} + p^I_{N_\alpha} M^I_{\alpha}\\
&= (1 - p^I_{N_\alpha}) M^V_{\alpha} + p^I_{N_\alpha} K M^V_{\alpha}\\
&= M^V_{\alpha} - M^V_{\alpha} p^I_{N_\alpha} + K p^I_{N_\alpha} M^V_{\alpha}\\
&= M^V_{\alpha} \left[1 + p^I_{N_\alpha} (K-1) \right].
\end{aligned}
\]

\end{proof}

Again, it can be helpful to use Result \ref{res:sens-invis-alt} to
obtain an expression for the relative error that would follow from using
the visible death rate \(M^V_\alpha\) as an estimate of the total death
rate \(M_\alpha\):

\begin{equation}
\begin{aligned}
\frac{M^V_\alpha - M_\alpha}{M_\alpha} &= 
\frac{M^V_\alpha}{M_\alpha} - 1\\
&= \frac{1}{1 + p^I_{N_\alpha}(K-1)} - 1\\
&= \frac{p^I_{N_\alpha}(1-K)}{1 - p^I_{N_\alpha}(1-K)}.
\end{aligned}
\label{eq:sens-invis-alt-relerr}\end{equation}

We can further develop intuition about how large we might expect biases
due to invisible deaths to be, we can investigate different scenarios.
For example, suppose that 10\% of exposure is invisible, and the death
rate is 20\% higher among the invisible population than among the
visible population. Then \(K = 1.2\), \(p^I_{N_{\alpha}} = 0.1\), and
the relative error from Equation~\ref{eq:sens-invis-alt-relerr} is about
-0.019; in other words, in this scenario, death rate estimates based on
the visible population alone will be too low by about 1.9 percent.

Figure~\ref{fig:invis-sens-surface-alt} illustrates this second
expression for the relative error over a range of values of \(K\) and
\(p^I_{N_\alpha}\).

\hypertarget{sec:sampling}{%
\section{Sampling}\label{sec:sampling}}

In Appendix \ref{sec:estimands}, we developed several estimands based on
sibling reports. These estimands describe quantities that could be
estimated from a census of the frame population. In practice,
researchers do not have a census of the frame population, but rather a
sample from the frame population. This section develops some results
that will be helpful in understanding how to develop sample-based
estimators for the estimands in Appendix \ref{sec:estimands}.

\hypertarget{sec:sampling-setup}{%
\subsection{Sampling setup}\label{sec:sampling-setup}}

We use the design-based sampling framework described in Sarndal,
Swensson, and Wretman (2003), repeating a few key definitions here for
convenience. We assume we have a probability sample \(s\) from a frame
population \(F\); common frame populations include all adults, all
adults aged 15-59, and in many DHS surveys, all women aged 15-59. The
random variable \(I_i\) takes the value 1 when \(i \in F\) is included
in the sample, and 0 otherwise. Each \(i \in F\) has a nonzero
\emph{probability of inclusion} \(\pi_i =  \mathbb{E} [I_i]\) and the
sampling weights are given by \(w_i = \frac{1}{\pi_i}\).

Suppose some quantity \(y_i\) is defined for every \(i \in F\). Then the
\emph{Horvitz Thompson estimator} for the population total
\(Y = \sum_{i \in F} y_i\) from a probability sample \(s\) is given by

\[
\widehat{Y} = \sum_{i \in s} w_i y_i.
\]

Sarndal, Swensson, and Wretman (2003) shows that Horvitz-Thompson
estimators are consistent and unbiased\footnote{In this paper, we use
  the framework of design-based sampling, so the properties of
  estimators -- such as unbiasedness and consistency -- are with respect
  to the probability sampling mechanism. There are many types of
  consistency; we refer in this work to design-consistency, also called
  Fisher consistency.}, a fact that will be useful below.

~

\begin{Result}

\label{res:fisher} Suppose a Horvitz-Thompson estimator \[
\widehat{Y}^{\text{HT}} = \sum_{i \in s} w_i y_i
\] is design-unbiased for a total \(Y\). Then
\(\widehat{Y}^{\text{HT}}\) is also (design) consistent for \(Y\).

\end{Result}

\begin{proof}

This result follows from taking the sampling design to assign
\(\pi_i = 1\) for all \(i \in F\). We then have \(s = F\) and
\(w_i = 1\) for all \(i\). Since the estimator is unbiased, design
consistency follows.

\end{proof}

~

Next, we state a Result that is helpful when devising estimators that
are ratios of other estimators.

~

\begin{Result}

\label{res:compound-ratio} Suppose that
\(\widehat{y}_1, \dots, \widehat{y}_n\) are estimators that are
consistent and unbiased for \(Y_1, \dots, Y_n\) respectively. Then the
compound ratio estimator

\[
\widehat{R} = \frac{\widehat{y}_1 \dots \widehat{y}_k}{\widehat{y}_{k+1} \dots \widehat{y}_{n}}.
\]

is consistent and essentially unbiased for
\(R = (Y_1 \dots Y_{k}) / (Y_{k+1} \dots Y_n)\).

\end{Result}

\begin{proof}

See Rao and Pereira (1968), Wolter (2007) (pg. 233), and Feehan and
Salganik (2016a) for more details.

\end{proof}

The references that derive the compound ratio estimator in Result
\ref{res:compound-ratio} show that, technically, the estimator is
biased. However, the bias has been shown to be of order \(O(n^{-1})\);
that is, the bias goes to 0 as the sample size \(n \to \infty\).
Further, Feehan and Salganik (2016a) investigated compound ratio bias in
network reporting estimators and found suggestive evidence that it is
likely to be very small. This finding is consistent with the literature
on ratio estimators, which has revealed that bias is often very small.
Thus, in these technical results, we refer to estimators whose only bias
arises from being compound ratios -- i.e., whose bias is of order
\(O(n^{-1})\) -- as \emph{essentially unbiased}.

\hypertarget{adjusting-for-visibility-in-samples}{%
\subsection{Adjusting for visibility in
samples}\label{adjusting-for-visibility-in-samples}}

In this section, we briefly state key results that explain which
conditions are required for individual and aggregate visibility
estimators to be consistent and unbiased. More detailed derivations of
these results can be found elsewhere.

\textbf{Individual visibility estimation from a probability sample}

~

\begin{Definition}

\label{res:def-ind-vis} The \textbf{individual visibility} estimator for
a total \(\boldsymbol{Y} = \sum_{i \in U} y_i\) is defined to be
\begin{equation}
\widehat{Y}^\text{ind} = \sum_{i \in s} w_i \sum_{j \sim i} \frac{y_j}{v_{j,F}},
\label{eq:ind-mult-def}\end{equation} where \(j \sim i\) indexes the
neighbors \(j\) of respondent \(i\), \(w_i\) is the design weight for
respondent \(i\), and \(v_{j,F}\) is the visibility of person \(j\) to
the frame population \(F\). In the situation where there are no errors
in reporting, the visibility \(v_{j,F}\) has also been called the
multiplicity of person \(j\) (Sirken 1970).

\end{Definition}

~

\begin{Result}

\label{res:ind-vis} Suppose that there are no false positive reports.
The individual visibility estimator is consistent and unbiased for the
total \(\boldsymbol{Y}\).

\end{Result}

\begin{proof}

First, we show that the estimator is unbiased. If \(\pi_i\) is \(i\)'s
probability of inclusion under the sampling design, then the design
weights are \(w_i = \frac{1}{\pi_i}\). Thus, we have \[
\begin{aligned}
\widehat{\boldsymbol{Y}}^{\text{ind}} 
&= \sum_{i \in s} \frac{1}{\pi_i}  \sum_{j \sim i} \frac{y_j}{v_{j,F}}\\
\Longleftrightarrow  \mathbb{E} [\widehat{\boldsymbol{Y}}^{\text{ind}}]
&= \sum_{i \in F}  \mathbb{E} [I_i] \frac{1}{\pi_i}  \sum_{j \sim i} \frac{y_j}{v_{j,F}}\\
&= \sum_{i \in F} \sum_{j \sim i} \frac{y_j}{v_{j,F}}\\
&= \sum_{l \in F} v_{l,F} \frac{y_l}{v_{l,F}}\\
&= \sum_{l \in F} y_l.
\end{aligned}
\] The last step follows because, as long as there are no false positive
reports, in a census of \(F\), each unit \(j\) appears once for each
time it is visible to \(F\); that is, \(v_{j,F} = y_{F,j}\) (see Feehan
and Salganik (2016a) for details).

Note that the derivation above reveals that the individual visibility
estimator can be written as a Horvitz-Thompson estimator (Thompson 2002;
Sirken 1970): to see how, define
\(z_i = \sum_{j \sim i} \frac{y_j}{v_{j,F}}\) for all \(i \in F\). The
individual visibility estimator in Equation~\ref{eq:ind-mult-def} then
becomes \(\widehat{Y} = \sum_{i \in s} w_i z_i\). Since
\(\widehat{\boldsymbol{Y}}^{\text{ind}}\) can be written as a
Horvitz-Thompson estimator, Result \ref{res:fisher} shows that
unbiasedness implies consistency.

\end{proof}

Good references for individual visibility estimators include Sirken
(1970) and Lavallee (2007).

~

\textbf{Aggregate visibility estimation from a probability sample}

~

\begin{Definition}

\label{res:def-agg-vis} The \textbf{aggregate visibility} estimator for
a total \(\boldsymbol{Y} = \sum_{i \in U} y_i\) is defined to be \[
\widehat{Y}^{\text{agg}} = \frac{\sum_{i \in s} w_i \sum_{j \sim i} y_j}{\widehat{\bar{v}}_{Y,F}},
\] where \(\bar{v}_{Y,F} = N^{-1} \sum_{l \in U} v_{l,F}\) is the
average of the individual visibilities of each person who could be
reported about in a census of \(F\).

\end{Definition}

~

The network scale-up estimator, the network survival estimator, and
related approaches are examples of estimators based on the idea of
aggregate visibility (Bernard et al. 1989; Feehan and Salganik 2016a;
Feehan, Mahy, and Salganik 2017).

~

\begin{Result}

\label{res:agg-vis} Suppose that there are no false positive reports.
Then the aggregate visibility estimator is consistent and essentially
unbiased for the total \(\boldsymbol{Y}\).

\end{Result}

\begin{proof}

See Feehan and Salganik (2016a).

\end{proof}

\hypertarget{sec:agg-vis-deathrate}{%
\section{Estimating death rates using aggregate
visibility}\label{sec:agg-vis-deathrate}}

This section presents results for death rate estimators that adjust for
visibility at the aggregate level; that is, the sum of reported
connections to siblings across respondents is adjusted for using the
average visibility of the reported siblings.

Section~\ref{sec:agg-vis} showed that there is no one obvious way to
estimate the aggregate visibility of deaths or exposure using data
collected from the sibling histories. However, if the visibility of
deaths is the same as the visibility of exposure, then the death rate
can be estimated without directly estimating the visibility of deaths or
exposure, since these two quantities will cancel. This condition is
analogous to the requirement, discussed at length in the sibling
survival literature, that there be no correlation between sibship size
and mortality. We suggest that researchers assess the sensitivity of
aggregate visibility estimators to this condition using the sensitivity
framework we develop below.

We will first provide an estimator that is based on excluding
information about the survey respondents and only using reports about
respondents' siblings. We will see that this approach produces an
estimator that is essentially the one recommended in the DHS program's
official documentation (Rutstein and Guillermo Rojas 2006, pg. 156) and
by other researchers (e.g., Masquelier 2013). An additional estimator,
based on including information about the survey respondents, turns out
to be very similar, and follows as a Corollary of the derivation of the
without-respondent estimator.

The relationship derived in Equation~\ref{eq:agg-rate-noresp} suggests

\begin{equation}
\begin{aligned}
\widehat{M}^V_\alpha 
&= \frac{\widehat{D}^V_\alpha}{\widehat{N}^V_\alpha}\\
&= \frac{\sum_{i \in s} w_i~y(i, D_\alpha)}{\sum_{i \in s} w_i~y(i, N_\alpha)}.
\end{aligned}
\label{eq:agg-vis-noresp-estimator}\end{equation}

Equation~\ref{eq:agg-vis-noresp-estimator} is based on the idea that we
can plug in sample-based aggregate visibility estimators for the number
of siblings who died and the number of siblings who contribute exposure.
Result \ref{res:mvis-agg-estimator} formally states the conditions are
required for the estimator to produce consistent and unbiased estimates
for the visible death rate.

~

\begin{Result}

\label{res:mvis-agg-estimator} \textbf{(Aggregate visibility estimation
for death rates)} Suppose we have a probability sample \(s \subset F\)
and the associated weights \(w_i\) for all \(i \in s\). Suppose that, in
aggregate, there are no false positive reports about deaths or exposure,
so that \(y(F, D_\alpha) = y^{+}(F, D_\alpha)\) and
\(y(F, N_\alpha) = y^{+}(F, N_\alpha)\). Suppose also that the
visibility of deaths and exposure to the frame population are the same,
and nonzero, so that
\(\bar{v}(D_\alpha, F) = \bar{v}(N_\alpha, F) > 0\). Then
\begin{equation}
\begin{aligned}
\widehat{M}^{V} 
&= \frac{\widehat{y}(F, D_\alpha)}{\widehat{y}(F, N_\alpha)}\\
&= \frac{\sum_{i \in s} w_i~y(i, D_\alpha)}{\sum_{i \in s} w_i~y(i, N_\alpha)}
\end{aligned}
\label{eq:res-mvis-agg-estimator}\end{equation} is consistent and
essentially unbiased for
\(M^V_\alpha = \frac{D^{V}_\alpha}{N^{V}_\alpha}\), where the exposure
\(N^{V}_\alpha\) is approximated by the number of visible people, living
or dead, in group \(\alpha\).

\end{Result}

\begin{proof}

The numerator and denominator are Horvitz-Thompson estimators, and are
therefore consistent and unbiased for the population quantities
\(y(F,D_\alpha)\) and \(y(F, N_\alpha)\). By Result
\ref{res:compound-ratio}, \(\widehat{M}^{V}\) is then consistent and
essentially unbiased for the estimand
\(\frac{y(F, D_\alpha)}{y(F, N_\alpha)}\), so it remains to show that
this estimand is equal to the visible death rate. Using the fact that
\(\bar{v}(D_\alpha,F) = \bar{v}(N_\alpha,F)\), we have \[
\frac{y(F, D_\alpha)}{y(F, N_\alpha)} 
= \frac{y(F, D_\alpha)}{y(F, N_\alpha)}~\frac{\bar{v}(N_\alpha, F)}{\bar{v}(D_\alpha,F)}.
\] Since there are no false positive reports about death or exposure, we
have \[
\frac{y(F, D_\alpha)}{y(F, N_\alpha)}~\frac{\bar{v}(N_\alpha, F)}{\bar{v}(D_\alpha,F)} =
\frac{y^{+}(F, D_\alpha)}{y^{+}(F, N_\alpha)}~\frac{\bar{v}(N_\alpha, F)}{\bar{v}(D_\alpha,F)}.
\]

By the aggregate reporting identity,
\(y^{+}(F,D_\alpha) / \bar{v}(D_\alpha, F) = D^V_\alpha\) and
\(y^{+}(F,N_\alpha) / \bar{v}(N_\alpha, F) = N^V_\alpha\). Therefore, we
have \[
\frac{y^{+}(F, D_\alpha)}{y^{+}(F, N_\alpha)}~\frac{\bar{v}(N_\alpha, F)}{\bar{v}(D_\alpha,F)}
= \frac{D^V_\alpha}{N^V_\alpha} = M^V_\alpha.
\]

\end{proof}

To recap, Result \ref{res:mvis-agg-estimator} shows that the estimator
for visible death rates relies on the important condition that the
visibility of deaths is the same as the visibility of exposure. This
requirement does not mean that reports have to be perfectly accurate; in
particular, omitting siblings will only cause a problem if it happens at
a different rate for dead siblings than it does for living siblings. We
discuss this issue in more detail when we develop the sensitivity
analysis framework, below.

~

The Corollary below says that, when respondents are included in reports,
then an analogous set of conditions to the ones required by Result
\ref{res:mvis-agg-estimator} will allow consistent and unbiased
estimation of the death rate. ~\\
\hspace*{0.333em}

\begin{Corollary}

\label{res:mvis-agg-estimator-withresp} \textbf{(Aggregate visibility
estimation for death rates when respondents are included in reports)}.
Result \ref{res:mvis-agg-estimator} holds under analogous conditions
when reports include the respondents. Specifically, suppose we have a
probability sample \(s \subset F\) and the associated weights \(w_i\)
for all \(i \in s\). Suppose that there are no false positive reports
about deaths or exposure, so that
\(y^{\prime}(F, D_\alpha) = y^{\prime+}(F, D_\alpha)\) and
\(y^\prime(F, N_\alpha) = y^{\prime+}(F, N_\alpha)\). Suppose also that
the visibility of deaths and exposure to the frame population are the
same, and nonzero, so that
\(\bar{v}^{\prime}(D_\alpha, F) = \bar{v}^{\prime}(N_\alpha, F) > 0\).
Then \begin{equation}
\begin{aligned}
\widehat{M}^{\prime V} 
&= \frac{\widehat{y}^{\prime}(F, D_\alpha)}{\widehat{y}^{\prime}(F, N_\alpha)}\\
&= \frac{\sum_{i \in s} w_i~y^{\prime}(i, D_\alpha)}
        {\sum_{i \in s} w_i~y^{\prime}(i, N_\alpha)}\\
\end{aligned}
\label{eq:res-mvis-agg-estimator-withresp}\end{equation} is consistent
and essentially unbiased for
\(M^{\prime V}_\alpha = \frac{D^{\prime V}_\alpha}{N^{\prime V}_\alpha}\),
where the exposure \(N^{V}_\alpha\) is approximated by the number of
visible people, living or dead, in group \(\alpha\).

\end{Corollary}

\begin{proof}

The proof follows the same steps as the proof of Result
\ref{res:mvis-agg-estimator}.

\end{proof}

It is important to note that the aggregate visibility estimator based on
not including information about respondents (Result
\ref{res:mvis-agg-estimator}) is consistent and unbiased for a visible
death rate \(M^V_\alpha\), and that the aggregate visibility estimator
based on including information about respondents (Corollary
\ref{res:mvis-agg-estimator-withresp}) is consistent and unbiased for a
different visible death rate \(M^{\prime V}_\alpha\). In other words,
the \emph{definition} of the visible and invisible populations will be
affected by whether or not respondents are included in the sibling
reports. We analyze this issue in more detail in Appendix
\ref{sec:includerespondent}.

\hypertarget{sensitivity-of-the-aggregate-visibility-death-rate-estimator}{%
\subsection*{Sensitivity of the aggregate visibility death rate
estimator}\label{sensitivity-of-the-aggregate-visibility-death-rate-estimator}}
\addcontentsline{toc}{subsection}{Sensitivity of the aggregate
visibility death rate estimator}

We now turn to the sensitivity of the aggregate visibility death rate
estimator to the various conditions it relies upon. Our approach is to
introduce several quantities, called \emph{adjustment factors}, that
capture the degree to which the conditions the estimator relies upon are
satisfied. These adjustment factors will be equal to 1 under ideal
conditions, and will be different from 1 when a condition required by
the death rate estimator is not satisfied. Our approach is closely
related to other network reporting analyses, including Feehan and
Salganik (2016a) and Feehan, Mahy, and Salganik (2017).

We consider here the case where respondents are not included in reports;
however, an analogous sensitivity expression would result from extended
our approach to the case where respondents are included in sibling
reports.

The first adjustment factor is the \emph{true positive rate} for reports
about visible deaths, \(\tau(F,D_\alpha)\); it is defined as
\begin{equation}
\tau(F, D_\alpha) 
= \frac{\text{average \# times a visible death in $\alpha$ would be reported by someone in $F$}}{\text{average number of connections a visible death in $\alpha$ has to $F$}}
= \frac{\bar{v}(D^V_\alpha, F)}{\bar{d}(D^V_\alpha, F)}.
\label{eq:defn-tau}\end{equation}

Note that we can also write
\(\tau(F,D_\alpha) = \frac{\bar{y}^+(F,D_\alpha)}{\bar{d}(F,D^V_\alpha)}\),
since \(v(D^V_\alpha,F) = y^+(F,D_\alpha)\).

The second adjustment factor is the \emph{precision} for reports about
visible deaths, \(\eta(F, D_\alpha)\); it is defined as

\begin{equation}
\eta(F, D_\alpha) 
= \frac{\text{\# of reported connections from $F$ to $D_\alpha$ that actually lead to $D_\alpha$}}{\text{\# of reported connections from $F$ to $D_\alpha$}}
= \frac{y^+(F, D_\alpha)}{y(F, D_\alpha)}.
\label{eq:defn-eta}\end{equation}

We define analogous adjustment factors \(\tau(F, N_\alpha)\) and
\(\eta(F, N_\alpha)\) for \(F\)'s reports about siblings' exposure.

We can use the adjustment factors introduced in the previous section to
decompose the aggregate visibility sibling survival estimator as

\begin{equation}
\begin{aligned}
M_\alpha &= \frac{D_\alpha}{N_\alpha} \\
&= \frac{D^V_\alpha}{N^V_\alpha} \times 
\frac{p^V_{N_\alpha}}{p^V_{D_\alpha}}
\\
&= \frac{D^V_\alpha}{N^V_\alpha} \times 
\frac{M_\alpha}{M^V_\alpha}
\\
&= \frac{y_{F, D_\alpha}}{y_{F, N_\alpha}} \times
\frac{\bar{d}^V(N_\alpha,F)}{\bar{d}^V(D_\alpha,F)} \times
\frac{\eta(F, D_\alpha)}{\eta(F,N_\alpha)} \times
\frac{\tau(F, N_\alpha)}{\tau(F, D_\alpha)} \times
\frac{M_\alpha}{M^V_\alpha}.
\end{aligned}
\label{eq:agg-mult-sens-deriv}\end{equation}

In order to simplify the eventual framework, let us introduce two
quantities that capture net reporting about deaths and about exposure:

\[
\gamma(F, D_\alpha) = \frac{\tau(F, D_\alpha)}{\eta(F, D_\alpha)},
\]

and

\[
\gamma(F, N_\alpha) = \frac{\tau(F, N_\alpha)}{\eta(F, N_\alpha)}.
\]

The final step is to incorporate the expression for sensitivity to
invisible siblings. Let the visible and invisible death rates differ by
a factor \(K\) so that \(M_\alpha^I = K M_\alpha^V\).
Equation~\ref{eq:invistotaldiff} then tells us that \begin{equation}
\frac{M^V_\alpha}{M_\alpha} = \frac{p^V_{D_{\alpha}} + K(1-p^V_{D_{\alpha}})}{K}.
\label{eq:mva-ma}\end{equation}

Combining Equation~\ref{eq:invistotaldiff} and
Equation~\ref{eq:agg-mult-sens-deriv}, and substituting
\(\gamma(F, D_\alpha)\) and \(\gamma(F, N_\alpha)\), we obtain an
expression that relates the aggregate visibility estimand to the true
death rate:

\begin{equation}
M_\alpha 
= 
\underbrace{
\frac{y(F, D_\alpha)}{y(F, N_\alpha)} 
}_{\substack{\text{aggregate} \\ \text{mulitiplicity} \\ \text{estimand}}}
\times
\underbrace{
\frac{\bar{d}^V(N_\alpha,F)}{\bar{d}^V(D_\alpha,F)} 
}_{\substack{\text{degree} \\ \text{ratio}}}
\times
\underbrace{
\frac{\gamma(F, N_\alpha)}{\gamma(F, D_\alpha)} 
}_{\substack{\text{reporting} \\ \text{accuracy}}}
\times
\underbrace{
\left[\frac{K}{p^I_{D_{\alpha}} + K(1-p^I_{D_{\alpha}})}\right].
}_{\substack{\text{difference between} \\ \text{invisible and} \\ \text{visible} \\ \text{populations}}}
\label{eq:agg-mult-ubersens}\end{equation}

This decomposition relates the quantities we can observe or estimate
from a survey--- i.e., \(y(F,D_\alpha)\) and \(y(F, N_\alpha)\)--to the
quantity that we actually wish to estimate, i.e., \(M_\alpha\).

The decomposition in Equation~\ref{eq:agg-mult-ubersens} produces two
groups of factors that are the ratio of (i) an adjustment factor for
deaths; and (ii) the same adjustment factor for exposure (for example,
\(\frac{\gamma(F,D_\alpha)}{\gamma(F,N_\alpha)}\)). To the extent that
reporting about deaths and reporting about exposure is similar, this is
advantageous: these adjustment factors can cancel or counteract one
another. Intuitively, Equation~\ref{eq:agg-mult-sens-deriv} reveals that
the death rate estimator is quite robust to situations in which
respondents' reports are imperfect, but imperfect in similar ways for
deaths and for people who didn't die.

Equation~\ref{eq:agg-mult-sens-deriv} parameterizes the difference
between the visible and visible populations in terms of
\(p^D_{D_\alpha}\), the proportion of deaths that is invisible.
Researchers may prefer to parameterize this factor in terms of
\(p^I_{N_\alpha}\), the proportion of exposure that is invisible. In
that case, Equation~\ref{eq:agg-mult-sens-deriv} becomes:

\begin{equation}
M_\alpha 
= 
\underbrace{
\frac{y(F, D_\alpha)}{y(F, N_\alpha)} 
}_{\substack{\text{aggregate} \\ \text{mulitiplicity} \\ \text{estimand}}}
\times
\underbrace{
\frac{\bar{d}^V(N_\alpha,F)}{\bar{d}^V(D_\alpha,F)} 
}_{\substack{\text{degree} \\ \text{ratio}}}
\times
\underbrace{
\frac{\gamma(F, N_\alpha)}{\gamma(F, D_\alpha)} 
}_{\substack{\text{reporting} \\ \text{accuracy}}}
\times
\underbrace{
\left[1 + p^I_{N_\alpha} (K-1)\right].
}_{\substack{\text{difference between} \\ \text{invisible and} \\ \text{visible} \\ \text{populations}}}
\label{eq:agg-mult-ubersens-alt}\end{equation}

\hypertarget{sec:ind-vis-deathrate}{%
\section{Estimating death rates using individual
visibility}\label{sec:ind-vis-deathrate}}

This section presents results for death rate estimators that adjust for
visibility at the individual level; that is, each reported connection to
a siblings is adjusted for using the visibility of that specific
reported sibling.

First, we provide an estimator that is based on excluding information
about the survey respondents and only using reports about respondents'
siblings.

~

\begin{Result}

\label{res:mvis-ind-estimator} \textbf{(Individual visibility estimation
for death rates)} Suppose that reports are accurate at the individual
level, so that \(y(i, D_\alpha) = d(i, D_\alpha)\) and
\(y(i, N_\alpha) = d(i, N_\alpha)\). Then \begin{equation}
\begin{aligned}
\widehat{M}^{V} 
&= \frac{\widehat{D_\alpha}}{\widehat{N_\alpha}}\\
&= \frac{\sum_{i \in s} w_i \frac{y(i, D_\alpha)}{y(i, F) + 1}}
       {\sum_{i \in s} w_i \left[ \frac{y(i, N_\alpha \cap F)}{y(i,F)} + \frac{y(i, N_\alpha -F)}{y(i,F) + 1} \right]}
\end{aligned}
\label{eq:res-mvis-ind-estimator}\end{equation} is consistent and
essentially unbiased for
\(M^V_\alpha = \frac{D^{V}_\alpha}{N^{V}_\alpha}\), where the exposure
\(N^{V}_\alpha\) is approximated by the number of visible people, living
or dead, in group \(\alpha\).

\end{Result}

\begin{proof}

The numerator and denominator are Horvitz-Thompson estimators, and
therefore \[
\begin{aligned}
\sum_{i \in s} w_i \frac{y(i, D_\alpha)}{y(i, F) + 1} 
&\longrightarrow 
\sum_{i \in F} \frac{y(i, D_\alpha)}{y(i, F) + 1}\\
&= 
\sum_{i \in F} \frac{d(i, D_\alpha)}{d(i, F) + 1}
&&
\text{by the perfect reporting condition}\\
&= 
D^{V}_{\alpha}.
&&
\text{by the derivation in Equation \ref{eq:ind-deaths-noresp}}
\end{aligned}
\] Similarly, \[
\begin{aligned}
\sum_{i \in s} w_i \left[ \frac{y(i, N_\alpha \cap F)}{y(i,F)} + \frac{y(i, N_\alpha -F)}{y(i,F) + 1} \right] 
&\longrightarrow 
\sum_{i \in F} \left[ \frac{y(i, N_\alpha \cap F)}{y(i,F)} + \frac{y(i, N_\alpha -F)}{y(i,F) + 1} \right]\\
&=
\sum_{i \in F} \left[ \frac{d(i, N_\alpha \cap F)}{d(i,F)} + \frac{d(i, N_\alpha -F)}{d(i,F) + 1} \right]
&&
\text{by perfect reporting condition}\\
&=
N_\alpha
&&
\substack{\text{by relationship in} \\ \text{Equation \ref{eq:ind-exposure-noresp}}}\\
\end{aligned}
\]

Therefore, the numerator of Equation~\ref{eq:res-mvis-ind-estimator} is
consistent and unbiased for the visible deaths, \(D^V_\alpha\), and the
denominator of Equation~\ref{eq:res-mvis-ind-estimator} is consistent
and unbiased for the visible exposure \(N^V_\alpha\). Result
\ref{res:compound-ratio} then shows that the ratio of these two
estimators,
\(\widehat{M}^V = \frac{\widehat{D}_\alpha}{\widehat{N}_\alpha}\) will
be consistent and essentially unbiased for,
\(\frac{D^V_\alpha}{N^V_\alpha}\), the visible death rate.

\end{proof}

Result \ref{res:mvis-ind-estimator} applies in the situation where
respondents are not included in the sibship reports. We now turn to a
Corollary that addresses the situation where information about
respondents is included in the sibship reports.

~

\begin{Corollary}

\label{res:mvis-ind-estimator-withresp} \textbf{(Individual visibility
estimation for death rates when respondents are included in reports)}
Result \ref{res:mvis-ind-estimator} holds under analogous conditions
when reports include the respondents. Specifically, suppose we have a
probability sample \(s \subset F\) and the associated weights \(w_i\)
for all \(i \in s\). Suppose that reports are accurate at the individual
level, so that \(y^{\prime}(i, D_\alpha) = d^{\prime}(i, D_\alpha)\) and
\(y^{\prime}(i, N_\alpha) = d^{\prime}(i, N_\alpha)\). Then
\begin{equation}
\begin{aligned}
\widehat{M}^{\prime V} 
&= \frac{\widehat{D^{\prime}_\alpha}}{\widehat{N^{\prime}_\alpha}}\\
&= \frac{\sum_{i \in s} w_i \frac{y^{\prime}(i, D_\alpha)}{y^{\prime}(i, F)}}
       {\sum_{i \in s} w_i \frac{y^{\prime}(i, N_\alpha)}{y^{\prime}(i,F)}}
\end{aligned}
\label{eq:res-mvis-ind-estimator}\end{equation} is consistent and
essentially unbiased for
\(M^{\prime V}_\alpha = \frac{D^{\prime V}_\alpha}{N^{\prime V}_\alpha}\),
where the exposure \(N^{V}_\alpha\) is approximated by the number of
visible people, living or dead, in group \(\alpha\).

\end{Corollary}

\begin{proof}

The proof is essentially the same as the proof of Result
\ref{res:mvis-ind-estimator}; we repeat it here for clarity. The
numerator and denominator are Horvitz-Thompson estimators, and therefore
\[
\begin{aligned}
\sum_{i \in s} w_i \frac{y^{\prime}(i, D_\alpha)}{y^{\prime}(i, F)} 
&\longrightarrow 
\sum_{i \in F} \frac{y^{\prime}(i, D_\alpha)}{y^{\prime}(i, F)}\\
&= 
\sum_{i \in F} \frac{d^{\prime}(i, D_\alpha)}{d^{\prime}(i, F)}
&&
\text{by the perfect reporting condition}\\
&= 
D^{\prime V}_{\alpha}.
&&
\text{by the derivation in Equation \ref{eq:ind-deaths-withresp}}
\end{aligned}
\] Similarly, \[
\begin{aligned}
\sum_{i \in s} w_i \frac{y^{\prime}(i, N_\alpha}{y^{\prime}(i,F)} 
&\longrightarrow 
\sum_{i \in F} \frac{y^{\prime}(i, N_\alpha)}{y^{\prime}(i,F)}\\
&=
\sum_{i \in F} \frac{d^{\prime}(i, N_\alpha)}{d^{\prime}(i,F)}
&&
\text{by perfect reporting condition}\\
&=
N^{\prime}_\alpha.
&&
\substack{\text{by derivation in} \\ \text{Equation \ref{eq:ind-exposure-withresp}}}\\
\end{aligned}
\] Therefore, the numerator of the estimator
(Equation~\ref{eq:res-mvis-ind-estimator}) is consistent and unbiased
for the visible deaths, \(D^{\prime V}_\alpha\), and the denominator of
Equation~\ref{eq:res-mvis-ind-estimator} is consistent and unbiased for
the visible exposure \(N^{\prime V}_\alpha\). Result
\ref{res:compound-ratio} then shows that the ratio of these two
estimators,
\(\widehat{M}^{\prime V} = \frac{\widehat{D}^{\prime}_\alpha}{\widehat{N}^{\prime}_\alpha}\)
will be consistent and essentially unbiased for,
\(\frac{D^{\prime V}_\alpha}{N^{\prime V}_\alpha}\), the visible death
rate when respondents are included in reports.

\end{proof}

\hypertarget{sensitivity-of-individual-death-rate-estimator}{%
\subsection{Sensitivity of individual death rate
estimator}\label{sensitivity-of-individual-death-rate-estimator}}

Now we turn to an analysis of the sensitivity of the individual
visibility estimator. We follow the same approach that we did for the
aggregate visibility estimator: we introduce a series of adjustment
factors that relate reports to the underlying sibship network. Using
these adjustment factors, we develop expressions that show how reporting
errors and other factors will affect estimated death rates.

\hypertarget{adjustment-factors-for-deaths}{%
\subsubsection*{Adjustment factors for
deaths}\label{adjustment-factors-for-deaths}}
\addcontentsline{toc}{subsubsection}{Adjustment factors for deaths}

We start by defining individual-level adjustment factors that will be
useful in understanding how reporting errors can affect the individual
estimator. For a particular respondent \(i \in F\) in sibship
\(\sigma\), let

\[
\tau(i, D_\alpha \cap \sigma)
= \frac{\text{\# sibs $i$ reports having died in group $\alpha$ that actually died in group $\alpha$}}{\text{\# siblings of $i$ that actually died in group $\alpha$}}
= \frac{y^{+}(i, D^V_\alpha \cap \sigma)}{d^V(D_\alpha \cap \sigma, i)},
%= \frac{v_{D^V_\alpha \cap \sigma, i}}{d^V_{D_\alpha \cap \sigma, i}},
\]

and let

\[
\eta(i, D_\alpha \cap \sigma)
= \frac{\text{\# siblings $i$ reports having died in group $\alpha$ that actually died in group $\alpha$}}{\text{\# of siblings $i$ reports having died in group $\alpha$}}
= \frac{y^{+}(i, D_\alpha \cap \sigma)}{y(i,D_\alpha \cap \sigma)}.
\]

\(\tau(i, D_\alpha \cap \sigma)\) and \(\eta(i, D_\alpha \cap \sigma)\)
are the individual-level analogues of the quantities we introduced in
the previous section\footnote{In the special case where
  \(d^V(i,D_\alpha \cap \sigma) = 0\), we define
  \(\tau(i, D_\alpha \cap \sigma) = 0\); similarly, we define all
  \(\tau\) quantities to be zero when their denominators are zero.
  (Additional \(\tau\) quantities are introduced below.)}. Analogous
reporting quantities can be defined for reports about exposure among
siblings, \(\tau(i, N_\alpha \cap \sigma)\) and
\(\eta(i, N_\alpha \cap \sigma)\), and for reports about siblings'
membership in the frame population, \(\tau(i, \sigma \cap F)\) and
\(\eta(i, \sigma \cap F)\).

We also define a combined adjustment factor for reports about deaths
made by each individual \(i\) in sibship \(\sigma\)\footnote{In the
  special case when \(\eta(i, D_\alpha \cap \sigma) = 0\), we define
  \(\gamma(i, D_\alpha \cap \sigma) = 0\). Similarly, we define all
  \(\gamma\) quantities to be zero when their denominators are zero.
  (Additional \(\gamma\) quantities are introduced below.)}:

\[
\gamma(i, D_\alpha \cap \sigma) = 
\frac{\tau(i, D_\alpha \cap \sigma)}{\eta(i, D_\alpha \cap \sigma)}
\frac{\eta^\prime(i, \sigma \cap F)}{\tau^\prime(i, \sigma \cap F)}.
\]

(Following our notational convention, \(\tau^{\prime}\) and
\(\eta^{\prime}\) refer to individual-level true positive rates and
precision for reports that include the respondent.)
\(\gamma_{i, D_\alpha \cap \sigma}\) can be considered the net reporting
parameter for respondent \(i\). Using \(\gamma(i, D_\alpha)\) instead of
\(\tau(i, D_\alpha)\) and \(\eta(i, D_\alpha)\) will help simplify the
expressions we derive below.

Note that, by definition,

\[
\begin{aligned}
\gamma(i, D_\alpha \cap \sigma) &= 
\frac{y(i, D_\alpha \cap \sigma)}{d(i, D_\alpha \cap \sigma)}
\frac{[d(i, F \cap \sigma) + 1]}{[y(i, F \cap \sigma) + 1]}.
\end{aligned}
\]

This means that

\begin{equation}
\begin{aligned}
\frac{y(i, D_\alpha \cap \sigma)}{y(i, F \cap \sigma) + 1} &=
\frac{d(i, D_\alpha \cap \sigma)}{d(i, F \cap \sigma) + 1}~\gamma(i, D_\alpha \cap \sigma),
\end{aligned}
\label{eq:intermed-eqn1}\end{equation}

a relationship that will prove useful below.

\hypertarget{adjustment-factors-for-exposure}{%
\subsubsection*{Adjustment factors for
exposure}\label{adjustment-factors-for-exposure}}
\addcontentsline{toc}{subsubsection}{Adjustment factors for exposure}

The exposure in the denominator of the estimator in Result
\ref{res:mvis-ind-estimator} is the sum of two terms: one term related
to reports about siblings on the frame population and one term related
to reports about siblings not on the frame population. Therefore, we
require two types of reporting terms: one appropriate for reports about
siblings on the frame population, and one appropriate for reports about
siblings who are not on the frame population. This distinction is
necessary because the individual-level visibility adjustment is
different for siblings who are on and off the frame population.

For reported exposure among siblings on the frame population, we define
net reporting factors:

\[
\gamma(i, N_\alpha \cap F) = 
\frac{\tau(i, N_\alpha \cap F)}{\eta(i, N_\alpha \cap F)}
\frac{\eta(i, \sigma \cap F)}{\tau(i, \sigma \cap F)}.
\]

Similarly, for reported exposure among siblings not on the frame
population, we define net reporting factors:

\[
\gamma(i, N_\alpha -F) = 
\frac{\tau(i, N_\alpha -F)}{\eta(i, N_\alpha -F)}
\frac{\eta^{\prime}(i, \sigma \cap F)}{\tau^{\prime}(i, \sigma \cap F)}.
\]

(Again, following our notational convention, \(\tau^{\prime}\) and
\(\eta^{\prime}\) refer to individual-level true positive rates and
precision for reports that include the respondent.)

\hypertarget{sensitivity-of-deaths}{%
\subsubsection*{Sensitivity of deaths}\label{sensitivity-of-deaths}}
\addcontentsline{toc}{subsubsection}{Sensitivity of deaths}

Using these adjustment factors, we can now derive an expression that
relates reported quantities to the actual visible death rate. We'll do
this separately for the numerator and the denominator; starting with the
numerator, deaths. It will make sense to begin by considering reports
about deaths in a particular sibship \(\sigma\):

\begin{equation}
\begin{aligned}
\sum_{i \in \sigma \cap F} \frac{y(i, D_\alpha)}{y(i,F) + 1}
&=
\sum_{i \in \sigma \cap F} \frac{d(i, D_\alpha)}{d(i,F) + 1}~\gamma(i, D_\alpha)
&&
\text{from Equation }\ref{eq:intermed-eqn1}\\
&=
\bar{d}(\sigma \cap F, D_\alpha)
\sum_{i \in \sigma \cap F} \frac{1}{d(i,F) + 1}~\gamma(i, D_\alpha)
&&
\text{since } d(i,D_\alpha) = \bar{d}(\sigma \cap F, D_\alpha)~\forall i \in \sigma \cap F\\
&=
\frac{\bar{d}(\sigma \cap F, D_\alpha)}{|\sigma \cap F|}
\sum_{i \in \sigma \cap F} \gamma(i, D_\alpha)
&&
\text{since }d(i,F) + 1 = |\sigma \cap F|\\ 
&=
D^V_{\alpha} ~
\bar{\gamma}(\sigma \cap F, D_\alpha),
\end{aligned}
\label{eq:onesib-ind-d-sens}\end{equation}

where we have defined
\(\bar{\gamma}(\sigma \cap F, D_\alpha) = |\sigma \cap F|^{-1} \sum_{i \in \sigma \cap F} \gamma(i, D_\alpha)\)
to be the average net reporting factor for deaths in sibship \(\sigma\).

The derivation in Equation~\ref{eq:onesib-ind-d-sens} is based on the
idea that the net adjustment factor \(\gamma(i, D_\alpha \cap \sigma)\)
relates the individual reports to the actual underlying sibship network;
by applying the net adjustment factor, we develop an understanding of
how the actual number of visible deaths in the sibship,
\(D^V_{\alpha \cap \sigma}\) is related to the individual visibility
estimate.

Equation~\ref{eq:onesib-ind-d-sens} shows how the reported connections
to deaths for one sibship can be expressed as the product of the actual
number of deaths in the sibship and the average net reporting factor for
the sibship, \(\bar{\gamma}(\sigma \cap F, D_\alpha)\). Summing across
sibships produces an expression for the total visible deaths in terms of
all of the reports:

\[
\begin{aligned}
\sum_{i \in F} \frac{y(i, D_\alpha \cap \sigma)}{y_(i, \sigma \cap F) + 1} &=
    \sum_{\sigma \in \Sigma} \sum_{i \in \sigma} \frac{y(i, D_\alpha \cap \sigma)}{y_(i, \sigma \cap F) + 1}\\
    &=  \sum_{\sigma \in \Sigma} D^V_{\alpha \cap \sigma}~\bar{\gamma}(\sigma \cap F, D_\alpha \cap \sigma)\\
\end{aligned}
\]

To keep notation as minimal as possible, we introduce the abbreviations
\(\bar{\gamma}_{\sigma,D} = \bar{\gamma}(\sigma \cap F, D_\alpha \cap \sigma)\),
the average \(\gamma\) for deaths in sibship \(\sigma\);
\(\bar{\gamma}_D = |\Sigma|^{-1} \sum_{\sigma \in \Sigma} \bar{\gamma}_{\sigma,D}\),
the average \(\gamma\) across sibships; and
\(D_\sigma = D^V_{\sigma \cap \alpha}\), sibship \(\sigma\)s visible
deaths in group \(\alpha\). To continue simplifying the expression:

\begin{equation}
\begin{aligned}
    &=  D^V_\alpha \bar{\gamma}_D + |\Sigma|  \text{cov}_\Sigma(D_{\sigma}, \bar{\gamma}_{\sigma,D})\\
    &=  D^V_\alpha 
    \left[
        \bar{\gamma}_D + \frac{ |\Sigma|  \text{cov}_{\Sigma}(D_{\sigma}, \bar{\gamma}_{\sigma, D})}{D_\alpha^V}
    \right]\\
    &=  D^V_\alpha 
    \left[
        \bar{\gamma}_D + \frac{ |\Sigma|  \text{cor}_{\Sigma}(D_{\sigma}, \bar{\gamma}_{\sigma, D}) \text{sd}_{\Sigma}(D_\sigma) \text{sd}_{\Sigma}(\bar{\gamma}_{\sigma,D})}{D_\alpha^V}
    \right]\\
    &=  D^V_\alpha 
    \left[
        \bar{\gamma}_D +  \text{cor}_{\Sigma}(D_{\sigma}, \bar{\gamma}_{\sigma, D}) \text{cv}_{\Sigma}(D_\sigma) \text{sd}_{\Sigma}(\bar{\gamma}_{\sigma,D})
    \right]\\
    &=  D^V_\alpha 
    \left[
        \bar{\gamma}_D +  \text{cor}_{\Sigma}(D_{\sigma}, \bar{\gamma}_{\sigma, D}) \text{cv}_{\Sigma}(D_\sigma) \text{cv}_{\Sigma}(\bar{\gamma}_{\sigma,D}) \bar{\gamma}_D
    \right]\\
    &=  D^V_\alpha 
    ~\bar{\gamma}_D~
    \left[
        1 +  \text{cor}_{\Sigma}(D_{\sigma}, \bar{\gamma}_{\sigma, D}) \text{cv}_{\Sigma}(D_\sigma) \text{cv}_{\Sigma}(\bar{\gamma}_{\sigma,D}) 
    \right].\\
\end{aligned}
\label{eq:all-ind-d-adj}\end{equation}

\hypertarget{sensitivity-of-exposure}{%
\subsubsection*{Sensitivity of exposure}\label{sensitivity-of-exposure}}
\addcontentsline{toc}{subsubsection}{Sensitivity of exposure}

Equation~\ref{eq:all-ind-d-adj} is an expression for the sensitivity of
the numerator of the individual visibility estimator. We now wish to
derive a sensitivity expression for the denominator of the individual
visibility estimator. However, the denominator of the individual
visibility estimator is more complex than the numerator because the
denominator involves two terms: one for reports about exposure to frame
population members and one for reports about exposure to non frame
population members. Considering each of these two terms separately, an
argument parallel to the one for the sensitivity of deaths
(Equation~\ref{eq:all-ind-d-adj}) shows that

\begin{equation}
\begin{aligned}
\sum_{i \in \sigma \cap F} \frac{y(i, N_\alpha \cap F)}{y(i,F)}
&=
N^V_{\alpha \cap F} ~
\bar{\gamma}(\sigma \cap F, N_\alpha \cap F),
\end{aligned}
\label{eq:onesib-ind-nf-sens}\end{equation}

and that

\begin{equation}
\begin{aligned}
\sum_{i \in \sigma \cap F} \frac{y(i, N_\alpha - F)}{y(i,F) + 1}
&=
N^V_{\alpha -F} ~
\bar{\gamma}(\sigma \cap F, N_\alpha -F).
\end{aligned}
\label{eq:onesib-ind-nnotf-sens}\end{equation}

For a particular sibship \(\sigma\),
\(N^V_{\alpha} = N^V_{\alpha \cap F} + N^V_{\alpha -F}\), so we have

\begin{equation}
\begin{aligned}
\sum_{i \in \sigma \cap F} \frac{y(i, N_\alpha \cap F)}{y(i,F)}
+
\sum_{i \in \sigma \cap F} \frac{y(i, N_\alpha - F)}{y(i,F) + 1}
&=
N^V_{\alpha \cap F} ~
\bar{\gamma}(\sigma \cap F, N_\alpha \cap F)
+
N^V_{\alpha -F} ~
\bar{\gamma}(\sigma \cap F, N_\alpha -F)\\
&=
N^V_{\alpha \cap \sigma}
\left[
p_{F|N_\alpha \cap \sigma} \bar{\gamma}(\sigma \cap F, N_\alpha \cap F)
\right. \\
&
\left. ~~~+
(1-p_{F|N_\alpha \cap \sigma}) \bar{\gamma}(\sigma \cap F, N_\alpha -F)
\right],
\end{aligned}
\label{eq:onesib-ind-nboth-sens}\end{equation}

where we have defined
\(p_{F|N_\alpha \cap \sigma} = |\sigma \cap N_\alpha \cap F|/|\sigma \cap N_\alpha|\)
to be the proportion of siblings with expsoure that is on the frame
population.

Equation~\ref{eq:onesib-ind-nboth-sens} shows that, for a single
sibship, we can write the individual visibility estimand as the visible
exposure in the sibship, \(N^V_{\alpha \cap \sigma}\), times a weighted
average of the net reporting factor for exposure on the frame and
exposure not on the frame. In order to simplify this expression, for a
sibship \(\sigma\), we define

\[
\bar{\gamma}^{*}(\sigma \cap F, N_\alpha) = 
p_{F|N_\alpha \cap \sigma} \bar{\gamma}(\sigma \cap F, N_\alpha \cap F) 
+
(1-p_{F|N_\alpha \cap \sigma}) \bar{\gamma}(\sigma \cap F, N_\alpha -F).
\]

Having defined \(\gamma^{*}\), we can rewrite
Equation~\ref{eq:onesib-ind-nboth-sens} as

\begin{equation}
\sum_{i \in \sigma \cap F} \frac{y(i, N_\alpha \cap F)}{y(i,F)}
+
\sum_{i \in \sigma \cap F} \frac{y(i, N_\alpha - F)}{y(i,F) + 1}
= N^V_{\alpha \cap \sigma}~\gamma^{*}(\sigma \cap F, N_\alpha).
\label{eq:onesib-ind-nboth-sens2}\end{equation}

This need to define \(\gamma^{*}\) is awkward, but necessary because of
the way that the denominator of the individual visibility estimator
mixes together siblings with different visibilities.

Equation~\ref{eq:onesib-ind-nboth-sens2} shows the relationship between
reports and the exposure for one sibship. Following a derivation
parallel to the one for deaths above (producing
Equation~\ref{eq:all-ind-d-adj}), we can add up over all sibships to
obtain an expression for the population-level visible exposure:

\begin{equation}
\begin{aligned}
\sum_{i \in F} \left[\frac{y(i, N_\alpha \cap F)}{y(i,F)}
+
\frac{y(i, N_\alpha - F)}{y(i,F) + 1}\right]
&=
    N^V_\alpha 
    ~\bar{\gamma}^{\star}_N~
    \left[
        1 +  \text{cor}_{\Sigma}(N_{\sigma}, \bar{\gamma}^{\star}_{\sigma, N}) \text{cv}_{\Sigma}(N_\sigma) \text{cv}_{\Sigma}(\bar{\gamma}^{\star}_{\sigma,D}) 
    \right].\\
\end{aligned}
\label{eq:all-ind-n-adj}\end{equation}

\hypertarget{sensitivity-of-death-rates}{%
\subsubsection*{Sensitivity of death
rates}\label{sensitivity-of-death-rates}}
\addcontentsline{toc}{subsubsection}{Sensitivity of death rates}

Combining the expression for sensitivity of the reported deaths (the
numerator, Equation~\ref{eq:all-ind-d-adj}) and the expression for the
sensitivity of reported exposure (the denominator,
Equation~\ref{eq:all-ind-n-adj}), we have

\begin{equation}
\begin{aligned}
    \text{individual estimand}
    &=  
    \frac{
    \sum_{i \in F} \frac{y(i, D_\alpha \cap \sigma)}{y_(i, \sigma \cap F) + 1}
    }{
    \sum_{i \in F} \left[\frac{y(i, N_\alpha \cap F)}{y(i,F)}
    +
    \frac{y(i, N_\alpha - F)}{y(i,F) + 1}\right]
    }\\
&=\frac{
    D^V_\alpha 
    ~\bar{\gamma}_D~
    \left[
        1 +  \text{cor}_{\Sigma}(D_{\sigma}, \bar{\gamma}_{\sigma, D}) \text{cv}(D_\sigma)\text{cv}(\bar{\gamma}_{\sigma,D}) 
    \right]
    }{
    N^V_\alpha 
    ~\bar{\gamma}^{\star}_N~
    \left[
        1 +  \text{cor}_{\Sigma}(N_{\sigma}, \bar{\gamma}^{\star}_{\sigma, N}) \text{cv}_{\Sigma}(N_\sigma) \text{cv}_{\Sigma}(\bar{\gamma}^{\star}_{\sigma,D}) 
    \right]
}\\
    &=  
    M^V_{\alpha} \times
    \frac{\bar{\gamma}_D}{\bar{\gamma}^{\star}_N} \times
    \frac{
    %\left[
        1 + K_D 
    %\right]
    }{
    %\left[
        1 + K_N 
    %\right]
},
\end{aligned}
\label{eq:ind-mult-sens-deriv}\end{equation}

where we have simplified the expression by introducing two aggregate
factors

\begin{itemize}
\tightlist
\item
  \(K_D =  \text{cor}_{\Sigma}(D_{\sigma}, \bar{\gamma}_{\sigma, D}) \text{cv}(D_\sigma)\text{cv}(\bar{\gamma}_{\sigma,D})\)
\item
  \(K_N =  \text{cor}_{\Sigma}(N_{\sigma}, \bar{\gamma}^{\star}_{\sigma, N}) \text{cv}_{\Sigma}(N_\sigma) \text{cv}_{\Sigma}(\bar{\gamma}^{\star}_{\sigma,D})\)
\end{itemize}

\(K_D\) and \(K_N\) are complex, but the intuition is that they capture
the relationship between sibship-level reporting factors and deaths (for
\(K_D\)) or exposure (for \(K_N\)). Taking deaths as an example, when
reporting is perfect, there is no relationship between sibship deaths
and reporting, so that \(K_D=0\). When reporting is not perfect, the
sign of \(K_D\) is determined by the correlation factor (since the other
two factors are non-negative). When reporting tends to omit deaths in
sibships with more deaths, then \(K_D > 1\); conversely, when reporting
tends to omit deaths in sibships with fewer deaths, then \(K_D < 1\).

The final step is to incorporate the expression for sensitivity to
invisible siblings. Let the visible and invisible death rates differ by
a factor \(K\) so that \(M_\alpha^I = K M_\alpha^V\).
Equation~\ref{eq:invistotaldiff} then tells us that \begin{equation}
\frac{M^V_\alpha}{M_\alpha} = \frac{p^V_{D_{\alpha}} + K(1-p^V_{D_{\alpha}})}{K}.
\label{eq:mva-ma-ind}\end{equation}

Combining Equation~\ref{eq:invistotaldiff} and
Equation~\ref{eq:ind-mult-sens-deriv}, we obtain an expression that
relates the individual visibility estimand to the true death rate:

\begin{equation}
\begin{aligned}
M_\alpha 
&= 
\frac{
\sum_{i \in F} \frac{y(i, D_\alpha \cap \sigma)}{y(i, \sigma \cap F) + 1}
}{
\sum_{i \in F} \left[\frac{y(i, N_\alpha \cap F)}{y(i,F)}
+
\frac{y(i, N_\alpha - F)}{y(i,F) + 1}\right]
}
\times
    \frac{\bar{\gamma}^{\star}_N}{\bar{\gamma}_D} \times
    \frac{
    %\left[
        1 + K_N 
    %\right]
    }{
    %\left[
        1 + K_D 
    %\right]
    }
    \times
\left[\frac{K}{p^I_{D_{\alpha}} + K(1-p^I_{D_{\alpha}})}\right]\\
&=
\underbrace{
\frac{\widehat{D}^V_\alpha}{\widehat{N}^V_\alpha} 
}_{\substack{\text{individual} \\ \text{multiplicity} \\ \text{estimand}}}
\times
\underbrace{
    \frac{\bar{\gamma}^{\star}_N}{\bar{\gamma}_D} 
}_{\substack{\text{average} \\ \text{reporting} \\ \text{adj. factors}}}
\times
\underbrace{
    \frac{
    %\left[
        1 + K_N 
    %\right]
    }{
    %\left[
        1 + K_D 
    %\right]
    }
}_{\substack{\text{correlation} \\ \text{between} \\ \text{adj. factors} \\ \text{and qoi}}}
    \times
\underbrace{
\left[\frac{K}{p^I_{D_{\alpha}} + K(1-p^I_{D_{\alpha}})}\right].
}_{\substack{\text{difference} \\ \text{between} \\ \text{visibile and} \\ \text{invisibile}}}
\end{aligned}
\label{eq:ind-mult-ubersens}\end{equation}

There are a few important things to note about
Equation~\ref{eq:ind-mult-ubersens}. First, there is no structural term
analogous to the factor
\(\frac{\bar{d}^V_{N_\alpha,F}}{\bar{d}^V_{D_\alpha,F}}\) from the
aggregate sensitivity framework; adjusting at the individual level
eliminates the need for this condition. Second, the factors related to
reporting errors are more complex than the analogous group of factors in
the aggregate estimator. It is no longer the case that reporting errors
can cancel each other out if they are the same for the numerator and
denominator on average; the condition in
Equation~\ref{eq:ind-mult-ubersens} requires that the relationship
between adjustment factors and sibship characteristics, captured by
\(K_D\) and \(K_N\), also agree in order for cancellation to go through.
Third, note that the averages in Equation~\ref{eq:ind-mult-ubersens} are
taken across sibships, and not individuals. This means that, in order to
collect data on the individual-level adjustment factors in
Equation~\ref{eq:ind-mult-ubersens}, we would need detailed information
about sibships, which seems likely to pose a challenge to data
collection efforts.

Finally, Equation~\ref{eq:ind-mult-ubersens} parameterizes the
difference between the visible and visible populations in terms of
\(p^D_{D_\alpha}\), the proportion of deaths that is invisible.
Researchers may prefer to parameterize this factor in terms of
\(p^I_{N_\alpha}\), the proportion of exposure that is invisible. In
that case, Equation~\ref{eq:ind-mult-ubersens} becomes:

\begin{equation}
\begin{aligned}
M_\alpha 
&= 
\underbrace{
\frac{\widehat{D}^V_\alpha}{\widehat{N}^V_\alpha} 
}_{\substack{\text{individual} \\ \text{multiplicity} \\ \text{estimand}}}
\times
\underbrace{
    \frac{\bar{\gamma}^{\star}_N}{\bar{\gamma}_D} 
}_{\substack{\text{average} \\ \text{reporting} \\ \text{adj. factors}}}
\times
\underbrace{
    \frac{
    %\left[
        1 + K_N 
    %\right]
    }{
    %\left[
        1 + K_D 
    %\right]
    }
}_{\substack{\text{correlation} \\ \text{between} \\ \text{adj. factors} \\ \text{and qoi}}}
    \times
\underbrace{
\left[1 + p^I_{N_\alpha} (K-1)\right].
}_{\substack{\text{difference} \\ \text{between} \\ \text{visibile and} \\ \text{invisibile}}}
\end{aligned}
\label{eq:ind-mult-ubersens-alt}\end{equation}

\hypertarget{sec:comparing}{%
\section{Comparing the four estimators}\label{sec:comparing}}

Our results suggest that there are four possible approaches to analyzing
sibling histories: there is the decision to include or exclude
respondents from sibling reports; and there is the choice between the
aggregate and individual visibility estimators. In this Appendix, we
discuss the differences between these four approaches in greater depth.
We first investigate the impact of deciding to include or exclude
respondents from reports, and we argue that it is preferable to exclude
respondents. Then we turn to a discussion of the individual versus
aggregate visibility estimator; our analysis leads us to suggest that,
in the absence of additional information about adjustment factors, the
individual visibility estimator is preferable.

\hypertarget{sec:includerespondent}{%
\subsection{\texorpdfstring{The difference between
\(M^{\prime V}_\alpha\) and
\(M^V_\alpha\)}{The difference between M\^{}\{\textbackslash{}prime V\}\_\textbackslash{}alpha and M\^{}V\_\textbackslash{}alpha}}\label{sec:includerespondent}}

Above, we mentioned that the \emph{definition} of the visible population
is affected by whether or not we include respondents themselves in
sibling reports. In this section, we explain how \(M^{\prime V}_\alpha\)
-- i.e., the death rate when respondents are included in reports --
differs from \(M^V_\alpha\) -- i.e., the death rate in the visible
population when respondents are not included in reports.

We will focus on visibility at the individual level (but note that the
definition of the visible population is not affected by the decision to
adjust for visibility at the individual or at the aggregate level).
Recall from Section \ref{sec:ind-vis} that, under perfect reporting,

\[
v^\prime(i,F) = | \sigma \cap F |,
\]

and

\[
v(i,F) = 
\begin{cases}
    | \sigma \cap F | - 1 & \text{when } i \in F\\
    | \sigma \cap F | & \text{when } i \notin F.
\end{cases}
\]

Thus, for a particular sibling \(i\), we can write the relationship
between \(v^\prime(i,F)\) and \(v(i,F)\) as

\[
v^{\prime}(i,F) =
\begin{cases}
    v(i,F) + 1 & \text{when } i \in F\\
    v(i,F)     & \text{when } i \notin F.
\end{cases}
\]

When switching from including respondents to not including respondents,
the only people whose visibility is affected are in the frame population
\(F\). In particular, this means that deaths -- who are never on the
frame population -- will not have their visibility affected by whether
or not the respondent is included in reports.

For exposure, on the other hand, individuals \(i \in N_\alpha \cap F\)
will have their visibility affected by whether or not respondents are
included in reports. We have two cases:

\hypertarget{case-1-n_alpha-cap-f-phi}{%
\subsubsection*{\texorpdfstring{Case 1:
\(N_\alpha \cap F = \phi\)}{Case 1: N\_\textbackslash{}alpha \textbackslash{}cap F = \textbackslash{}phi}}\label{case-1-n_alpha-cap-f-phi}}
\addcontentsline{toc}{subsubsection}{Case 1: \(N_\alpha \cap F = \phi\)}

Example: men in any age group when only women are interviewed.

In this case, visibility does not change whether or not respondents are
included. So
\(v(i,F) = v^{\prime}(i,F) = d(j,F) + 1~ \forall j \in \sigma[i] \cap F\).

\hypertarget{case-2-n_alpha-cap-f-neq-phi}{%
\subsubsection*{\texorpdfstring{Case 2:
\(N_\alpha \cap F \neq \phi\)}{Case 2: N\_\textbackslash{}alpha \textbackslash{}cap F \textbackslash{}neq \textbackslash{}phi}}\label{case-2-n_alpha-cap-f-neq-phi}}
\addcontentsline{toc}{subsubsection}{Case 2:
\(N_\alpha \cap F \neq \phi\)}

Example: women aged 30-35 in a typical DHS survey.

In this case, the visibility of deaths does not change based on whether
or not respondents are included in reports. However, the visibility of
exposure \emph{does} change. The key question is: whose visibility
switches from \(v^{\prime}(i,F) > 0\) to \(v(i, F) = 0\)? This is the
group of people who become invisible when respondents are excluded from
reports.

Since

\[
v(i,F) = v^{\prime}(i,F) - 1~\forall i \in (N_\alpha \cap F) - D_\alpha,
\]

this can only happen when \(v^{\prime}(i,F) = 1\). So, for each
\(i \in N_\alpha\), we have

\[
v^{\prime}(i,F) - v(i,F) =
\begin{cases}
    1 & \text{if } i \in \alpha \cap F\\
    0 & \text{otherwise.}
\end{cases}
\]

Thus, the set of people who become invisible when respondents are
excluded from reports is precisely the set of siblings \(i\) such that
\(v^{\prime}(i,F) = 1\) and \(i \in N_\alpha \cap F\). The number of
people whose exposure becomes invisible when switching from including
respondents to not including respondents is

\[
\begin{aligned}
    N^{V \prime}_\alpha - N^{V}_\alpha 
    &= \sum_{i \in N_\alpha \cap F} \mathbbm{1}_{[v^{\prime}(i,F) = 1]}.
\end{aligned}
\]

\hypertarget{upshot}{%
\subsubsection*{Upshot}\label{upshot}}
\addcontentsline{toc}{subsubsection}{Upshot}

The visible death rate including respondents in reports is

\[
M^{\prime V}_{\alpha} = \frac{D^{V \prime}_{\alpha}}{N^{V \prime}_\alpha}.
\]

The visible death rate not including respondents in reports is

\[
\begin{aligned}
M^{V}_\alpha &= \frac{D^{V}_{\alpha}}{N^{V}_{\alpha}}\\
&= \frac{D^{\prime V}_\alpha}{N^{\prime V}_{\alpha} - C}
\end{aligned}
\]

where
\(C =\sum_{i \in N_\alpha \cap F} \mathbbm{1}_{[v^{\prime}(i,F) = 1]}\).

\(C\) is a factor that captures the difference in visibility due to
including or not including respondents. \(C\) will tend to be bigger,
inducing a bigger difference between \(M^{\prime V}\) and \(M^V\), when

\begin{itemize}
\tightlist
\item
  \(|\alpha \cap F|\) is bigger
\item
  the distribution of \(v^{\prime}(i,F)\) in \(|\alpha \cap F|\) is
  smaller, meaning that more values of \(v^{\prime}(i,F) = 1\).
\end{itemize}

\hypertarget{sec:includerespondent-model}{%
\subsection{A simple model to study including and excluding respondents
from reports}\label{sec:includerespondent-model}}

In this section, we develop a simple model that can be used to
illustrate how including respondents in reports changes the definition
of the visible and invisible populations. The results of our model agree
with previous models in suggesting that researchers exclude respondents
from the denominator of sibling history estimates (Trussell and
Rodriguez 1990), even though our model makes no assumptions about the
parametric form of the distribution of sibship sizes.

For the purposes of this model, we will assume that we have a homogenous
population whose members all face the same probability of death \(q\).
(So, we disregard differences in age and sex.) We introduce some
notation for the model, which will be used in this section alone. Let
\(N\) be the set of everyone in the population, which is also the set of
people who are exposed to the possibility of death. People are organized
into sibships, and \(s_i \geq 0\) describes the number of siblings \(i\)
has, living or dead; thus, \(i\)'s sibship has \(s_i + 1\) members. We
make no additional assumptions about the distribution of sibship sizes.

Let \(Q_i\) be a random variable whose outcome determines whether or not
person \(i \in N\) dies. We take \( \mathbb{E} [Q_i] = q\) for all \(i\)
and \( \text{cov}[Q_i, Q_j]=0\) for all \(i, j \in N\), meaning that
everyone in the population faces the same probability of death \(q\) and
deaths are not correlated with one another. (In particular, deaths are
not correlated within sibships.) The expected proportion of people who
die is thus \(q\).

Our model describes a stochastic process; one realization of this
process produces a finite population. We are interested in the number of
reported deaths and the amount of reported exposure when (i) respondents
do not include themselves in reports; and (ii) respondents do include
themselves in reports. In both cases, we assume reporting is perfect,
and we assume that everyone who is alive is in the frame population. We
investigate the population-level reports, and the definition of the
invisible and visible populations in both cases.

\hypertarget{case-1-respondents-are-not-included-in-reports}{%
\subsubsection*{Case 1: Respondents are not included in
reports}\label{case-1-respondents-are-not-included-in-reports}}
\addcontentsline{toc}{subsubsection}{Case 1: Respondents are not
included in reports}

\textbf{Visible and invisible death rates}

When respondents are not included in reports, the total number of
visible deaths will be

\begin{equation}
\begin{aligned}
D^V 
&= \sum_{i \in N}
\underbrace{Q_i}_{\substack{\text{prob $i$} \\ \text{ dies}}}
\times
\underbrace{\left[1 - \Pi_{j \sim i} Q_j\right],}_{\substack{\text{prob} \\ \text{at least one of $i$'s} \\ \text{sibs survives}}}
\end{aligned}
\label{eq:model-nor-dv}\end{equation}

where \(j \sim i\) indexes the siblings \(j\) of person \(i\). In
expectation, we have

\begin{equation}
\begin{aligned}
 \mathbb{E} [D^V] 
&=
\sum_{i \in N} q \times (1 - q^{s_i}) = q \sum_{i \in N} (1-q^{s_i}).
\end{aligned}
\label{eq:model-nor-e-dv}\end{equation}

Since \( \mathbb{E} [D] = |N|q\) and \(D^I = D - D^V\),
\( \mathbb{E} [D^I] =  \mathbb{E} [D] -  \mathbb{E} [D^V]\). So we also
have

\begin{equation}
\begin{aligned}
 \mathbb{E} [D^I] 
&=
|N|q - \sum_{i \in N} q \times (1 - q^{s_i}) = q \sum_{i \in N} q^{s_i}.
\end{aligned}
\label{eq:model-nor-e-di}\end{equation}

The total amount of visible exposure will be

\[
\begin{aligned}
N^V
&= \sum_{i \in N}
1
\times
\underbrace{\left[1 - \Pi_{j \sim i} Q_j\right].}_{\substack{\text{prob} \\ \text{at least one of $i$'s} \\ \text{sibs survives}}}
\end{aligned}
\]

Taking expectations, we have

\begin{equation}
\begin{aligned}
 \mathbb{E} [N^V] &= \sum_{i \in N} (1 - q^{s_i}).
\end{aligned}
\label{eq:model-nor-e-nv}\end{equation}

Since \(N^I = N - N^V\), we have
\( \mathbb{E} [N^I] = |N| -  \mathbb{E} [N^V]\). So

\begin{equation}
\begin{aligned}
 \mathbb{E} [N^I] &= \sum_{i \in N} q^{s_i}.
\end{aligned}
\label{eq:model-nor-e-ni}\end{equation}

At the population level, the visible death rate is

\[
\begin{aligned}
M^V &\approx \frac{ \mathbb{E} [D^V]}{ \mathbb{E} [N^V]}\\
&= \frac{q \sum_{i \in N} (1 - q^{s_i})}{ \sum_{i \in N} (1 - q^{s_i})} = q,
\end{aligned}
\]

and the invisible death rate is

\[
\begin{aligned}
M^I &\approx \frac{ \mathbb{E} [D^I]}{ \mathbb{E} [N^I]}\\
&= \frac{q \sum_{i \in N} q^{s_i}}{ \sum_{i \in N} q^{s_i}} = q.
\end{aligned}
\] (In both cases, the approximation is to the first order, and is a
consequence of the fact that the death rate is a ratio of random
variables; in most cases, we expect this approximation to be highly
accurate.) Thus, under this model, when respondents are not counted in
the reports, the invisible and visible death rates are both
approximately \(q\), the model's underlying probability of death.

\textbf{Reporting quantities}

The finite population total number of deaths reported will be

\[
\begin{aligned}
y(F, D^V) &= \sum_{i \in F} y(i, D) = \sum_{i \in F} \sum_{j \sim i} Q_j.
\end{aligned}
\]

Since only people who are alive will be in the frame population, this
becomes

\[
\begin{aligned}
y(F, D^V) &= \sum_{i \in N} (1 - Q_i)~y(i, D) = \sum_{i \in F} (1 - Q_i) \sum_{j \sim i} Q_j.
\end{aligned}
\]

In expectation, we have

\begin{equation}
\begin{aligned}
 \mathbb{E} [y(F, D^V)] 
&=  \mathbb{E} \left[\sum_{i \in N} (1 - Q_i) \sum_{j \sim i} Q_j\right]\\ 
&= \sum_{i \in N}  \mathbb{E} \left[ (1 - Q_i)\right]~ \mathbb{E} \left[\sum_{i \sim j} Q_i \right] \\
&= \sum_{i \in N}  \mathbb{E} \left[ (1 - Q_i)\right]~s_i q\\
&= |N|~(1-q)~s_i q.
\end{aligned}
\label{eq:case1-d}\end{equation}

The finite population total amount of exposure reported will be

\[
\begin{aligned}
y(F, N^V) &= \sum_{i \in F} y(i, N) = \sum_{i \in F} s_i.
\end{aligned}
\] In expectation, we have

\begin{equation}
\begin{aligned}
 \mathbb{E} [y(F, N^V)] 
&=  \mathbb{E} \left[ \sum_{i \in N} (1 - Q_i)~y(i, N)\right]\\
&= |N|~(1-q)~s_i.
\end{aligned}
\label{eq:case1-n}\end{equation}

Thus, in this case, the finite population ratio of expected reports
about deaths (Equation~\ref{eq:case1-d}) and expected reports about
exposure (Equation~\ref{eq:case1-n}) is approximately the probability of
death \(q\). In other words, under this model, excluding respondents
from reports (i) induces the visible and invisible populations to have
the same death rate; and (ii) means that the aggregate visibility
estimator produces essentially unbiased estimates for the visible death
rate.

\hypertarget{case-2-respondents-are-included-in-reports}{%
\subsubsection*{Case 2: Respondents are included in
reports}\label{case-2-respondents-are-included-in-reports}}
\addcontentsline{toc}{subsubsection}{Case 2: Respondents are included in
reports}

\textbf{Visible and invisible death rates}

When respondents are included in reports, the total number of visible
deaths will be

\[
\begin{aligned}
D^{\prime V}
&= \sum_{i \in N}
\underbrace{Q_i}_{\substack{\text{prob $i$} \\ \text{ dies}}}
\times
\underbrace{\left[1 - \Pi_{j \sim i} Q_j\right],}_{\substack{\text{prob} \\ \text{at least one of $i$'s} \\ \text{sibs survives}}}
\end{aligned}
\]

which is the same expression as Equation~\ref{eq:model-nor-dv}, when
respondents are not included in reports. Thus, the expected values for
all of the deaths are the same, i.e.,
\( \mathbb{E} [D^V] =  \mathbb{E} [D^{\prime V}]\) and
\( \mathbb{E} [D^I] =  \mathbb{E} [D^{\prime I}]\).

The total amount of visible exposure will be

\[
\begin{aligned}
N^{\prime V} 
&= \sum_{i \in N}
\left[
\underbrace{Q_i}_{\substack{\text{prob} \\ \text{$i$ dies}}}
\times
\underbrace{\left[1 - \Pi_{j \sim i} Q_j\right]}_{\substack{\text{prob} \\ \text{at least one of $i$'s} \\ \text{sibs survives}}}
+
\underbrace{(1 - Q_i)}_{\substack{\text{prob} \\ \text{$i$ survives}}}
\times
1
\right]
\end{aligned}
\]

Taking expectations, we have

\begin{equation}
\begin{aligned}
 \mathbb{E} [N^{\prime V}] &= \sum_{i \in N} \left[ q (1 - q^{s_i}) + (1-q)\right]\\
&= q \sum_{i \in N} (1 - q^{s_i}) + |N| (1-q)\\
&= |N|q - q \sum_{i \in N} q^{s_i} + |N| - |N|q\\
&= |N| - q \sum_{i \in N} q^{s_i}.
\end{aligned}
\label{eq:model-withr-e-nv}\end{equation}

Since \(N^{\prime I} = N - N^{\prime V}\), we have
\( \mathbb{E} [N^{\prime I}] = |N| -  \mathbb{E} [N^{\prime V}]\). So

\begin{equation}
\begin{aligned}
 \mathbb{E} [N^{\prime I}] &= q \sum_{i \in N} q^{s_i}.
\end{aligned}
\label{eq:model-withr-e-ni}\end{equation}

At the population level, the visible death rate is thus

\begin{equation}
\begin{aligned}
M^{\prime V} &\approx \frac{ \mathbb{E} [D^{\prime V}]}{ \mathbb{E} [N^{\prime V}]}\\
&= \frac{q \sum_{i \in N} (1 - q^{s_i})}{|N| - q\sum_{i \in N}q^{s_i}}\\
&= \frac{q|N| - q\sum_{i \in N} q^{s_i}}{ |N| - q\sum_{i \in N} q^{s_i} }.
\end{aligned}
\label{eq:model-withr-mv}\end{equation}

In general, Equation~\ref{eq:model-withr-mv} is not equal to \(q\). The
invisible death rate is

\begin{equation}
\begin{aligned}
M^{\prime I} &\approx \frac{ \mathbb{E} [D^{\prime I}]}{ \mathbb{E} [N^{\prime I}]}\\
&= \frac{q \sum_{i \in N} q^{s_i}}{q \sum_{i \in N} q^{s_i}} = 1.
\end{aligned}
\label{eq:model-withr-mi}\end{equation}

\textbf{Reporting quantities}

The finite population total number of deaths reported will be \[
\begin{aligned}
y^{\prime}(F, D^V) &= \sum_{i \in F} y(i, D) = \sum_{i \in F} \sum_{j \sim i} Q_j.
\end{aligned}
\]

This is the same as case 1; thus, in expectation,

\begin{equation}
\begin{aligned}
 \mathbb{E} [y(F, D^V)] 
&= |N|~(1-q)~s_i q.
\end{aligned}
\label{eq:case2-d}\end{equation}

The finite population total amount of exposure reported will be

\[
\begin{aligned}
y^{\prime}(F, N^V) &= \sum_{i \in F} (y(i, N) + 1) = \sum_{i \in F} (s_i + 1),
\end{aligned}
\]

where the plus one adds the exposure of the respondent who is, by
definition, alive. In expectation, we have

\begin{equation}
\begin{aligned}
 \mathbb{E} [y^{\prime}(F, N^V)]
&=  \mathbb{E} \left[ \sum_{i \in N} (1 - Q_i)~[y(i, N) + 1]\right]\\
&= |N|~(1-q)~(1 + s_i)\\
&= |N|~(1-q) + |N|~(1-q) s_i.
\end{aligned}
\label{eq:case2-n}\end{equation}

Thus, in this case, the finite population ratio of expected reports
about deaths (Equation~\ref{eq:case2-d}) and expected reports about
exposure (Equation~\ref{eq:case2-n}) does not equal the probability of
death \(q\); there is an extra term in the denominator. In fact, this
term is precisely the \(C\) factor discussed in Appendix
\ref{sec:includerespondent}; that is, under this model, \(C= |N|(1-q)\).

In this case, people can only be invisible if they die, making the
invisible death rate equal to 1. The visible death rate, on the other
hand, will in general be different from \(q\). In other words, under
this model, including respondents in reports (i) induces a difference in
the death rates of the visible and invisible populations, even though
everyone in the population has the same probability of death \(q\); and
(ii) means that the aggregate visibility estimator does not necessarily
produce essentially unbiased estimates for the visible death rate.

\hypertarget{summary}{%
\subsubsection*{Summary}\label{summary}}
\addcontentsline{toc}{subsubsection}{Summary}

To recap, we introduced a model in which all members of a population
have the same probability of dying. Under this model, we saw that it was
appealing to exclude respondents from sibling reports; when respondents
are excluded, the visible and invisible populations have the same death
rate, and that death rate is equal to the probability of an individual
dying. On the other hand, including respondents in sibling reports
induced a difference in death rates between the visible and invisible
populations. Thus, this model agrees with Trussell and Rodriguez (1990)
in suggesting that it is most reasonable to exclude respondents from
sibling reports.

\hypertarget{sec:agg-vis-special}{%
\subsection{Differences between aggregate and individual
visibility}\label{sec:agg-vis-special}}

The aggregate and individual visibility estimators can produce different
results; for example female death rates estimates in
Figure~\ref{fig:sib-ests} are lower for the individual visibility
estimator than for the aggregate visibility estimator. On the other
hand, Figure~\ref{fig:sib-ests} also shows that results for males are
highly consistent with one another. What explains when and how aggregate
and individual visibility estimates differ? In this appendix, we address
this question in two stages: first we derive a relationship between the
aggregate visibility of deaths and the aggregate visibility of exposure;
and, second, we considering two heuristic approximations for the
aggregate visibility of exposure that empirically account for most of
the difference between the aggregate and individual visibility death
rate estimates in Malawi.

We start by deriving a relationship between the aggregate visibility of
deaths and the aggregate visibility of exposure. The basis for this
relationship is the adjustment factors called the \emph{visibility
ratio} in the aggregate sensitivity framework
(Equation~\ref{eq:main-agg-mult-ubersens}). The visibility ratio
captures the condition, required by the aggregate visibility estimator,
that the average visibility of exposure be equal to the average
visibility of deaths. This condition is not required by the individual
visibility estimator; thus, it is a possible source of differences
between the two estimators.

We will start from results about aggregate visibility derived in
Appendix \ref{sec:agg-vis}. Equation~\ref{eq:vis-noresp-agg-perfect-v2}
showed that, for a group \(A \subset U\):

\begin{equation}
\begin{aligned}
\bar{v}(A, F)  
  &=
\bar{v}^{\prime}(A,F) - \frac{|F \cap A|}{|A|}.
\end{aligned}
\label{eq:vis-includeresp-agg-diff}\end{equation}

Equation~\ref{eq:vis-includeresp-agg-diff} is stated in terms of a
generic group \(A\). We now investigate what
Equation~\ref{eq:vis-includeresp-agg-diff} implies for deaths and for
exposure. Deaths will never be on the sampling frame; thus, for deaths,
\(\frac{|F \cap D^V_\alpha|}{|D^V_\alpha|} = 0\).

For exposure, on the other hand, \(|F \cap N^V_\alpha| / |N^V_\alpha|\)
is the proportion of people who contribute exposure that is also on the
sampling frame. This quantity will depend on whether or not survivors in
group \(\alpha\) would be expected to be on the sampling frame. In a
typical Demographic and Health Survey, the sampling frame for sibling
histories will be women of reproductive age. Thus, we expect that:
\begin{equation}
|F \cap N^V_\alpha| / |N^V_\alpha| \approx 
\begin{cases}
%0 & \text{if people in $\alpha$ are not on the frame population}\\
%1 & \text{if people in $\alpha$ are on the frame population}
1 & \text{if $\alpha$ is women in a reproductive age group}\\
0 & \text{otherwise.}\\
\end{cases}
\label{eq:agg-vis-n-app}\end{equation}

Equation~\ref{eq:agg-vis-n-app} is an approximation because in the first
case, when people in \(\alpha\) are on the frame population, some of
those who contribute exposure will also die; thus, the true value will
be somewhat less than 1. So Equation~\ref{eq:agg-vis-n-app} is an
approximation based on the idea that the number of people who die will
usually be small relative to the amount of people who contribute
exposure.

In the paper, we focus on reports that exclude respondents from the
denominator; thus, we are most interested in \(\bar{v}(N^V_\alpha, F)\)
and \(\bar{v}(D^V_\alpha, F)\). It is important to distinguish between
two cases, based on which group \(\alpha\)'s death rate is being
estimated.

\hypertarget{case-1-n_alpha-cap-f-phi-1}{%
\subsubsection*{\texorpdfstring{Case 1:
\(N_\alpha \cap F = \phi\)}{Case 1: N\_\textbackslash{}alpha \textbackslash{}cap F = \textbackslash{}phi}}\label{case-1-n_alpha-cap-f-phi-1}}
\addcontentsline{toc}{subsubsection}{Case 1: \(N_\alpha \cap F = \phi\)}

Example: men in any age group when only women are interviewed.

We have
\(|D^V_\alpha \cap F|/|D^V_\alpha| = |N^V_\alpha \cap F|/|N^V_\alpha| = 0\).
Thus, in this case, it seems reasonable to assume that
\(\bar{v}(N^V_\alpha, F) = \bar{v}(D^V_\alpha, F)\), as the estimates in
DHS reports do.

\hypertarget{case-2-n_alpha-cap-f-neq-phi-1}{%
\subsubsection*{\texorpdfstring{Case 2:
\(N_\alpha \cap F \neq \phi\)}{Case 2: N\_\textbackslash{}alpha \textbackslash{}cap F \textbackslash{}neq \textbackslash{}phi}}\label{case-2-n_alpha-cap-f-neq-phi-1}}
\addcontentsline{toc}{subsubsection}{Case 2:
\(N_\alpha \cap F \neq \phi\)}

Example: women aged 30-35 in a typical DHS survey.

We have \(|D^V_\alpha \cap F|/|D^V_\alpha| = 0\), but
\(|N^V_\alpha \cap F|/|N^V_\alpha| \approx 1\). Thus, our analysis
suggests that, in the absence of additional information about the
adjustment factors, it is most natural to expect that
\(\bar{v}(N^V_\alpha, F) \approx \bar{v}(D^V_\alpha, F) - 1\).

The aggregate visibility estimator in Equation~\ref{eq:mhat-agg} (and
used in DHS reports) assumes that
\(\bar{v}(N^V_\alpha, F) = \bar{v}(D^V_\alpha, F)\). So, our analysis
suggests that this condition is reasonable when the group \(\alpha\) is
not on the frame population (like men in most DHS surveys); however,
when most members of \(\alpha\) will be on the frame population (like
women aged 15-49 in most DHS surveys), our analysis suggests that it is
more natural to assume that
\(\bar{v}(N^V_\alpha, F) = \bar{v}(D^V_\alpha, F) - 1\).

In order to illustrate this analysis, we propose two heuristic ways to
approximate \(\bar{v}(N^V_\alpha, F)\). Then we use these approximations
to adjust aggregate visibility estimates in Malawi. We shall see that
these approximations account for most of the difference between the
individual and aggregate visibility estimates.

The idea behind the two approximations is to use survey respondents'
visibilities to approximate the visibility of exposure. Of course,
deaths (who we do not interview) also contribute exposure, but between
ages of 15 and 49, we expect survivors to outnumber deaths by a
considerable margin, even when death rates are high. So we expect the
average visibility of survey respondents to be similar to the average
visibility of exposure.

The first approximation uses

\begin{equation}
\begin{aligned}
\bar{v}(N^V_\alpha, F) &\approx \bar{y}(F, F). && \text{(All ages approximation)}
\end{aligned}
\label{eq:s-approx-1}\end{equation}

Equation~\ref{eq:s-approx-1} approximates the visibility of exposure in
group \(\alpha\) by using the average visibility of survey respondents.
Under this approximation, estimates from the aggregate visibility
estimator in Equation~\ref{eq:mhat-agg} should be adjusted by a factor
of \(\frac{\widehat{\bar{y}}(F, F)}{\widehat{\bar{y}}(F,F) + 1}\) when
group \(\alpha\) overlaps with the frame population.

We call Equation~\ref{eq:s-approx-1} the \emph{all-ages} approximation,
to distinguish it from the second approximation:

\begin{equation}
\begin{aligned}
\bar{v}(N^V_\alpha, F) & \approx \bar{y}(F_\alpha, F). && \text{(Age-specific approximation)}
\end{aligned}
\label{eq:s-approx-2}\end{equation}

Equation~\ref{eq:s-approx-2} takes the idea from
Equation~\ref{eq:s-approx-1} and applies it to the specific group
\(\alpha\). This approximation is motivated by the form of individual
visibility estimator. Under this approximation, estimates from the
aggregate visibility estimator in Equation~\ref{eq:mhat-agg} should be
adjusted by a factor of
\(\frac{\widehat{\bar{y}}(F_\alpha, F)}{\widehat{\bar{y}}(F_\alpha,F) + 1}\)
when group \(\alpha\) overlaps with the frame population.

Figure~\ref{fig:agg-vis-adj-all} illustrates using the Malawi example
from the main paper. The figure shows the individual and aggregate
visibility estimates, along with aggregate visibility estimates that
have been adjusted using (i) the approximation based on all ages; and
(ii) the age-specific approximations. Note that the adjustments only
affect groups \(\alpha\) whose members could potentially also be members
of the frame population; in this case, that is women aged 15-49.
Figure~\ref{fig:agg-vis-adj-all} shows that the two heuristic
approximations do an excellent job of approximating the individual
visibility results.

\begin{figure}
\hypertarget{fig:agg-vis-adj-all}{%
\centering
\includegraphics{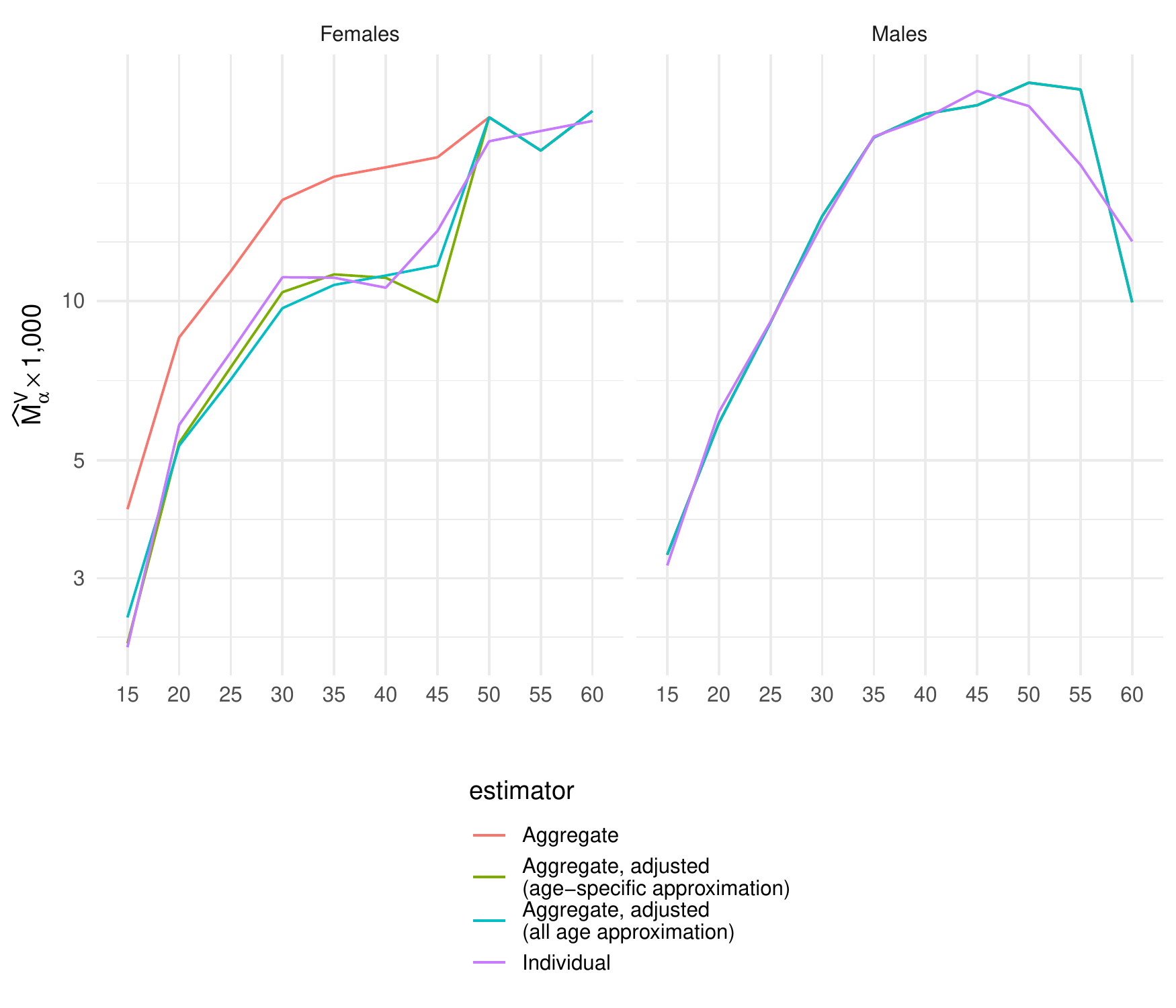}
\caption{Comparing three variants of adjusted aggregate visibility
estimates to the unadjusted aggregate visibility estimates and
individual visibility estimates from the 2000 Malawi
DHS.}\label{fig:agg-vis-adj-all}
}
\end{figure}

To recap, our derivations suggest that when the group \(\alpha\) is not
on the frame population, it may be reasonable to assume that
\(\bar{v}(N^V_\alpha, F) = \bar{v}(D^V_\alpha, F)\), like the aggregate
visibility estimator does. However, when the group \(\alpha\) overlaps
with the frame population -- for example, when it includes women aged
15-49 on a DHS survey -- then it seems more reasonable to assume that
\(\bar{v}(N^V_\alpha, F) = \bar{v}(D^V_\alpha, F) - 1\). To illustrate,
we approximated \(\bar{v}(N^V_\alpha, F)\) in two different ways. Using
the adjustment factors suggested by our approximations,
Figure~\ref{fig:agg-vis-adj-all} adjusted aggregate visibility
estimates, and the resulting adjusted estimates were very close to the
individual visibility estimates. Thus, for the 2000 Malawi DHS, the
difference between the aggregate and individual visibility estimates
appears to be explained by the implicit assumption about the visibility
adjustment factor made by the aggregate visibility estimator.

\hypertarget{sec:background-facts}{%
\section{Useful facts}\label{sec:background-facts}}

\hypertarget{covariances}{%
\subsubsection*{Covariances}\label{covariances}}
\addcontentsline{toc}{subsubsection}{Covariances}

The following fact will be useful in some of our analysis below; see,
for example, Feehan and Salganik (2016a) for a derivation.

~

\begin{Fact}

\label{res:sumprod-cov} Suppose we have a finite population \(U\) of
size \(N\) and that \(a_i, b_i \in \mathbb{R}\) are defined for all
\(i \in N\). Then \[
\sum_{i \in U} a_i b_i = N\left[\bar{a}\bar{b} +  \text{cov}_U(a_i,b_i)\right],
\] where \(\bar{a} = N^{-1} \sum_{i \in U} a_i\),
\(\bar{b} = N^{-1} \sum_{i \in U} b_i\), and \( \text{cov}_U(a_i,b_i)\)
is the finite population covariance of the \(a_i\) and \(b_i\) values.

\end{Fact}

\hypertarget{aggregating-death-rates-across-groups}{%
\subsection*{Aggregating death rates across
groups}\label{aggregating-death-rates-across-groups}}
\addcontentsline{toc}{subsection}{Aggregating death rates across groups}

In order to develop our sensitivity frameworks, we need a couple of
technical results. These results will help us understand how death rate
estimates can be affected by invisible deaths and invisible exposure.

Demographers frequently use the weighted arithmetic mean. However, many
other means exist; in understanding how the visible and invisible death
rates aggregate, the harmonic mean will play an important role. So, for
convenience, we now review the definition of the weighted harmonic mean:

~

\begin{Definition}

\label{def:whm} Let \(\boldsymbol{x}, \boldsymbol{w} \in \mathbb{R}^n\)
and let \(x_i > 0\) and \(w_i > 0\) for all \(i\). Then the
\textbf{Weighted Harmonic Mean} of the \(\boldsymbol{x}\) values, with
weights given by the \(\boldsymbol{w}\) values, is \[
H[\boldsymbol{x}; \boldsymbol{w}] = \frac{\sum_{i=1}^n w_i}{\sum_{i=1}^n \frac{w_i}{x_i}}.
\]

\end{Definition}

For comparison, the usual \textbf{weighted arithmetic mean} is given by

\begin{equation}
A[\boldsymbol{x}; \boldsymbol{w}] = \frac{\sum_{i=1}^n w_i x_i}{\sum_{i=1}^n w_i}.
\label{eq:wam}\end{equation}

The next derivation shows that the weights can be rescaled without
affecting the weighted harmonic mean.

~

\begin{Result}

\label{res:whmscale} Let
\(\boldsymbol{x}, \boldsymbol{w} \in \mathbb{R}^n\) and let \(x_i > 0\)
and \(w_i > 0\) for all \(i\). Let \(\boldsymbol{w^\prime}\) be defined
so that \(w^\prime_i = K w_i\) for all \(i\) and for \(K > 0\). Then
\(H[\boldsymbol{x}; \boldsymbol{w}] = H[\boldsymbol{x}; \boldsymbol{w^\prime}]\).

\end{Result}

\begin{proof}

This property follows directly from the definition of the weighted
harmonic mean (Definition \ref{def:whm}): \[
H[\boldsymbol{x}; \boldsymbol{w^\prime}] = \frac{\sum_{i=1}^n w^\prime_i}{\sum_{i=1}^n \frac{w^\prime_i}{x_i}}
= \frac{\sum_{i=1}^n K~w_i}{\sum_{i=1}^n \frac{K~w_i}{x_i}}
= \frac{\sum_{i=1}^n w_i}{\sum_{i=1}^n \frac{w_i}{x_i}} = H[\boldsymbol{x}; \boldsymbol{w}].
\]

\end{proof}

We can connect these insights about harmonic means to deepen our
understanding of how death rates aggregate across groups, as Result
\ref{res:agghm} shows:

~

\begin{Result}

\label{res:agghm} Suppose two demographic groups have death rates
\(M_{1} = \frac{D_{1}}{N_{1}}\) and \(M_{2} = \frac{D_{2}}{N_{2}}\),
where \(D_{1}\) is the number of deaths in the first group, \(N_{1}\) is
the person-years of exposure in the first group, \(D_{2}\) is the number
of deaths in the second group, and \(N_{2}\) is the person-years of
exposure in the second group. Now suppose we combine the two groups and
treat them as one aggregate group. Then the death rate for the aggregate
group is the weighted harmonic mean of the death rates in the subgroups,
with weights given by the number of deaths in each subgroup: \[
M_\text{agg} = H[(M_1, M_2); (D_1, D_2)] = \frac{D_1 + D_2}{N_1 + N_2}.
\]

\end{Result}

\begin{proof}

In the combined group, the total exposure is \(N_1 + N_2\) and the total
number of deaths is \(D_1 + D_2\). Thus, the combined death rate is
indeed \(M_\text{agg} = \frac{D_1 + D_2}{N_1 + N_2}\). It remains to
show that this is the weighted harmonic mean. By the definition of the
weighted harmonic mean (Definition \ref{def:whm}), we have \[
\begin{aligned}
H[(M_1, M_2); (D_1, D_2)] &= \frac{D_1 + D_2}{\frac{D_1}{M_1} + \frac{D_2}{M_2}}\\
&= \frac{D_1 + D_2}{N_1 + N_2} = M_\text{agg},
\end{aligned}
\] where the last step follows because
\(\frac{D_1}{M_1} = D_1 \times \frac{N_1}{D_1} = N_1\).

\end{proof}

Next, we will see that when groups are aggregated with a focus on the
amount of exposure, the weighted arithmetic mean describes the resulting
aggregate death rate.

~

\begin{Result}

\label{res:aggam} Suppose two demographic groups have death rates
\(M_{1} = \frac{D_{1}}{N_{1}}\) and \(M_{2} = \frac{D_{2}}{N_{2}}\),
where \(D_{1}\) is the number of deaths in the first group, \(N_{1}\) is
the person-years of exposure in the first group, \(D_{2}\) is the number
of deaths in the second group, and \(N_{2}\) is the person-years of
exposure in the second group. Now suppose we combine the two groups and
treat them as one aggregate group. Then the death rate for the aggregate
group is the weighted arithmetic mean of the death rates in the
subgroups, with weights given by the amount of exposure in each
subgroup: \[
M_\text{agg} = A[(M_1, M_2); (N_1, N_2)] = \frac{D_1 + D_2}{N_1 + N_2}.
\]

\end{Result}

\begin{proof}

By the definition of the weighted arithmetic mean
(Equation~\ref{eq:wam}), we have \[
\begin{aligned}
A[(M_1, M_2); (N_1, N_2)] &= \frac{N_1 M_1 + N_2 M_2}{N_1 + N_2}\\
&= \frac{D_1 + D_2}{N_1 + N_2} = M_\text{agg}.
\end{aligned}
\]

\end{proof}

Taken together, Results \ref{res:agghm} and \ref{res:aggam} will be
useful in constructing sensitivity frameworks because they allow us to
parameterize the difference between the visible and invisible
populations in terms of either the fraction of deaths that is invisible
(leading to the harmonic mean relationship) or the fraction of exposure
that is invisible (leading to the arithmetic mean relationship). Which
result to use depends on whether researchers want to parameterize
sensitivity using the proportion of deaths that is invisible (Result
\ref{res:agghm}) or the proportion of exposure that is invisible (Result
\ref{res:aggam}).

Note that Results \ref{res:agghm} and \ref{res:aggam} can be extended to
more than two groups; in general, the death rate for an aggregation of
several groups will be given by the weighted harmonic (arithmetic) mean
of the component death rates, with the weights given by the number of
deaths (amount of exposure) in each component group.

\hypertarget{sec:ap-variance}{%
\section{Variance estimation}\label{sec:ap-variance}}

Variance estimators can be useful in at least two ways. First, variance
estimators enable researchers to understand and communicate the sampling
variance associated with any point estimate. Second, before a survey
design is decided upon, variance estimators can be used to understand
how big a sample is needed to estimate a quantity at a given level of
precision. Somewhat surprisingly, little work has formally analyzed the
variance of sibling history estimates\footnote{Hanley, Hagen, and
  Shiferaw (1996) studied the sisterhood method, which is closely
  related to the sibling survival method. By approximating the
  proportion of sisters reported dead as a binomial variable, Hanley,
  Hagen, and Shiferaw (1996) derives an expression for the variance of a
  sisterhood estimate. Although this approach has been very useful,
  future work could likely improve upon Hanley, Hagen, and Shiferaw
  (1996)'s results; their expression is based on several simplifications
  and does not appear to account for the complex design used in almost
  all surveys that collect sibling history data.}.

All of the sibling estimators discussed in this study are variants of a
ratio or compound ratio estimator. A standard result in the survey
sampling literature shows that the variance of such an estimator can be
estimated using a Taylor approximation. Here, we state this result and
explain how it relates to the sibling survival estimators.

Sarndal, Swensson, and Wretman (2003, sec 5.6) shows that the relative
variance of any ratio estimator of the form
\(\widehat{M} = \frac{\widehat{D}}{\widehat{N}}\) can be approximated by

\begin{equation}
 \widehat{\text{Var}} [\widehat{M}] \approx \frac{1}{\widehat{N}^2} \left[
 \widehat{\text{Var}} [\widehat{D}] + \widehat{M}^2~ \widehat{\text{Var}} [\widehat{N}] - 2 \widehat{M}  \widehat{\text{Cov}} [\widehat{D}, \widehat{N}]
\right].
\label{eq:approxvar}\end{equation}

Multiplying Equation~\ref{eq:approxvar} through by
\(\frac{1}{\widehat{M}^2}\), we obtain an expression for the approximate
relative variance

\begin{equation}
\frac{ \widehat{\text{Var}} [\widehat{M}]}{\widehat{M}^2} =  \widehat{\text{Rel-Var} }[\widehat{M}^2] \approx 
 \widehat{\text{Rel-Var} }[\widehat{D}] +  \widehat{\text{Rel-Var} }[\widehat{N}] - 2~ \widehat{\text{Rel-Cov} }[\widehat{D}, \widehat{N}],
\label{eq:approxrelvar1}\end{equation}

where
\( \widehat{\text{Rel-Var} }[\widehat{X}] =  \widehat{\text{Var}} [\widehat{X}]/\widehat{X}^2\)
is the relative sampling variance and
\( \widehat{\text{Rel-Cov} }[\widehat{X},\widehat{Y}] =  \widehat{\text{Cov}} [\widehat{X},\widehat{Y}]/\widehat{X}\widehat{Y}\)
is the relative sampling covariance.

In the context of estimating death rates, the \(\widehat{M}\) in
Equation~\ref{eq:approxvar} is the estimated death rate, the
\(\widehat{D}\) is the reports about deaths, and the \(\widehat{N}\) is
the reports about exposure. Equation~\ref{eq:approxvar} and
Equation~\ref{eq:approxrelvar1} show for a given level of mortality, the
estimated sampling uncertainty will be lower when

\begin{itemize}
\tightlist
\item
  the reports exposure about deaths and exposure are estimated from the
  sample with a high degree of precision, making
  \( \widehat{\text{Rel-Var} }[\widehat{D}]\) and
  \( \widehat{\text{Rel-Var} }[\widehat{N}]\) small
\item
  the sampling design produces a high, positive covariance between the
  estimated numerators and denominators, making
  \( \widehat{\text{Rel-Cov} }[\widehat{D},\widehat{N}]\) large
\end{itemize}

\hypertarget{sec:simulation}{%
\section{Simulation study}\label{sec:simulation}}

We conducted a simulation study with two goals: (1) we wanted to confirm
the correctness of our analytical results, including the sensitivity
frameworks; and, (2) we wanted to empirically compare the performance of
the four possible estimators. We based our simulation on the sibships
reported about in the 2000 Malawi DHS study. Our aim was not to
re-construct a perfectly authentic population-level sibship structure;
instead, we wanted a reasonably realistic population that could help us
achieve the goals of confirming the analytical results and understanding
the differences between the four possible estimators.

\hypertarget{constructing-the-universe}{%
\subsection{Constructing the universe}\label{constructing-the-universe}}

We start with the observed set of sibships that are reported about in
the sibling history module; there is one sibship reported about for each
survey respondent. We assume that these sibships are distinct, ignoring
the fact that more than one member of a sibship may have been
interviewed in the study. We use these sibships as the basis for
constructing a pseudo-population of sibships as follows
(Figure~\ref{fig:sim-overview}):

\begin{enumerate}
\def\labelenumi{\arabic{enumi}.}
\tightlist
\item
  Starting from the 13,161 sibships in the dataset, sample
  \(M_{\text{sibships}}\) to form the pseudo-population of sibships. We
  sample with replacement, so some sibships are sampled multiple times,
  and each sibship is sampled with probability proportional to its
  visibility to the frame population, i.e.~the number of sibship
  members, including the respondent, who were eligible to respond to the
  survey.
\item
  For each of the \(M_{\text{sibships}}\) resampled sibships, with
  probability \(\frac{1}{2}\), we flip the sexes of the reported
  siblings. (This accounts for the fact that only females were
  interviewed in Malawi; without this step, we would end up with an
  unrealistic gender distribution.)
\item
  We then form a universe of siblings from the individual siblings
  corresponding to the (possibly gender-flipped) resampled sibships.
\end{enumerate}

\begin{figure}
\hypertarget{fig:sim-overview}{%
\centering
\includegraphics[width=\textwidth,height=0.9\textheight]{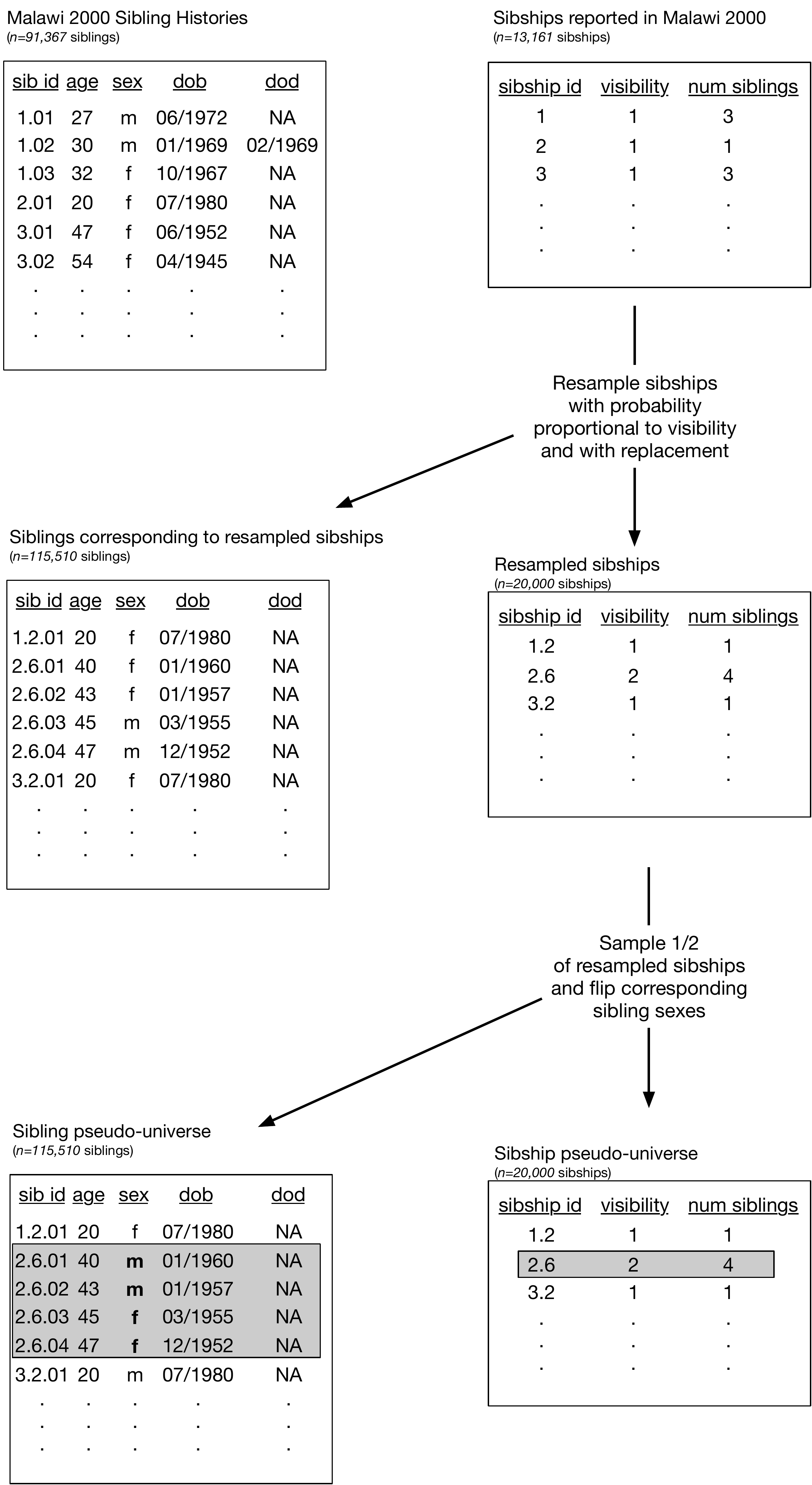}
\caption{Overview of the method used to construct a pseudo-universe that
forms the basis for the simulation study.}\label{fig:sim-overview}
}
\end{figure}

The result is a universe of siblings who are assigned to sibships that
is approximately representative of the 2000 Malawi sibship population.
Figure~\ref{fig:sim-universe-n} and Figure~\ref{fig:sim-universe-m} show
the age-sex distribution and death rates in the simulated universe. From
the universe of siblings and sibships, we create a sibship network
\texttt{igraph} object. We also create a census dataset, and use it as
the basis for calculating true death rates by age and sex.

\begin{figure}
\hypertarget{fig:sim-universe-n}{%
\centering
\includegraphics{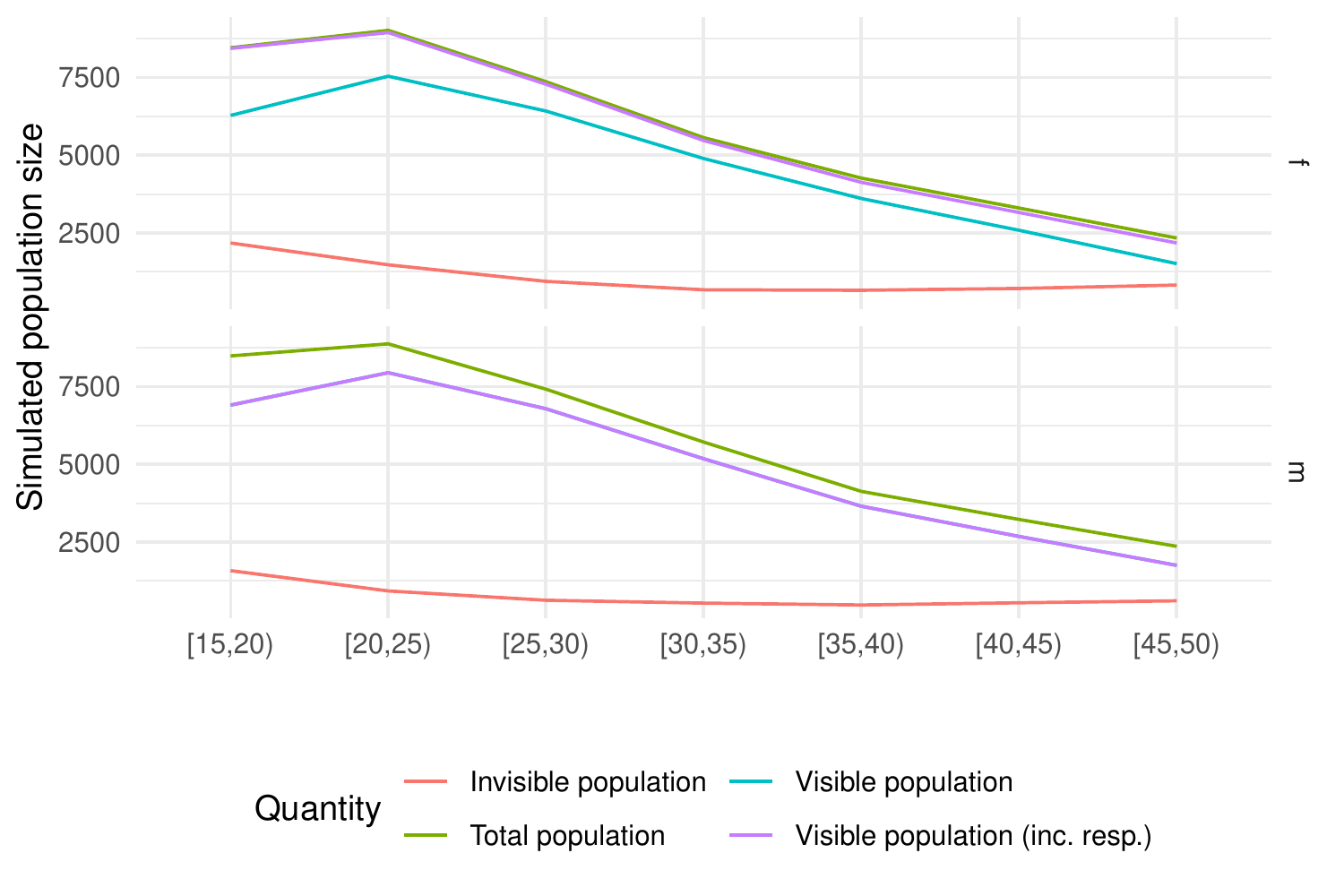}
\caption{Age-sex distribution of the simulated
population.}\label{fig:sim-universe-n}
}
\end{figure}

\begin{figure}
\hypertarget{fig:sim-universe-m}{%
\centering
\includegraphics{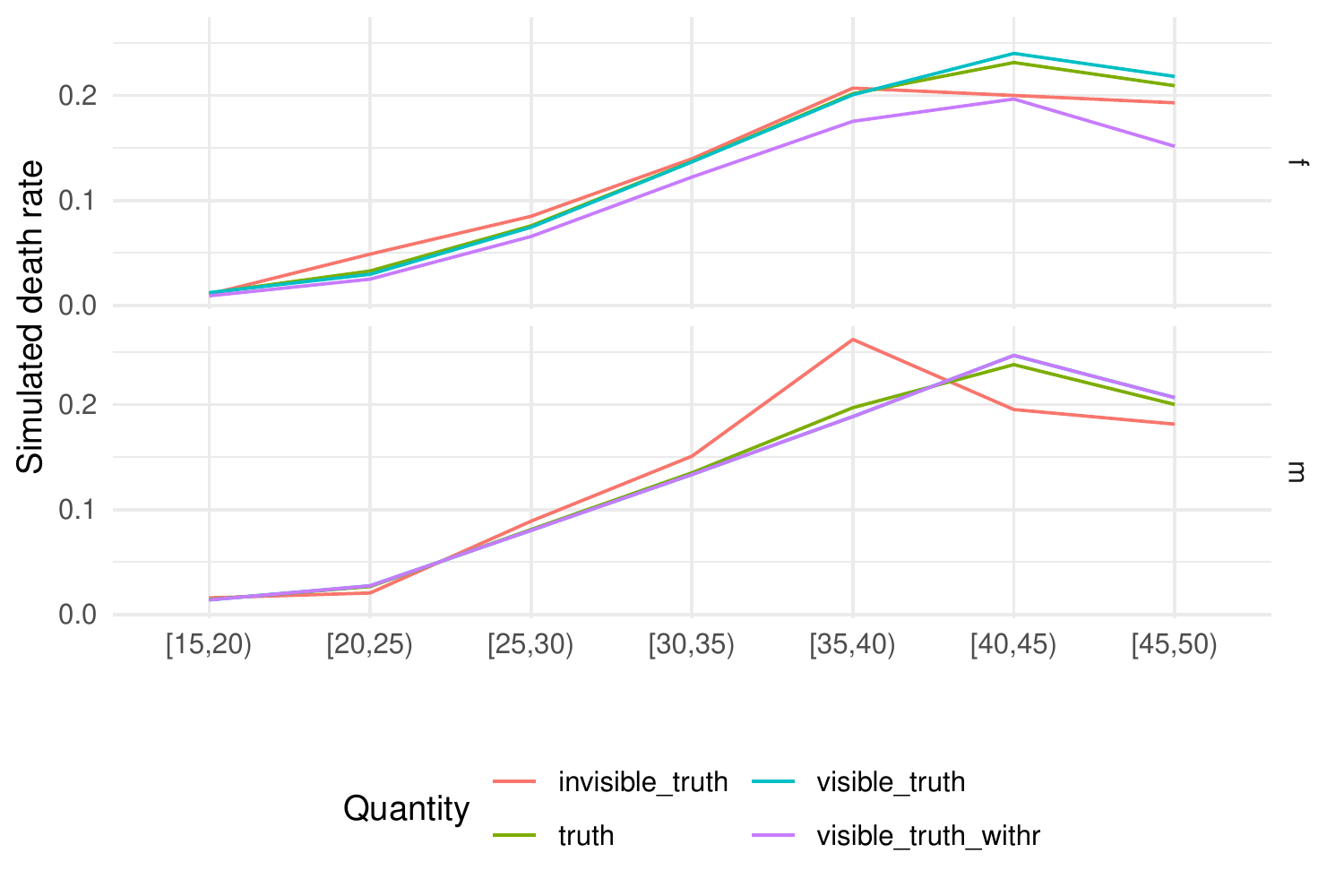}
\caption{Age- and sex-specific death rates among members of the
simulated population.}\label{fig:sim-universe-m}
}
\end{figure}

\hypertarget{simulating-reporting-error-and-sampling}{%
\subsection{Simulating reporting error and
sampling}\label{simulating-reporting-error-and-sampling}}

Having created a pseudo-universe of siblings linked together in
sibships, the goal is to use this population as the basis for simulated
sibling history surveys under different scenarios. By \emph{scenario},
we mean a set of aggregate reporting parameters together with a sampling
fraction. We now describe how we simulate sibling history surveys under
several different scenarios.

Using the sibling universe, we create several population-level reporting
networks, one for each combination of aggregate reporting parameters
\(\tau_D = \{0.8, 1\}\) and \(\tau_N = \{0.8, 1\}\). We assume there are
no false positives throughout (i.e., \(\eta_D = 1\) and \(\eta_N = 1\)).

For each population reporting network, we calculate all of the aggregate
and individual-level adjustment factors.

For each population reporting network, we simulate \(M=1000\) sample
surveys for each sampling fraction in
\(f = \{0.05, 0.1, 0.15, 0.3, 0.6\}\).

Finally we calculate the four estimates for each simulated survey:
aggregate visibility including/excluding the respondent and individual
visibility including/excluding the respondent. Within each scenario,
these \(M\) estimates are the simulated sampling distribution for the
given estimator.

To recap, we generate a pseudo-universe of people linked together in
sibships based on the actual sibships reported in the 2000 Malawi DHS
survey. From this pseudo-universe, we generate many simulated surveys of
different sample sizes, and different assumptions about reporting
errors. By calculating death rate estimates using the four estimators
for each of these simulated surveys and comparing the estimates to the
known underlying truth, we can assess how well the sampling
distributions of the four estimators recover true death rates under
different reporting conditions and sample sizes.

\hypertarget{results}{%
\subsection{Results}\label{results}}

\hypertarget{confirming-the-accuracy-of-the-sensitivity-framework}{%
\subsubsection{Confirming the accuracy of the sensitivity
framework}\label{confirming-the-accuracy-of-the-sensitivity-framework}}

First, we examine plots that compare the estimands for the aggregate and
individual visibility estimators to the true underlying age-specific
visible death rates. By examining the estimands, we remove the
complication of sampling and focus on the quantity that each estimator
would estimate if the entire frame population were interviewed. These
plots will confirm the correctness of our analytical frameworks and
provide some intuition about how estimators are affected by reporting
errors.

Figure~\ref{fig:check-agg} compares the true visible death rates (x
axis) to the adjusted and unadjusted aggregate death rate estimands (y
axis). Two important features emerge from Figure~\ref{fig:check-agg}:
first, the adjusted estimands all lie on the diagonal \(y=x\) line,
confirming the correctness of the sensitivity framework for the
aggregate estimator (Equation~\ref{eq:agg-mult-ubersens}). Second, by
comparing the unadjusted estimates across the four reporting scenarios,
it is clear that the unadjusted estimands for the scenario in which
\(\tau_D = 0.8\) and \(\tau_N = 0.8\) (top-left panel) are nearly as
accurate as the scenario in which \(\tau_D=1\) and \(\tau_N=1\)
(bottom-right panel). Intuitively, this suggests that in some cases
imperfect reporting may not be very problematic for the aggregate
sibling survival estimator, as long as the imperfect reporting is
similar for deaths and for exposure; in that case,
Figure~\ref{fig:check-agg} shows that the reporting errors can cancel
out (confirming the intuition from Equation~\ref{eq:agg-mult-ubersens}).

Figure~\ref{fig:check-ind} is analogous to Figure~\ref{fig:check-agg},
but for the individual visibility estimator.

\begin{figure}
\hypertarget{fig:check-agg}{%
\centering
\includegraphics{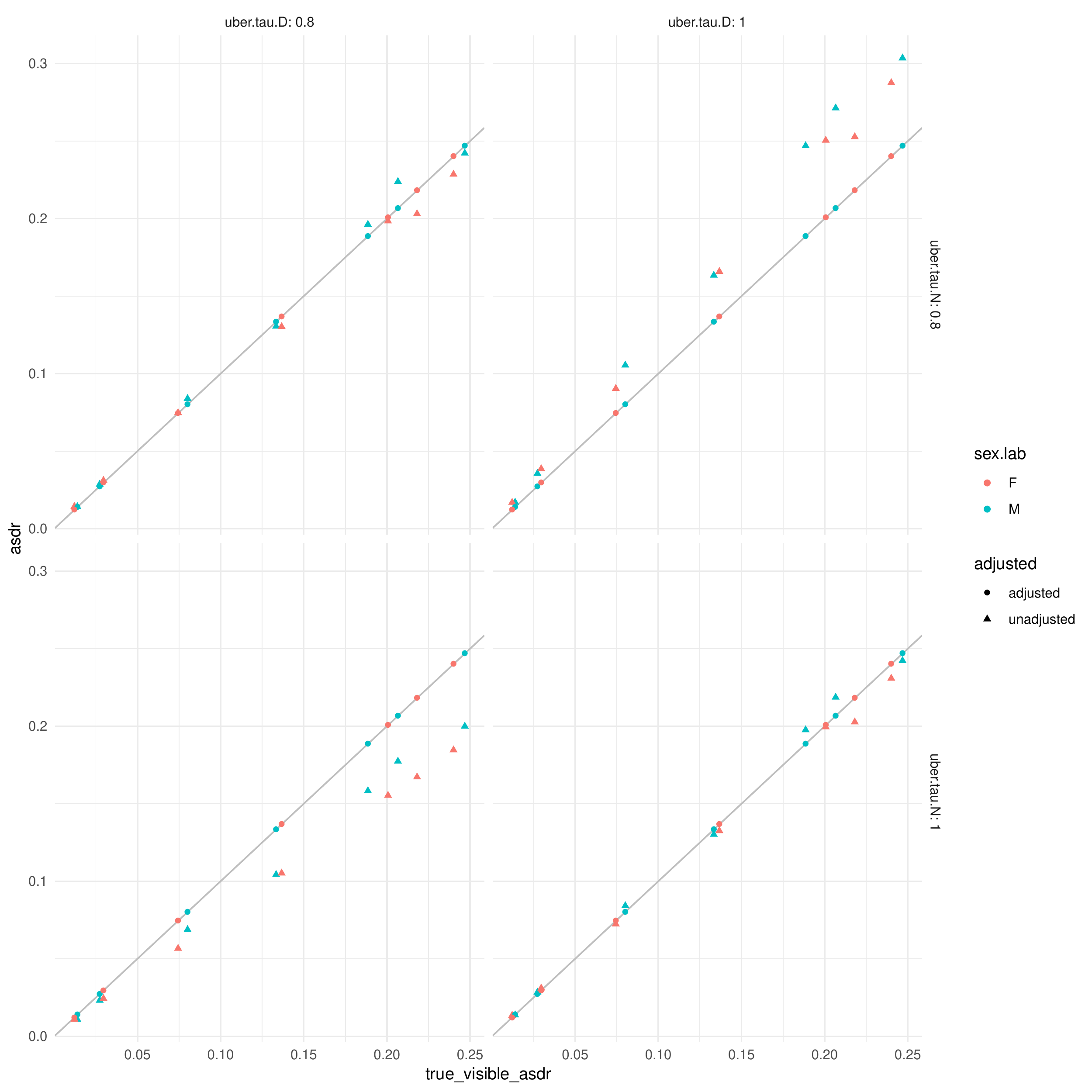}
\caption{True age-specific death rates (x axis) against adjusted and
unadjusted death rate aggregate visibility estimands (y axis). Male
death rates are in blue and females are in red. The four panels show
four different reporting scenarios. The adjusted estimands all agree
with the true underlying death rates, illustrating the correctness of
the aggregate sensitivity framework in
Equation~\ref{eq:agg-mult-ubersens}. These results are for the aggregate
estimator that does not include the respondent.}\label{fig:check-agg}
}
\end{figure}

\begin{figure}
\hypertarget{fig:check-ind}{%
\centering
\includegraphics{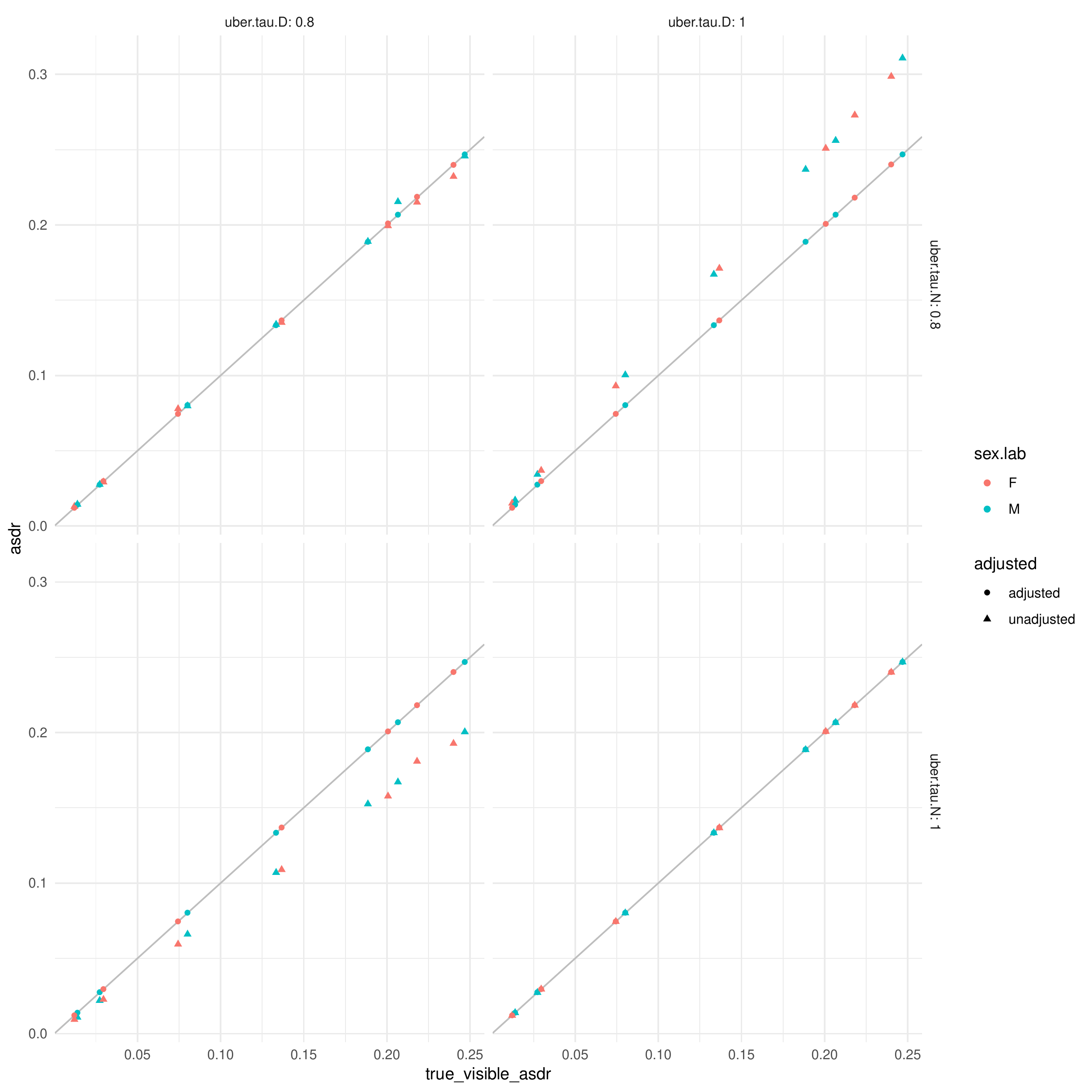}
\caption{True age-specific death rates (x axis) against adjusted and
unadjusted death rate individual visibility estimands (y axis). Male
death rates are in blue and females are in red. The four panels show
four different reporting scenarios. The adjusted estimands all agree
with the true underlying death rates, illustrating the correctness of
the individual sensitivity framework in
Equation~\ref{eq:ind-mult-ubersens} These results are for the individual
estimator that does not include the respondent.}\label{fig:check-ind}
}
\end{figure}

\hypertarget{comparing-the-estimators}{%
\subsubsection{Comparing the
estimators}\label{comparing-the-estimators}}

Next, we investigate the performance of the individual and aggregate
visibility estimators. In order to evaluate each estimator, we
calculated the relative mean square error across the \(K=1000\)
simulated surveys for each scenario. The relative mean square error of a
set of estimates \(\widehat{\vec{M}}^V\) is

\[
\text{rel-MSE}(\widehat{\vec{M}}^V) 
= \frac{\sum_{i=1}^K (\widehat{M}^V_i - M^V)^2}{K~(M^V)^2}.
\] The relative mean square error can be decomposed into the sum of the
squared relative bias: \[
\text{rel-bias}^2(\widehat{\vec{M}}^V) = \left( \frac{\sum_{i = 1}^K \widehat{M}^V_i - M^V}{K~M^V}\right)^2,
\] and the relative variance: \[
\text{rel-var}(\widehat{\vec{M}}^V) = \frac{\sum_{i = 1}^K (\widehat{M}^V_i - \bar{M})^2}{(M^V)^2},
\]

where \(\bar{M} = K^{-1} \sum_i \widehat{M}^V_i\) is the average of the
\(K\) estimates. Thus, we have

\[
\text{rel-MSE}(\widehat{\vec{M}}^V) = 
\text{rel-bias}^2(\widehat{\vec{M}}^V) +
\text{rel-var}(\widehat{\vec{M}}^V).
\]

Figure~\ref{fig:sim-mse-pr} compares the two approaches when the
sampling fraction (i.e., the sample size relative to the population
size) is 0.05 and reporting is perfect. The figure shows the relative
MSE as well as its decomposition into relative squared bias and relative
variance. Several observations can be made about
Figure~\ref{fig:sim-mse-pr}. First, the magnitude of the relative MSE is
comparable for the individual and aggregate visibility estimators,
though it is slightly higher using the individual visibility estimator
for the oldest age females. Second, the decomposition of the MSE into
squared bias and variance reveals that the individual visibility
estimator is essentially unbiased -- all of its MSE comes from variance.
The aggregate visibility estimator, on the other hand, is slightly
biased, even in this favorable simulation setup.

Figure~\ref{fig:sim-mse-ir} compares the two estimators when the
sampling fraction is 0.05 but reporting is imperfect; in this scenario,
the true positive rate for reports about deaths is 0.8, meaning that
some deaths are omitted from reports. Several observations can be made
about Figure~\ref{fig:sim-mse-ir}. First, compared to the perfect
reporting results shown in Figure~\ref{fig:sim-mse-pr}, the level of
relative MSE is higher in Figure~\ref{fig:sim-mse-ir}, presumably due to
the reporting errors. Second, again due to the reporting errors, we see
now that both the aggregate and individual visibility estimators are
biased. Generally, the MSE are again similar between the two approaches,
now with some cells producing slightly lower MSE for the individual
visibility estimator.

Importantly, both Figure~\ref{fig:sim-mse-pr} and
Figure~\ref{fig:sim-mse-ir} show the results of scenarios in which there
is, by design, no relationship between visibility and mortality. In
situations where such a relationship exists, we would expect the
individual visibility estimator to tend to perform better than the
aggregate visibility estimator, since the individual visibility
estimator avoids having to make the assumption that the visibility of
deaths is equal to the visibility of exposure.

Future research could extend this simulation analysis to better
understand the tradeoffs between these two approaches; our results here
(i) confirm that our analytical results are correct; and (ii) illustrate
that, in relatively favorable situations, the error from the aggregate
and individual visibility estimators is roughly comparable. A deeper
understanding of the difference in the two estimators is an important
topic for future research.

\begin{figure}
\hypertarget{fig:sim-mse-pr}{%
\centering
\includegraphics{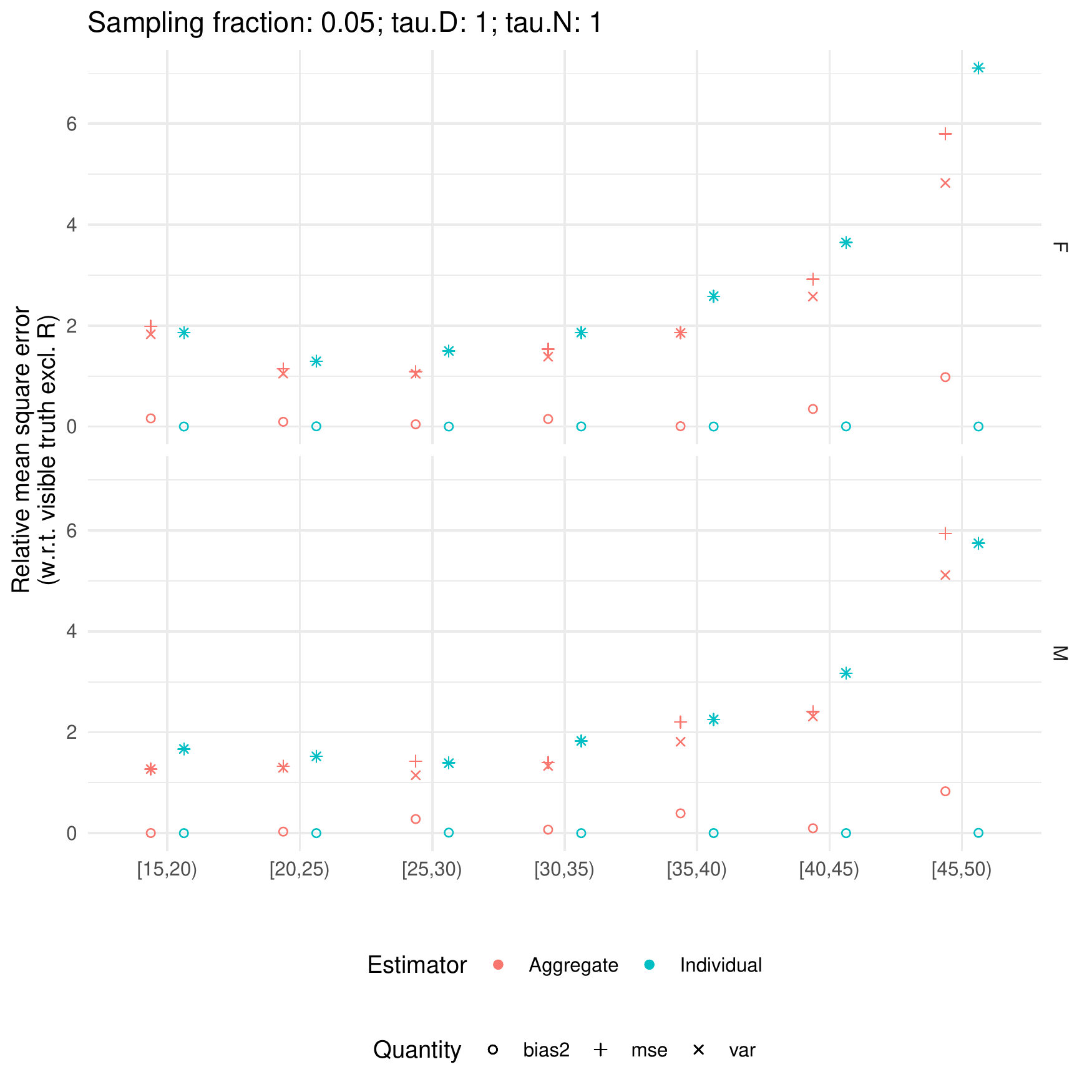}
\caption{Mean square error, squared bias, and variance for the
individual and aggregate visibility estimators when the sampling
fraction is 0.05.}\label{fig:sim-mse-pr}
}
\end{figure}

\begin{figure}
\hypertarget{fig:sim-mse-ir}{%
\centering
\includegraphics{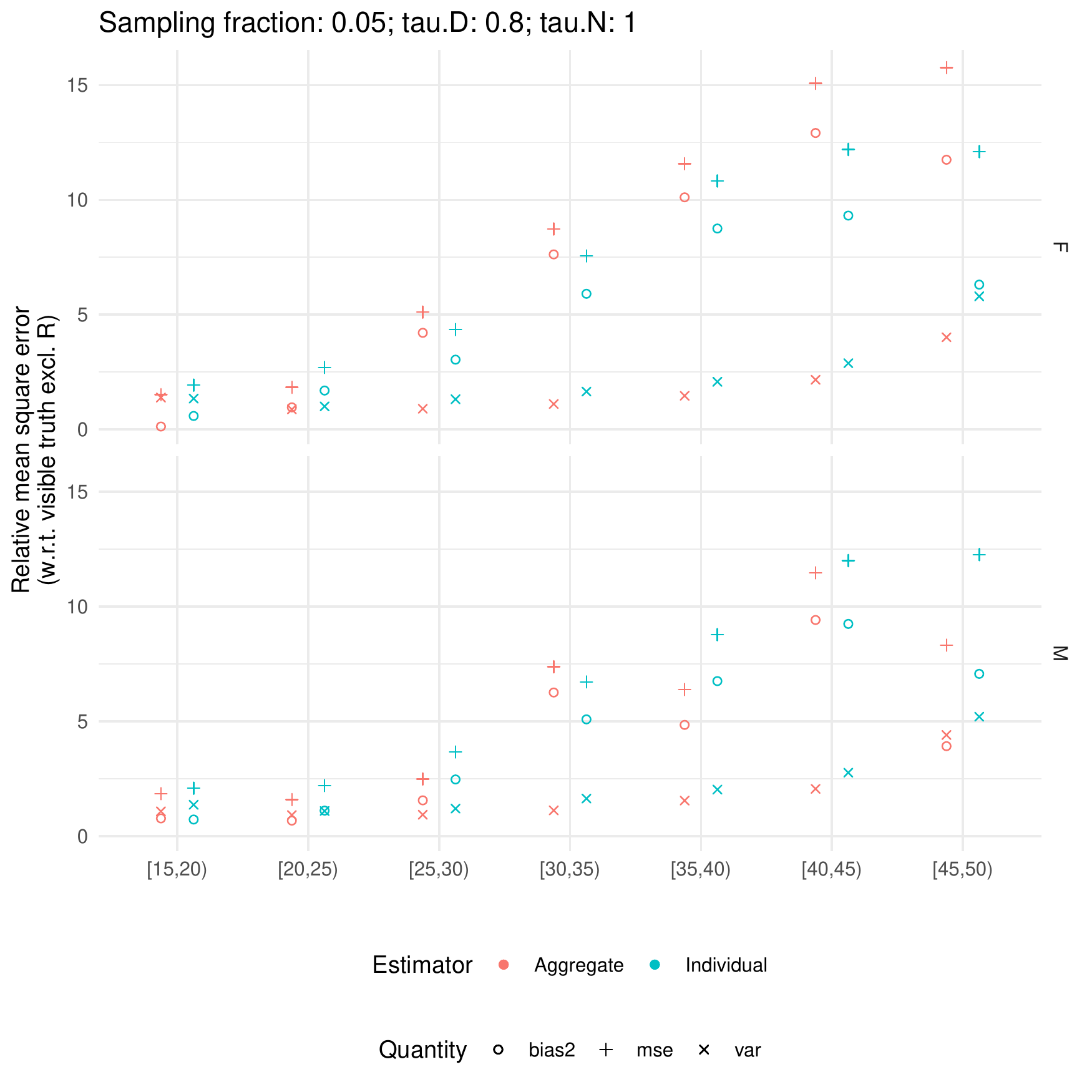}
\caption{Mean square error, squared bias, and variance for the
individual and aggregate visibility estimators when the sampling
fraction is 0.05.}\label{fig:sim-mse-ir}
}
\end{figure}

\end{document}